\documentclass{lmcs} 
\pdfoutput=1

\usepackage{lastpage}
\lmcsdoi{18}{4}{8}
\lmcsheading{}{\pageref{LastPage}}{}{}%
{Aug.~25,~2021}{Nov.~23,~2022}{}

\usepackage[utf8]{inputenc}
\DeclareUnicodeCharacter{03BB}{\ensuremath{\lambda}} 
\DeclareUnicodeCharacter{03B4}{\ensuremath{\delta}} 


\usepackage{hyperref}
\usepackage{amsfonts,amsmath,mathrsfs}
\usepackage{graphicx}
\usepackage{marvosym}
\usepackage{url}
\usetikzlibrary{cd,arrows}
\tikzset{
	bend angle=45,
	auto,
	baseline=(current  bounding  box.center),
}

\let\temp\varphi
\let\varphi\phi
\let\phi\temp

\newcommand{\sect}{\mathcal S} 

\let\temp\S
\let\S\sect
\let\sect\temp

\newcommand{\bb}[1]{[\![ #1 ]\!]}
\newcommand{\zero}{\varepsilon}
\newcommand{\Cppo}{\mathbb{CPPO}}

\newcommand{\C}{\mathbb{C}}
\newcommand{\Set}{\mathbb S \mathrm{et}}
\newcommand{\Rel}{\mathbb R \mathrm{el}}
\newcommand{\Rp}{\mathbb R^+}
\newcommand{\EM}[1]{\mathbb {EM}(#1)}  \newcommand{\Kl}[1]{\mathbb K \mathrm l(#1)} \newcommand{\F}[1]{F^{#1}} \newcommand{\U}[1]{U^{#1}} \newcommand{\FK}[1]{F_{#1}} \newcommand{\UK}[1]{U_{#1}} \newcommand{\op}[1]{{#1}^\mathit{o}} 

\newcommand{\N}{\mathbb N}
\newcommand{\natset}[1]{\underline{#1}} \newcommand{\Bool}{\mathbb B \mathrm{ool}}

\renewcommand{\P}{\mathcal P} \newcommand{\Pf}{\mathcal P_f} \newcommand{\pow}{\mathcal P}
\newcommand{\etaP}{\eta^\P}
\newcommand{\muP}{\mu^\P}
\newcommand{\Graph}[1]{\Gamma(#1)}
\newcommand{\pre}[1]{{#1}^{-1}}

\newcommand{\mon}[1]{\mathcal S} \newcommand{\etaS}{\eta^\S}
\newcommand{\muS}{\mu^\S}
\DeclareMathOperator{\supp}{\mathit{supp}}

 \newcommand{\convpow}{\mathcal P_c}
\newcommand{\ConvPow}[2]{\mathcal P_c^{#2} #1}  \newcommand{\convclos}[2]{\overline{#1}^{#2}}
\newcommand{\choice}[1]{\mathfrak c(#1)} 

\newcommand{\convpowS}{{\mathcal P_c \mathcal S}} \newcommand{\convpowfS}{{\mathcal P_{fc} \mathcal S}} 

\newcommand{\id}[1]{\mathit{id}_{#1}}

\newcommand{\Alg}{\mathbb{A}\mathrm{lg}} \newcommand{\CSL}{\mathcal{CSL}} \newcommand{\LSM}{\mathcal{LSM}}
\newcommand{\SL}{\mathcal{SL}}
\newcommand{\D}{\mathcal{D}}
\newcommand{\ext}[1]{\mathrm{ext}(#1)}

\newcommand{\Sup}{\bigsqcup}

\newcommand{\card}[1]{|#1|}

\newcommand{\Sharp}[1]{{#1}^\sharp}

\newcommand{\tvect}[2]{\ensuremath{\Bigl(\negthinspace\begin{smallmatrix}#1\\#2\end{smallmatrix}\Bigr)}}

\newcommand{\myeq}[1]{\stackrel{(#1)}{=}}

\def\eg{{\em e.g.}}
\def\cf{{\em cf.}}
\def\ie{{\em i.e.\ }}

\begin{document}
		
		\title{Convexity via Weak Distributive Laws}
		\titlecomment{An extended abstract of this paper has appeared in the proceedings of the FoSSaCS 2021 conference, see~\cite{bonchi_combining_2021}.}
		
		\author[F.~Bonchi]{Filippo Bonchi\lmcsorcid{0000-0002-3433-723X}}
		\author[A.~Santamaria]{Alessio Santamaria\lmcsorcid{0000-0001-7683-5221}}
		
		\address{Dipartimento di Informatica, Università di Pisa, Largo B.~Pontecorvo 3, 56127 Pisa, Italy}	
		\email{filippo.bonchi@unipi.it, alessio.santamaria@di.unipi.it}  
		
		
		
		
		
		
		\begin{abstract}
			\noindent We study the canonical weak distributive law $\delta$ of the powerset monad over the semimodule monad for a certain class of semirings containing, in particular, positive semifields. For this subclass we characterise $\delta$ as a convex closure in the free semimodule of a set. Using the abstract theory of weak distributive laws, we compose the powerset and the semimodule monads via $\delta$, obtaining the monad of convex subsets of the free semimodule.
		\end{abstract}
		
		\maketitle
		
		\section*{Introduction}
		
		Convexity plays an important role in several areas of computer science ranging from linear programming to program analysis. In the former, the space of solutions of an optimisation problem is convex and this property is crucially exploited to find optimal solutions. In program analysis, the space of stores possibly encountered during the execution of a program is usually \emph{not} convex but, to make their analysis feasible, one can exploit abstract interpretation~\cite{10.1145/512950.512973} and consider an abstracted program dealing with convex spaces. Convexity also appears in the study of concurrent probabilistic systems that usually behave both in probabilistic and in non-deterministic ways. Since the very beginning of this field, the behaviour of such systems is often assumed to be convex (see \eg~convex bisimilarity~\cite{DBLP:conf/concur/SegalaL94,DBLP:journals/corr/Mio13}) and much emphasis has been given in finding algebras to properly deal with both non-determinism and probability~\cite{bandini2001axiomatizations,mislove2000nondeterminism,DBLP:journals/entcs/MisloveOW04,varacca2003probability}.
		
		\begin{table}[t]
			\centering
			\scalebox{0.88}{
				\begin{tabular}{|c|}
					\hline 
					\begin{tabular}{rcl||rclrcl}
						Semilattice & & & $S$-semimodule \\
						\hline
						$(x \sqcup y) \sqcup z$ & $\myeq{A}$ & $ x\sqcup(y \sqcup z)$  & $(x+y)+z$ &$\myeq{A}$& $x+(y+z)$ & $\lambda \cdot 0$&$\myeq{Z_s}$& $0$   \\
						$x\sqcup y $&$\myeq{C}$&$ y \sqcup x$     & $x+y $&$\myeq{C}$&$ y+x$     & $0_S \cdot x $&$\myeq{An_s}$&$ 0$  \\
						$x\sqcup \bot $&$\myeq{U}$&$ x$  & $x+0 $&$\myeq{U}$&$ x$  &   $1_S \cdot x $&$\myeq{U_s}$&$ x$ \\
						$x\sqcup x $&$\myeq{I}$&$ x$ & && & $(\lambda \mu) \cdot x $&$\myeq{\mathit{Co}_s}$&$ \lambda \cdot (\mu \cdot x)$  \\
						&& & && & $\lambda \cdot (x+y)$&$\myeq{D_s1}$& $\lambda \cdot x + \lambda \cdot y$ \\
						&& & && &  $ (\lambda +_S \mu) \cdot x$ &$\myeq{D_s2}$& $\lambda \cdot x + \mu \cdot x $ \\
						\hline
						\hline
					\end{tabular} \\
					\begin{tabular}{rclrcl}
						Distributivity \\
						\hline
						$\lambda \cdot \bot$ & $\myeq{D1}$ & $\bot \text{ for } \lambda\neq 0_S$ \quad & \quad $\lambda \cdot (x \sqcup y) $&$\myeq{D3}$&$(\lambda \cdot x) \sqcup (\lambda \cdot y)$ \\
						$x + \bot $&$\myeq{D2}$&$\bot$ \quad & \quad $x + (y \sqcup z) $&$\myeq{D4}$&$ (x + y) \sqcup (x + z)$ \\
					\end{tabular}\\
					\hline 
				\end{tabular}
			}
			\caption{The sets of axioms $E_{\SL}$ for semilattices (left), $E_{\LSM}$ for $S$-semimodules (right) and $E_{\D'}$ for their distributivity (bottom).}\label{tab:axiomsinitial}
		\end{table}
		Such algebras inspire the main intuition of this work, that is convexity arises somewhat naturally when combining non-deterministic and quantitative features. To give a first motivation of this intuition, we show now that convexity can be derived as a law from a few axioms. Consider Table~\ref{tab:axiomsinitial}: the axioms on the top left, forming the algebraic theory of semilattices, are the rules for non-determinism; those on the top right form the algebraic theory of semimodules for some fixed semiring $S$ that are often exploited to model different quantitative aspects; the four axioms at the bottom of the table \emph{distribute} the structure of semilattice over the semimodule one.
		As illustrated in Table~\ref{tab:derivation}, from these axioms one can derive the following law: \[
		x \sqcup y =  x \sqcup y \sqcup (\alpha \cdot x + \beta \cdot y) \text{ for all } \alpha, \beta \in S \text{ such that } \alpha+\beta=1\text{.}
		\]
		The above equality expresses convexity. Take indeed $S$ to be $\mathbb{R}^+$, the semiring of non negative reals, and consider the algebra freely generated by the two-element set $\{x_1,x_2\}$, \ie the algebra of terms with variables in $\{x_1,x_2\}$ built out of $\bot, \sqcup, 0, +, 1, \lambda \cdot$ and quotiented by the axioms in Table \ref{tab:axiomsinitial}.
		Because of the distributivity axioms, the general term is of the form $ (\lambda_1 \cdot x_1 + \mu_1 \cdot x_2) \sqcup \dots \sqcup (\lambda_n \cdot x_1 + \mu_n \cdot x_2)$. Now, we can interpret geometrically these terms as follows. The structure of semimodule allows to express single (two-dimensional) vectors, for instance the terms $x_1$, $x_2$, $x_1 + 3 \cdot x_2$ denote  the vectors $\tvect{1}{0}$, $\tvect{0}{1}$, $\tvect{1}{3}$; the join of two vectors is interpreted by forming a set which, as expressed by the above equality, must contain the two vectors and also all the vectors between them. For instance the term $x_1 \sqcup x_2$ is \emph{not} just the set of vectors $\{\tvect{1}{0}, \tvect{0}{1}\}$ but its convex closure $\{\alpha \cdot \tvect{1}{0} + \beta \cdot \tvect{0}{1}\mid \alpha +\beta=1\}$, \ie the segment whose vertices are $x_1$ and $x_2$. Similarly, $x_1 \sqcup x_2 \sqcup (x_1 + 3 \cdot x_2)$ is the set of all points enclosed within (and including) the triangle with vertices $\tvect{1}{0}$, $\tvect{0}{1}$, $\tvect{1}{3}$. More generally, as we will show at the end of this paper in Example~\ref{ex:segments and polytopes}, these terms are in one-to-one correspondence with convex polygons in the $\mathbb{R}^+$-semimodule $(\mathbb{R}^+)^2$.
		
		\begin{table}\small{
				\begin{tabular}{|c|}
					\hline \\
					$\begin{array}{rcl} x \sqcup y \sqcup (\alpha \cdot x + \beta \cdot y) &
						\myeq{*} &
						x\sqcup (\alpha \cdot y +\beta \cdot x) \sqcup (\alpha \cdot x + \beta \cdot y) \sqcup y \sqcup (\alpha \cdot x+\beta \cdot y)\\
						& \myeq{I} & x\sqcup (\alpha \cdot y +\beta \cdot x) \sqcup (\alpha \cdot x + \beta \cdot y) \sqcup y \\
						& \myeq{*} & x \sqcup y
					\end{array}$ \\
					\hline 
					\hline \\
					$\begin{array}{rcl}  (x \sqcup y) &\myeq{U_s} &
						1\cdot (x\sqcup y) \\
						& \myeq{\alpha+\beta=1}  &
						(\alpha+\beta)\cdot (x \sqcup y) 
						\\ &\myeq{D_s2}&
						\alpha\cdot (x \sqcup y ) + \beta \cdot (x \sqcup y) \\ 
						& \myeq{D3} &
						(\alpha \cdot x \sqcup \alpha \cdot y) + (\beta \cdot x \sqcup \beta \cdot y) \\ 
						&\myeq{D4} &
						((\alpha \cdot x \sqcup \alpha \cdot y) + \beta \cdot x) \sqcup ((\alpha \cdot x \sqcup \alpha \cdot y) + \beta \cdot y) \\ 
						&\myeq{D4}  &
						(\alpha \cdot x+\beta \cdot x) \sqcup (\alpha \cdot y +\beta \cdot x) \sqcup (\alpha \cdot x + \beta \cdot y) \sqcup (\alpha \cdot y + \beta \cdot y) \\ 
						& \myeq{D_s2} &
						(\alpha + \beta) \cdot x \sqcup (\alpha \cdot y +\beta \cdot x) \sqcup (\alpha \cdot x + \beta \cdot y) \sqcup (\alpha + \beta) \cdot y \\ 
						&\myeq{\alpha+\beta=1} &
						1 \cdot x \sqcup (\alpha \cdot y +\beta \cdot x) \sqcup (\alpha \cdot x + \beta \cdot y) \sqcup 1 \cdot y \\ 
						& \myeq{D_s1} &
						x\sqcup (\alpha \cdot y +\beta \cdot x) \sqcup (\alpha \cdot x + \beta \cdot y) \sqcup y
					\end{array}$\\
					\hline
			\end{tabular}}
			\caption{Derivation for the convexity law ($\alpha$ and $\beta$ are assumed to be such that $\alpha+\beta=1$). The steps marked with $(*)$ are justified by the derivation in the bottom.}\label{tab:derivation}
		\end{table}

		The aim of this paper is to make this intuition formal by using the language of category theory. Rather than considering the aforementioned algebraic theories~\cite{DBLP:journals/tcs/HylandPP06}, we work with the equivalent notion of \emph{monad}. Monads are nowadays common tools in theoretical computer science, as they embody different notions of computations~\cite{DBLP:journals/iandc/Moggi91}, like nondeterminism, side effects and exceptions. 
		
		We show that combining  somehow the finite powerset monad, $\P_f$, corresponding to the algebraic theory of semilattices, together with the monad $\S$ of semimodules, one obtains the monad of convex subsets that are \emph{finitely generated} (roughly, that arise as convex combinations of finitely many elements). Moreover, the algebraic theory corresponding to the composed monad is exactly the one depicted in Table \ref{tab:axiomsinitial}. The monad for arbitrary--not necessarily finitely generated--convex subsets can instead be achieved by combining the full powerset monad $\P$ with $\S$; the corresponding algebraic theory can be easily obtained from Table~\ref{tab:axiomsinitial} by replacing finite joins with arbitrary ones: see Table~\ref{table:otheraxioms}.

		We perform the combination of these monads in a principled way, as in the last decades many works proposed different methods to compose monads and related notions~\cite{DBLP:journals/tcs/HylandPP06,DBLP:journals/mscs/VaraccaW06,bohm_weak_2010,lack2004composing,cheng2011distributive}. Indeed, the standard approach of composing monads by means of \emph{distributive laws}~\cite{beck_distributive_1969} turned out to be somehow unsatisfactory. On the one hand, distributive laws do not exist in many relevant cases: see~\cite{DBLP:journals/entcs/KlinS18,zwart_no-go_2019} for some no-go theorems; on the other hand, proving their existence is error-prone: see~\cite{DBLP:journals/entcs/KlinS18} for a list of results that were mistakenly assuming the existence of a distributive law of the powerset monad over itself. 
		Nevertheless, some sort of weakening of the notion of distributive law--\eg~distributive laws of functors over monads~\cite{DBLP:journals/tcs/Klin11}--proved to be ubiquitous in computer science: they are GSOS specifications~\cite{turi1997towards}, they are sound coinductive up-to techniques~\cite{DBLP:conf/csl/BonchiPPR14} and complete abstract domains~\cite{DBLP:conf/lics/BonchiGGP18}. 
		In this paper we will exploit \emph{weak distributive laws} in the sense of~\cite{garner_vietoris_2020} that have been recently shown to be successful in composing the monads for nondeterminism and probability~\cite{goy_combining_2020}.
		
		\begin{table} 
			\centering
			\small
			\begin{tabular}{|c|}
				\hline 
				\begin{tabular}{c}
					Complete Semilattice \\
					\hline \\
					$\begin{array}{rcl}
						\Sup\limits_{i \in \{ 0\}} x_i  & \myeq{Si} &  x_0 \\[.5em]
						\Sup\limits_{j \in J}  x_j  & \myeq{CI} &  \Sup\limits_{i \in I}  x_{f(i)}  \text{ for all } f \colon I \to J \text{ surjective}  \\[.5em]
						\Sup\limits_{i \in I} x_i  & \myeq{AU} & \Sup\limits_{j \in J} \, \, \Sup\limits_{i \in f^{-1}\{j\}}  x_i   \text{ for all } f \colon I \to J 
					\end{array}$
				\end{tabular} \\
				\hline 
				\hline
				\begin{tabular}{c}
					Distributivity \\
					\hline \\
					$\begin{array}{rcl}
						\lambda \cdot \Sup\limits_{i \in I} x_i & \myeq{D1} & \Sup\limits_{i \in I} \lambda \cdot x_i \quad \text{for }\lambda \ne 0 \\  
						\Sup\limits_{i \in I} x_i + \Sup\limits_{j \in J} y_j & \myeq{D2} & \Sup\limits_{(i,j) \in I \times J} x_i + y_j
					\end{array}$
				\end{tabular} \\
				\hline
			\end{tabular}
			\caption{The sets of axioms $E_{\CSL}$ for complete semilattices (top) and $E_{\D}$ for distributivity  with $S$-semimodules (bottom). Note that the  axiom $(CI)$ of $E_{\CSL}$ generalises the usual idempotency and commutativity properties of finitary $\sqcup$, while $(AU)$ generalises associativity and unity of $\Sup\limits_{i \in \emptyset} x_i =\bot$. }\label{table:otheraxioms}
		\end{table}

		We start with the full powerset monad $\P$, rather than the finite one $\P_f$. In this way we can reuse several results in~\cite{clementino_monads_2014} that provide necessary and sufficient conditions on the semiring $S$ for the existence of a canonical weak~\cite{garner_vietoris_2020} distributive law  $\delta \colon \S \P \to \P \S$. 
		Our first contribution (Theorem~\ref{thm:delta for positive refinable semifields}) consists in showing that $\delta$ has a convenient alternative characterisation whenever the underlying semiring is a \emph{positive semifield}, a condition that is met for example by the semiring of non-negative reals $\mathbb{R}^+$ and by the semiring of Booleans.
		
		This characterisation allows us to give a handy definition of the \emph{canonical weak lifting} of $\P$ over $\EM{\S}$ (Theorem~\ref{thm:weaklift}) and to observe that this lifting is \emph{almost} the same as the monad $\mathcal{C} \colon \EM{\S} \to \EM{\S} $ defined by Jacobs in~\cite{jacobs_coalgebraic_2008} (Remark~\ref{remarkJacobs}): the only difference is the absence in $\mathcal{C}$ of the empty subset. This divergence becomes crucial when considering the composed monads, named $\mathcal{C}\mathcal M \colon \Set \to \Set$ in~\cite{jacobs_coalgebraic_2008} and $\convpowS\colon \Set \to \Set$ in this paper: the latter maps a set $X$ into the set of convex subsets of $\S X$, while the former additionally requires the subsets to be non-empty. It turns out that while the Kleisli category $\Kl{\mathcal{C}\mathcal M}$ is not $\mathbb{CPPO}$-enriched, a crucial condition in the setting of~\cite{hasuo_generic_2006} (on which~\cite{jacobs_coalgebraic_2008} is based) to define trace semantics within a coalgebraic framework, $\Kl{\convpowS}$ indeed is (Theorem~\ref{thm:Kleisli is CPPO enriched}).

		Composing monads by means of weak distributive laws is rewarding in many respects: here we exploit the fact that algebras for the composed monad $\convpowS$ coincide with $\delta$-algebras, namely algebras for both $\P$ and $\S$ satisfying a certain pentagonal law. One can extract from this law some distributivity axioms (in the bottom of Table \ref{table:otheraxioms}) that, together with the axioms for semimodules (algebras for the monad $\S$) and those for complete semilattices (algebras for the monad $\P$), provide an algebraic theory presenting the monad $\convpowS$ (Theorem \ref{thm:pres}).
		
		We conclude with the finite powerset monad $\P_f$. By replacing, in the above theory, complete semilattices with semilattices (algebras for the monad $\P_f$) one obtains a theory  presenting the monad $\convpowfS$ of finitely generated convex subsets (Theorem~\ref{thm:convpowfS is presented by (Sigma', E')}), which is formally defined as a restriction of the canonical $\convpowS$. The corresponding algebraic theory is exactly the one in Table~\ref{tab:axiomsinitial}.
		
		\medskip
		
		This article is an extend version of~\cite{bonchi_combining_2021}. 
		Beyond providing all proofs which could not be published in~\cite{bonchi_combining_2021} because of space limitations, we illustrate here a detailed example (Example~\ref{ex:working example delta=delta' part 2}) that gives a sketch of the proof of the fact that the canonical weak distributive law $\delta \colon \S\pow \to \pow\S$ can be expressed in terms of a convex closure when $S$ is a positive semifield and, in Remark~\ref{rem:convexity still present even if not positive semifield}, we explain that even when $S$ only satisfies the minimal assumptions of Theorem~\ref{thm:existence of delta}, convexity still plays a crucial role.
		Moreover, we discuss  at the end of Section~\ref{sec:finite joins and finitely generated convex sets}  how a natural candidate for a weak distributive law of $\Pf$ over $\S$ fails to work, thus justifying our approach via the full powerset monad.
		Finally, in Remark~\ref{rem:Kleisli vs Barr extensions} and Appendix~\ref{app:another weak distributive law}, we show the existence of a weak extension of the powerset functor to $\Rel$ that is \emph{not} locally monotone, and thus is not an extension in the sense of Barr~\cite{barr_relational_1970}. Barr extensions, if they exist, are unique, but this is not true for arbitrary extensions to $\Rel$. Since weak distributive laws are in bijective correspondence with these latter, simpler extensions, this shows that one could have various weak distributive laws of $\pow$ over a monad $T$.

		\subsection*{Notation}We assume the reader to be familiar with monads and their maps.	Given a monad $(M,\eta^M,\mu^M)$ on $\C$, $\EM M$ and $\Kl M$ denote, respectively, the Eilenberg-Moore category and the Kleisli category of $M$. The latter is defined as the category whose objects are the same as $\C$ and a morphism $f \colon X \to Y$ in $\Kl M$ is a morphism $f \colon X \to M(Y)$ in $\C$. We write $\U M \colon \EM M \to \C$ and $\UK M \colon \Kl M \to \C$ for the canonical forgetful functors, and $\F M \colon \C \to \EM M$, $\FK M \colon \C \to \Kl M$ for their respective left adjoints. Recall, in particular, that $\F M (X) = (X,\mu^M_X)$ and, for $f \colon X \to Y$, $\F M(f) = M(f)$. 
		
		If $f \colon X \to Y$ is in $\Set$, we denote its graph as $\Gamma(f) = \{(x,f(x)) \mid x \in X\}$, which is a morphism from $X$ to $Y$ in the category $\Rel$ of sets and relations. If $R \subseteq X \times Y$, then we write $\op R$ for its opposite relation given by $\op R=\{(y,x) \mid (x,y) \in R \}$.

		Given $n$ a natural number, we denote by $\natset n$ the set $\{ 1, \dots, n\}$. \section{(Weak) Distributive laws}
		
		Given two monads $S$ and $T$ on a category $\C$, is there a way to compose them to form a new monad $ST$ on $\C$? This question was answered by Beck~\cite{beck_distributive_1969} and his theory of \emph{distributive laws}, which are natural transformations $\delta \colon TS \to ST$ satisfying four axioms and that provide a canonical way to endow the composite functor $ST$ with a monad structure. We begin by recalling the classic definition. In the following, let $(T,\eta^T,\mu^T)$ and $(S,\eta^S,\mu^S)$ be two monads on a category $\C$.
		
		\begin{defi}\label{def:distributive law}
			A \emph{distributive law} of the monad $S$ over the monad $T$ is a natural transformation $\delta\colon TS \to ST$ such that the following diagrams commute.
			\begin{equation}\label{eqn:def distributive law}
				\begin{tikzcd}[ampersand replacement=\&,column sep=2em]
					TSS \ar[r,"\delta S"] \ar[d,"T\mu^S"'] \& STS \ar[r,"S\delta"] \& SST \ar[d,"\mu^S T"] \& \& TTS \ar[r,"T\delta"] \ar[d,"\mu^T S"'] \& TST \ar[r,"\delta T"] \& STT \ar[d,"S\mu^T"] \\
					TS \ar[rr,"\delta"] \& \& ST \& \& 	TS \ar[rr,"\delta"] \& \& ST \\
					\& T \ar[dl,"T\eta^S"'] \ar[dr,"\eta^S T"] \& \& \& \& S \ar[dl,"\eta^T S"'] \ar[dr,"S \eta^T"] \\
					TS \ar[rr,"\delta"] \& \& ST \& \& 	TS \ar[rr,"\delta"] \& \& ST
				\end{tikzcd}	
			\end{equation}
		\end{defi}
		
		One important result of Beck's theory is the correspondence between distributive laws, liftings to Eilenberg-Moore algebras and extensions to Kleisli categories, in the following sense.
		
		\begin{defi}\label{def:liftings and extensions}
			Let $(T,\eta^T,\mu^T)$ and $(S,\eta^S,\mu^S)$ be two monads on a category $\C$.
			\begin{enumerate}
				\item A \emph{lifting} of the functor $S$ to $\EM T$ is a functor $\tilde S$ on $\EM T$ such that
				\[
				\begin{tikzcd}
					\EM T \ar[r,"\tilde S"] \ar[d,"\U T"'] & \EM T  \ar[d,"\U T"]\\
					\C \ar[r,"S"]  & \C 
				\end{tikzcd}
				\]
				commutes.
				\item\label{extension functor to Kleisli} An \emph{extension} of the functor $T$ to $\Kl S$ is a functor $\tilde T$ on $\Kl S$ such that
				\[
				\begin{tikzcd}
					\C \ar[r,"T"] \ar[d,"\FK S"'] & \C \ar[d,"\FK S"] \\
					\Kl S \ar[r,"\tilde T"] & \Kl S
				\end{tikzcd}
				\]
				commutes.
			\end{enumerate}
			Let also:
			\begin{itemize}[left=\parindent]
				\item $S_1$, $S_2$, $T_1$ and $T_2$ be functors on $\C$,
				\item $\tilde{S_1}$ and $\tilde{S_2}$ be liftings of $S_1$ and $S_2$ to $\EM T$,
				\item $\tilde{T_1}$ and $\tilde{T_2}$ be extensions of $T_1$ and $T_2$ to $\Kl S$,
				\item $\phi \colon S_1 \to S_2$ and $\psi \colon T_1 \to T_2$ be natural transformations. 
			\end{itemize}  
			\begin{enumerate}[resume]
				\item A \emph{lifting} of $\phi$ to $\EM T$ is a natural transformation $\tilde \phi \colon \tilde{S_1} \to \tilde{S_2}$ such that $\U T \tilde \phi = \phi \U T$. 
				\item An \emph{extension} of $\psi$ to $\Kl S$ is a natural transformation $\tilde \psi \colon \tilde{T_1} \to \tilde{T_2}$ such that $\tilde\psi \FK S = \FK S \psi$.
				\item\label{extension natural transformation to Kleisli} A \emph{lifting}  of the monad $(S,\eta^S,\mu^S)$ to $\EM T$ is a monad $(\tilde S, \eta^{\tilde S}, \mu^{\tilde S} )$ such that $\tilde S$, $\eta^{\tilde S}$ and $\mu^{\tilde S}$ are lifting of, respectively, $S$, $\eta^S$ and $\mu^S$.
				\item An \emph{extension} of the monad $(T,\eta^T,\mu^T)$ to $\Kl S$ is a monad $(\tilde T, \eta^{\tilde T}, \mu^{\tilde T})$ such that $\tilde T$, $\eta^{\tilde T}$ and $\mu^{\tilde T}$ are extensions of, respectively, $T$, $\eta^T$ and $\mu^T$.
			\end{enumerate}
		\end{defi}
		
		\begin{thmC}[\cite{beck_distributive_1969}]\label{thm:correspondence distributive laws and extensions}
			Let $S$ and $T$ be two monads on $\C$. There are bijections between distributive laws of the form $TS \to ST$, liftings of the monad $S$ to $\EM T$ and extensions of the monad $T$ to $\Kl S$.
		\end{thmC}
		\begin{proof}[Proof (sketch)]
			Given a distributive law $\delta \colon TS \to ST$, the corresponding lifting $\tilde S$ of $S$ to $\EM T$ is given on objects as 
			\[
			\tilde S (X,a \colon TX \to X) = (SX,
			\begin{tikzcd}[cramped,sep=small]
				TSX \ar[r,"\delta_X"] & STX \ar[r,"Sa"] & SX
			\end{tikzcd}
			)
			\]
			and on morphisms as $\tilde S(f) = S(f)$; while $\eta^{\tilde S} = \eta^S$ and $\mu^{\tilde S} = \mu^S$. Vice versa, given a lifting of the monad $(S,\eta^S,\mu^S)$ to $\EM T$ with $\tilde S(X,a) = (S X, \sigma_{X,a})$, the corresponding distributive law $\delta \colon TS \to ST$ has components:
			\begin{equation}\label{eqn:Garner 3.1}
				\delta_X = \begin{tikzcd}
					TSX \ar[r,"TS\eta^T_X"] &  TSTX \ar[r,"\sigma_{TX,\mu^T_X}"] & STX.
				\end{tikzcd}
			\end{equation}
			Next, given a distributive law $\delta \colon TS \to ST$, the corresponding extension $\tilde T$ is given on objects by $\tilde T(X)=T(X)$ and on a morphism $f \colon X \to S(Y)$ in $\C$ as
			\[
			\tilde T(f) = \begin{tikzcd}T(X) \ar[r,"T(f)"] & TS(Y) \ar[r,"\delta_Y"] & ST(Y),
			\end{tikzcd}
			\]
			while the $X$-th component of $\eta^{\tilde T}$ and $\mu^{\tilde T}$ are, respectively,
			\[
			\begin{tikzcd}X \ar[r,"\eta^T_X"] & T(X) \ar[r,"\eta^S_{T(X)}"] & ST(X)
			\end{tikzcd}
			\quad \text{and} \quad
			\begin{tikzcd}TT(X) \ar[r,"\mu^T_X"] & T(X) \ar[r,"\eta^S_{T(X)}"] & ST(X).
			\end{tikzcd}
			\]
			Vice versa, given an extension $(\tilde T, \eta^{\tilde T}, \mu^{\tilde T})$ of $T$ to $\Kl S$, consider the identity map in $\C$ $\id{S(X)}^\C$, which is a morphism $\id{S(X)}^\C \colon S(X) \to X$ in $\Kl S$. The image along $\tilde T$ provides a morphism $\tilde T (\id{S(X)}^\C) \colon  TS(X) \to T(X)$ in $\Kl S$ and therefore a morphism $TS(X) \to ST(X)$ in $\C$, which is the $X$-th component of the desired distributive law $\delta$:
			\[
			\delta_X = \tilde T(\id{S(X)}^\C). \qedhere
			\]
		\end{proof}
		
		Given a distributive law $\delta \colon TS \to ST$, and therefore a lifting of $S$ to $\EM T$, the \emph{composite} monad on $\C$ is the monad induced by the composite adjunction:
		\[
		\begin{tikzcd}[bend angle=30,column sep=3em]
			\C \ar[r,bend left,pos=0.6,"\F T"{name=I}] & \EM {T} \ar[l,bend left,pos=0.4,"\U T"{name=K}] \ar[r,bend left,"\F {\tilde S}"{name=F}] & \EM {\tilde S}. \ar[l,bend left,"\U{\tilde S}"{name=U}]
			\ar[from=I,to=K,draw=none,"\bot" description]
			\ar[from=F,to=U,draw=none,"\bot" description]
		\end{tikzcd}
		\]
		Its underlying functor is $ST$, the unit is $\eta^S \eta^T$ and the multiplication is $(\mu^S \mu^T) \circ (S \delta T)$. 
		
		It turns out, however, that distributive laws often cannot exist in many practical examples, see for instance~\cite{zwart_no-go_2019}. Böhm~\cite{bohm_weak_2010} and Street~\cite{street_weak_2009} have studied various weaker notions of distributive law, obtained by relaxing in different ways the axioms required by the definition of distributive law. We are interested in one in particular, as done in~\cite{garner_vietoris_2020}, obtained by dropping the requirement of commutativity of the bottom-right triangle of (\ref{eqn:def distributive law}), which concerns $\eta^T$.

		\begin{defi}\label{def:weak distributive law}
			Let $(T,\eta^T,\mu^T)$ and $(S,\eta^S,\mu^S)$ be two monads on $\C$. A \emph{weak distributive law} of $S$ over $T$ is a natural transformation $\delta \colon TS \to ST$ such that the following diagrams commute:
			\[
			\begin{tikzcd}[ampersand replacement=\&,column sep=2em]
				TSS \ar[r,"\delta S"] \ar[d,"T\mu^S"'] \& STS \ar[r,"S\delta"] \& SST \ar[d,"\mu^S T"] \& \& TTS \ar[r,"T\delta"] \ar[d,"\mu^T S"'] \& TST \ar[r,"\delta T"] \& TTS \ar[d,"S\mu^T"] \\
				TS \ar[rr,"\delta"] \& \& ST \& \& 	TS \ar[rr,"\delta"] \& \& ST \\
				\& \& \& T \ar[dl,"T\eta^S"'] \ar[dr,"\eta^S T"]  \\
				\& \& TS \ar[rr,"\delta"] \& \& ST
			\end{tikzcd}	
			\]
		\end{defi}
		
		There are now corresponding notions of weak lifting and weak extensions, whose definitions we recall next.
		
		\begin{defi}\label{def:weak lifting}
			Let $(S,\eta^S,\mu^S)$ and $(T,\eta^T,\mu^T)$ be two monads on a category $\C$.
			\begin{enumerate}
				\item A \emph{weak extension} of the monad $T$ to $\Kl{S}$ is a functor $\tilde T \colon \Kl S \to \Kl S$ together with a natural transformation $\mu^{\tilde T} \colon \tilde T \tilde T \to \tilde T$ such that $\tilde T$ and $\mu^{\tilde T}$ are extensions, respectively, of $T$ and $\mu^T$.
				
				\item A \emph{weak lifting} of the monad $S$ to $\EM T$ consists of a monad $(\tilde S,\eta^{\tilde S},\mu^{\tilde S})$ on $\EM{T}$ and two natural transformations
				\[
				\begin{tikzcd}
					\U T \tilde S \ar[r,"\iota"] & S \U T \ar[r,"\pi"] & \U T \tilde S
				\end{tikzcd}
				\]
				such that $\pi \iota = id_{\U T \tilde S}$ and such that the following diagrams commute:
				\begin{equation}\label{eqn:weak lifting diagrams iota}
					\begin{tikzcd}
						\U T \tilde S \tilde S \ar[r,"\iota \tilde S"] \ar[d,"\U T \mu^{\tilde S}"'] & S \U T \tilde S \ar[r,"S \iota"] & S S \U T \ar[d,"\mu^S \U T"] \\
						\U T \tilde S \ar[rr,"\iota"] & & S \U T
					\end{tikzcd}
					\quad
					\begin{tikzcd}
						& \U T \ar[dl,"\U T \eta^{\tilde S}"'] \ar[dr,"\eta^S \U T"] \\
						\U T \tilde S \ar[rr,"\iota"] & & S \U T
					\end{tikzcd}
				\end{equation}
				\begin{equation}\label{eqn:weak lifting diagrams pi}
					\begin{tikzcd}
						S S \U T \ar[r,"S\pi"] \ar[d,"\mu^S \U T"'] & S \U T \tilde S \ar[r,"\pi \tilde S"] & \U T \tilde S \tilde S \ar[d,"\U T \mu^{\tilde S}"] \\
						S \U T \ar[rr,"\pi"] & & \U T \tilde S
					\end{tikzcd}
					\quad
					\begin{tikzcd}
						& \U T \ar[dl,"\eta^S \U T"'] \ar[dr,"\U T \eta^{\tilde S}"] \\
						S \U T \ar[rr,"\pi"] & & \U T \tilde S
					\end{tikzcd}
				\end{equation}
			\end{enumerate}
		\end{defi}
		
		Again we have a bijective correspondence between weak distributive laws and weak extensions, whereas the relationship between weak distributive laws and weak liftings is more subtle. We state it in the following Theorem, which is proved in~\cite{garner_vietoris_2020}.
		
		\begin{thmC}[{\cite[Propositions 11, 13]{garner_vietoris_2020}}]\label{thm:correspondence weak distributive law with weak extensions-liftings}
			Let $(T,\eta^T,\mu^T)$ and $(S,\eta^S,\mu^S)$ be two monads on $\C$. Then:
			\begin{enumerate}[(i),leftmargin=\parindent]
				\item There is a bijective correspondence between weak distributive laws of type $TS \to ST$ and weak extensions of $T$ to $\Kl S$.
				\item If idempotents split in $\C$, then there is a bijective correspondence between weak distributive laws of type $TS \to ST$ and weak lifting of $S$ to $\EM T$.
			\end{enumerate}
		\end{thmC}
		\begin{proof}
			(i) is proved exactly as in Theorem~\ref{thm:correspondence distributive laws and extensions}.
			
			For (ii): let $\delta \colon TS \to ST$ be a weak distributive law. Then the weak lifting of $S$ to $\EM T$ is obtained as follows.
			\begin{itemize}
				\item The action of the functor $\tilde S \colon \EM T \to \EM T$ on a $T$-algebra $(X,a)$ is given by the $T$-algebra $(Y,b)$ where:
				\begin{itemize}
					\item $Y$ is the result of splitting the idempotent $Sa \circ \delta_X \circ \eta^T_{SX} \colon SX \to SX$ in $\C$:
					\[
					\begin{tikzcd}
						SX \ar[r,"\pi_{(X,a)}"] & Y \ar[r,"\iota_{(X,a)}"] & SX,
					\end{tikzcd}
					\]
					\item $b$ is the composite:
					\[
					b= \begin{tikzcd}
						TY \ar[r,"T\iota_{(X,a)}"] & TSX \ar[r,"\delta_X"] & STX \ar[r,"Sa"] & SX \ar[r,"\pi_{(X,a)}"] & Y.
					\end{tikzcd}
					\]
				\end{itemize}
				\item Given $f \colon (X,a) \to (X',a')$ in $\EM T$ and
				\[
				\begin{tikzcd}
					SX' \ar[r,"\pi_{(X',a')}"] & Y' \ar[r,"\iota_{(X',a')}"] & SX'
				\end{tikzcd}
				\]
				a splitting of the idempotent $Sa' \circ \delta_{X'} \circ \eta^T_{S X'}$, the action of the functor $\tilde S$ on $f$ is the composite:
				\[
				\tilde S (f) = 
				\begin{tikzcd}
					Y \ar[r,"\iota_{(X,a)}"] & SX \ar[r,"Sf"] & SX' \ar[r,"\pi_{(X',a')}"] & Y'.
				\end{tikzcd}
				\]
				\item The $(X,a)$-th components of the two natural transformations $\iota$ and $\pi$ are given as above by splitting the idempotent $Sa \circ \delta_X \circ \eta^T_{SX}$.
				\item The unit $\eta^{\tilde S}$ is the unique map rendering commutative the triangles in (\ref{eqn:weak lifting diagrams iota}) and (\ref{eqn:weak lifting diagrams pi}).
				\item The multiplication $\mu^{\tilde S}$ is the unique map rendering commutative the rectangles in (\ref{eqn:weak lifting diagrams iota}) and (\ref{eqn:weak lifting diagrams pi}).
			\end{itemize}
			
			Vice versa, suppose we have a weak lifting of $S$ to $\EM T$. For each $T$-algebra $(X,a)$ with $\tilde S (X,a) = (Y,b)$, define the morphism $\sigma_{X,a}$ as 
			\[
			\sigma_{X,a} = 
			\begin{tikzcd}
				TSX \ar[r,"T\pi_{(X,a)}"] & TY \ar[r,"b"] & Y \ar[r,"\iota_{(X,a)}"] & SX
			\end{tikzcd}
			\]
			and then define $\delta \colon TS \to ST$ as the family of morphisms whose $X$-th component is~(\ref{eqn:Garner 3.1}). Then $\delta$ is a weak distributive law.
		\end{proof}
		
		As in the ``strong'' case, given a weak distributive law $\delta \colon TS \to ST$, and therefore a weak lifting of $S$ to $\EM T$ (when idempotents split in $\C$), we can define a ``composite'' monad $\widetilde{ST}$ as the monad arising from the composite adjunction
		\[
		\begin{tikzcd}[bend angle=30,column sep=3em]
			\C \ar[r,bend left,pos=0.6,"\F T"{name=I}] & \EM {T} \ar[l,bend left,pos=0.4,"\U T"{name=K}] \ar[r,bend left,"\F {\tilde S}"{name=F}] & \EM {\tilde S}, \ar[l,bend left,"\U{\tilde S}"{name=U}]
			\ar[from=I,to=K,draw=none,"\bot" description]
			\ar[from=F,to=U,draw=none,"\bot" description]
		\end{tikzcd}
		\]
		whose underlying functor this time is not $ST$, but is obtained by splitting the idempotent
		\[
		\begin{tikzcd}
			ST \ar[r,"\eta^T ST"] & TST \ar[r,"\delta T"] & STT \ar[r,"S\mu^T"] & ST.
		\end{tikzcd}
		\] 
		\section{The Powerset and Semimodule Monads}\label{sec:monad}
		Through this paper we will take as $S$  and $T$ the powerset and semimodule monads, respectively. In this section we recall their definitions as well as some useful known results.
		\subsection{The Monad \texorpdfstring{$\P$}{\textbf{\textit{P}}}.}Let us now consider, as $S$, the \emph{powerset} monad $(\pow,\eta^\P,\mu^\P)$, where $\etaP_X(x)=\{x\}$ and $\muP_X(\mathcal U) = \bigcup_{U \in \mathcal U} U$. Its algebras are precisely the complete semilattices and we have that $\Kl \P$ is isomorphic to the category $\Rel$ of sets and relations. Hence, giving a distributive law $T\P \to \P T$ is the same as giving an extension of $T$ to $\Rel$. 
		
		As $\Rel$ is enriched over the category of partial orders, and thus is a 2-category, it is quite natural to require $T$ to be a 2-functor, \ie to preserve the partial order:  if $R \subseteq S$, then $T(R) \subseteq T(S)$.  It is not difficult to prove (see, for example,~\cite[\sect 5.1]{bird_algebra_1997}) that any 2-endofunctor $F$ on $\Rel$ preserves functions and opposite relations, in the sense that if $f$ is a function then so is $F(f)$ and if $\op R$ is the opposite relation of $R$, then $F(\op R) = \op{F(R)}$. This and the fact that every relation $R \colon X \to Y$ with projections $\pi_X \colon R \to X$ and $\pi_Y \colon R \to Y$ can be written as a composite in $\Rel$:
		\[
		R = \Gamma(\pi_Y) \circ \op{\Gamma(\pi_X)}
		\]
		(recall that $\Gamma$ computes the graph of functions) imply that\[
		F(R) = F\bigl(\Gamma(\pi_Y) \circ \op{\Gamma(\pi_X)} \bigr) = F(\Gamma(\pi_Y)) \circ \op{F(\Gamma(\pi_X))}.
		\]
		Therefore, any functor $T \colon \Set \to \Set$ has at most one extension to a 2-functor $\tilde T$ on $\Rel$, which is necessarily defined as $T$ on objects and as 
		\[
		\tilde T(R) = \Graph{T(\pi_Y)} \circ \op{\Graph{T(\pi_X)}}
		\]
		on morphisms. $\tilde T$ so defined is called the \emph{Barr extension} of $T$, as it was discussed first in~\cite{barr_relational_1970}. The problem is that $\tilde T$ is not always functorial: for this to happen the notions of weakly cartesian functor and natural transformation are crucial, and we report them next. First, recall that a commutative diagram
		\[
		\begin{tikzcd}
			W \ar[r,"q"] \ar[d,"p"'] & Y \ar[d,"g"] \\
			X \ar[r,"f"] & Z
		\end{tikzcd}
		\]
		is a \emph{weak pullback} if for every pair of morphisms $(u \colon P \to X, \, v \colon P \to Y)$ there exists a (not necessarily unique) morphism $t \colon P \to W$ such that $p \circ t = u$ and $q \circ t = v$.
		\begin{defi}
			A functor $T \colon \Set \to \Set$ is said to be \emph{weakly cartesian} if and only if it preserves weak pullbacks. A natural transformation $\phi \colon F \to G$ is said to be \emph{weakly cartesian} if and only if its naturality squares are weak pullbacks.
		\end{defi}
		
		Kurz and Velebil~\cite{kurz_relation_2016} proved, using an original argument of Barr~\cite{barr_relational_1970}, that $\tilde T$ is functorial if and only if $T$ is weakly cartesian. Moreover, one can show using a similar technique that a natural transformation $\phi$ between weakly cartesian functors on $\Set$ has at most one extension to a 2-natural transformation, which exists precisely when $\phi$ is weakly cartesian, as mentioned in~\cite{garner_vietoris_2020}. These considerations are summarised in the following Theorem.
		
		\begin{thm}\label{thm:extension to Rel iff weakly cartesian}
			Any functor $T \colon \Set \to \Set$ has at most one extension to a 2-functor $\tilde T$ on $\Rel$; such an extension exists if and only if $T$ is weakly cartesian. $\tilde T$ is defined as $T$ on objects and, for every relation $R \colon X \to Y$ with projections $\pi_X \colon R \to X$ and $\pi_Y \colon R \to Y$, as
			\[
			\tilde T(R) = \Graph{T(\pi_Y)} \circ \op{\Graph{T(\pi_X)}}
			\]
			where $\Graph f$, for $f \colon A \to B$ a function, is its graph $\Graph f = \{ (a,f(a)) \mid a \in A \}$.
			
			Let $T_1$ and $T_2$ be weakly cartesian functors. Any natural transformation $\phi \colon T_1 \to T_2$ has at most one extension $\tilde\phi \colon \tilde{T_1} \to \tilde{T_2}$; such an extension exists if and only if $\phi$ is weakly cartesian. Its $X$-th component, for $X$ a set, is given by the graph of $\phi_X$:
			\[
			\tilde\phi_X = \Graph{\phi_X} = \{ (x,\phi_X(x)) \mid x \in T_1(X)  \}.
			\]
		\end{thm}

		The following result is therefore immediate. 
		\begin{propC}[{\cite[Corollary 16]{garner_vietoris_2020}}]\label{prop:(weak)distributive law with P iff weakly cartesian}
			For any monad $(T,\eta^T,\mu^T)$ on $\Set$:
			\begin{enumerate}
				\item If $T$, $\eta^T$ and $\mu^T$ are weakly cartesian, then there exists a distributive law of $\P$ over $T$.
				\item If $T$ and $\mu^T$ are weakly cartesian, then there exists a \emph{weak} distributive law of $\P$ over $T$.
			\end{enumerate}
		\end{propC}
		
		We follow~\cite{garner_vietoris_2020} and call the (weak) distributive law of Proposition~\ref{prop:(weak)distributive law with P iff weakly cartesian}, arising from the unique extension of $T$, $\mu^T$ and possibly $\eta^T$ given by Theorem~\ref{thm:extension to Rel iff weakly cartesian}, as the \emph{canonical} (weak) distributive law of $\pow$ over $T$.
		
		\begin{rem}[Kleisli extensions vs Barr extensions]\label{rem:Kleisli vs Barr extensions} Theorem~\ref{thm:extension to Rel iff weakly cartesian} states that any endofunctor $T$ on $\Set$ has at most one extension to a \emph{2-functor} on $\Rel$, seen as a poset-enriched category. However, this is asking for more than a Kleisli extension of $T$ to $\Kl{\pow}=\Rel$ in the sense of Definition~\ref{def:liftings and extensions}.\ref{extension functor to Kleisli}, since the latter does not require the extension to be locally monotone, but only to agree with $T$ on $\Set$. Kurz and Velebil~\cite{kurz_relation_2016} left an open question, in this regard: 
			\begin{center}
				Is there a Kleisli extension of a \emph{monad} $(T,\eta^T,\mu^T)$ to $\Rel$ that is not a Barr extension?
			\end{center}
			We can provide a partial answer to this question: there exists a \emph{weak} extension of a monad that is not a Barr extension. To illustrate this, we reuse from~\cite{bird_generic_1996} a non-monotone extension of a functor on $\Set$ to $\Rel$. Consider the powerset functor $\pow$ itself: define $E \colon \Rel \to \Rel$ as $E(X)=\pow (X)$ and, for $R \subseteq X \times Y$ a relation, $E(R)$ as the \emph{function} (notice, not just a relation!):
			\[
			\begin{tikzcd}[row sep=0em]
				\pow(X) \ar[r,"E(R)"] & \pow(Y) \\
				A \ar[r,|->] & \{ y \in Y \mid \exists a \in A \ldotp a R y \}
			\end{tikzcd}
			\]
			This is a somewhat trivial extension of $\pow$ to $\Rel$, as for all functions $f\colon X \to Y$, $E(f)$ has exactly the same definition of $\pow(f)$. It is clear, therefore, that $E$ and $\pow$ agree on $\Set$ and it is not difficult to see that $E$ is a functor. However, it is not locally monotone: the main reason being that $E(R)$ is a function, and a function in $\Rel$ is included in another if and only if they are equal. For a specific counterexample, take $X=\{0\}$, $Y=\{ 1, 2 \}$, $R=\{ (0,1) \}$, $S=\{ (0,1), \, (0,2) \}$. Then $ R \subseteq S$ but 
			\[
			\begin{tikzcd}[row sep=0em]
				\pow(X) \ar[r,"E(R)"] & \pow(Y) \\
				\emptyset \ar[r,|->] & \emptyset \\
				\{ 0 \} \ar[r,|->] & \{ 1 \}
			\end{tikzcd}
			\quad \ne \quad
			\begin{tikzcd}[row sep=0em]
				\pow(X) \ar[r,"E(S)"] & \pow(Y) \\
				\emptyset \ar[r,|->] & \emptyset \\
				\{ 0 \} \ar[r,|->] & \{ 1, 2 \}
			\end{tikzcd}
			\]
			It turns out that $E$ is not just an extension of the powerset functor, but it is actually a weak extension of the monad: the interested reader can find a detailed explanation in Appendix~\ref{app:another weak distributive law}.
		\end{rem}
		
		\subsection{The Monad \texorpdfstring{$\mon S$}{\textbf{\textit{S}}}.} Recall that a \emph{semiring} is a tuple $(S,+,\cdot,0,1)$ such that $(S,+,0)$ is a commutative monoid, $(S,\cdot,1)$ is a monoid, $\cdot$ distributes over $+$ and $0$ is an annihilating element for $\cdot$. In other words, a semiring is a ring where not every element has an additive inverse. Natural numbers $\N$ with the usual operations of addition and multiplication form a semiring. Similarly, integers, rationals and reals form semirings. Also the booleans $\Bool=\{0,1\}$ with $\vee$ and $\wedge$ acting as $+$ and $\cdot$, respectively, form a semiring. 
		
		Every semiring $S$ generates a \emph{semimodule} monad $\S$ on $\Set$ as follows. Given a set $X$, $\S(X) = \{ \phi \colon X \to S \mid \supp \phi \text{ finite} \}$, where $\supp \phi = \{ x \in X \mid \phi(x) \ne 0 \}$. For $f \colon X \to Y$, define for all $\phi\in \S (X)$
		\[
		\S(f)(\phi) = \Bigl( y \mapsto \sum_{x \in f^{-1}\{y\}} \phi(x) \Bigr) \colon Y \to S.
		\]
		This makes $\S$ a functor. The unit $\etaS_X \colon X \to \S(X)$ is given by $\etaS_X(x)=\Delta_x$, where $\Delta_x$ is the Dirac function centred in $x$, while the multiplication $\muS_X \colon \S^2(X) \to \S(X)$ is defined for all $\Psi\in \S^2(X)$ as
		\[
		\muS_X(\Psi) = \Bigl( x \mapsto \sum_{\phi \in \supp \Psi} \Psi(\phi) \cdot \phi(x)  \Bigr) \colon X \to S.
		\]
		An algebra for $\S$ is precisely a \emph{left-S-semimodule}, namely a set $X$ equipped with a binary operation $+$, an element $0$ and a unary operation $\lambda \cdot$ for each $\lambda\in S$, satisfying the equations in Table~\ref{tab:axiomsinitial}, right.
		Indeed, if $X$ carries a semimodule structure then one can define a map $a \colon \S X \to X$ as, for $\phi \in \mon S X$,
		\begin{equation}\label{eq:S-algebra associated to semimodule}
			a(\phi) = \sum_{x \in X} \phi(x) \cdot x
		\end{equation}
		where the above sum is finite because so is $\supp \phi$. Vice versa, if $(X,a)$ is an $\mon S$-algebra, then the corresponding left-semimodule structure on $X$ is obtained by defining for all $\lambda \in S$ and $x,y \in X$
		\begin{equation}\label{eqn:semimodule associated to S-algebra}    x+^a y = a (x \mapsto 1, y \mapsto 1), \qquad
			0^a = a(\zero), \qquad
			\lambda \cdot^a x = a (x\mapsto \lambda).
		\end{equation}
		Above and in the remainder of the paper, we write the list $(x_1\mapsto s_1, \dots, x_n \mapsto s_n)$ for the only function $\phi\colon X \to S$ with support $\{x_1,\dots , x_n\}$ mapping $x_i$ to $s_i$ and we write the empty list $\zero$ for the function constant to $0$. For instance, for $a=\muS_X \colon \mon S \mon S X \to \mon S X$, the left-semimodule structure is defined for all $\phi_1,\phi_2 \in \mon S X$ and $x\in X$ as
		\[(\phi_1+^{\muS_X} \phi_2)(x)=\phi_1(x)+\phi_2(x), \qquad 0^{\muS_X}(x)=0, \qquad (\lambda \cdot^{\muS_X} \phi_1 )(x)= \lambda \cdot \phi_1( x).\]
		
		Proposition~\ref{prop:(weak)distributive law with P iff weakly cartesian} tells us when a canonical (weak) distributive law of the form $T\P \to \P T$ exists for an arbitrary monad $T$ on $\Set$. Take then $T=\S$: when are the functor $\S$ and the natural transformations $\etaS$ and $\muS$ weakly cartesian? The answer has been given in~\cite{clementino_monads_2014} (see also~\cite{gumm_monoid-labeled_2001}), where a complete characterisation in purely algebraic properties for $S$ is provided. In Table~\ref{tab:semiring properties} we recall such properties.
		
		\begin{table}[t]
			\centering
			\begin{tabular}{|c|p{30em}|}
				\hline
				Positive     & $a+b=0 \implies a=0=b$ \\ \hline
				Semifield    & $a \ne 0 \implies \exists x \ldotp a \cdot x = x \cdot a = 1$ \\ \hline
				Refinable    & $a+b=c+d \implies \exists x,y,z,t \ldotp x+y=a,\, z+t=b, \, x+z=c, \,y+t = d$ \\ \hline
				(A) & $a+b=1 \implies a = 0$ or $b=0$ \\ \hline
				(B) & $a \cdot b = 0 \implies a=0 \text{ or } b=0$ \\ \hline
				(C) & $a+c=b+c \implies a = b$ \\ \hline
				(D) & $\forall a,b$. $\exists x \ldotp a+x = b$ or $b+x = a$\\ \hline
				(E) & $a+b = c \cdot d \implies \exists t \colon \{ (x,y) \in S^2 \mid x+y=d  \} \to S$ such that \newline
				$
				\sum\limits_{x+y=d} t(x,y) x = a, \quad \sum\limits_{x+y=d} t(x,y)y=b, \quad \sum\limits_{x+y=d} t(x,y)=c.
				$ \\ \hline
			\end{tabular}
			\caption{Definition of some properties of a semiring $S$. Here $a,b,c,d \in S$.}
			\label{tab:semiring properties}
		\end{table}
		
		\begin{thmC}[{\cite{clementino_monads_2014}}]\label{thm:S,etaS,muS weakly cartesian iff}
			Let $S$ be a semiring.
			\begin{enumerate}
				\item The functor $\S$ is weakly cartesian if and only if $S$ is positive and refinable.
				\item $\etaS$ is weakly cartesian if and only if $S$ enjoys \emph{(A)} in Table \ref{tab:semiring properties}.
				\item\label{condition:muS weakly cartesian iff} If $\S$ is weakly cartesian, then $\muS$ is weakly cartesian if and only if $S$ enjoys \emph{(B)} and \emph{(E)} in Table \ref{tab:semiring properties}.
			\end{enumerate}
		\end{thmC}
		
		\begin{rem}\label{rem:S positive refinable semifield then ok}
			In~\cite[Proposition 9.1]{clementino_monads_2014} it is proved that if $S$ enjoys (C) and (D), then $S$ is refinable; if $S$ is a positive semifield, then it enjoys (B) and (E). 
		\end{rem}
		
		In the next Proposition we prove that if $S$ is a positive semifield then it is also refinable, hence $\S$ and $\mu^\S$ are weakly cartesian.
		
		\begin{prop}\label{prop:positive semifield implies refinable}
			If $S$ is a positive semifield, then it is refinable.
		\end{prop}
		\begin{proof}
			Let $a$, $b$, $c$ and $d$ in $S$ be such that $a+b=c+d$. If $a+b=0$, then take $x=y=z=t=0$, otherwise take
			\[
			x=\frac{ac}{c+d}, \quad y = \frac{ad}{c+d}, \quad z=\frac{bc}{c+d}, \quad t=\frac{bd}{c+d}.
			\]
			Then $x+y=a,\, z+t=b, \, x+z=c, \,y+t = d$.
		\end{proof}
		
		\begin{exa}\label{example:distributive law for N}
			It is known that, for $\S=\N$,  a distributive law $\delta \colon \mon S \P \to \P \mon S$ exists. Indeed one can check  that all conditions of Theorem~\ref{thm:S,etaS,muS weakly cartesian iff} are satisfied, therefore we can apply Proposition~\ref{prop:(weak)distributive law with P iff weakly cartesian}.1. In this case, the monad $\mon S X$ is naturally isomorphic to the commutative monoid monad, which given a set $X$ returns 
			the collection of all finite \emph{multisets} of elements of $X$. 
			The law $\delta$ is well known (see \eg~\cite{garner_vietoris_2020,hyland_category_2006}): given a finite multiset $\langle A_1,\dots,A_n \rangle$ of subsets of $X$ in $\mon S \pow X$, where the $A_i$'s need not be distinct, it returns the set of multisets
			$\{ \langle a_1, \dots, a_n \rangle \mid a_i\in A_i \}$.
		\end{exa}
		
		Theorem~\ref{thm:S,etaS,muS weakly cartesian iff} together with Proposition~\ref{prop:(weak)distributive law with P iff weakly cartesian}.1 tell us that whenever the element $1$ of $S$ can be decomposed as a non-trivial sum we do not have a canonical distributive law $\delta \colon \mon S \P \to \P \mon S$. Semirings with this property abound, for example $\mathbb Q$, $\mathbb R$, $\Rp$ with the usual operations of sum an multiplication, as well as $\Bool$ (since $1 \vee 1 = 1$). These semirings are precisely those for which the notion of \emph{convex subset} of their left-semimodules is non-trivial. For the existence of the canonical \emph{weak} distributive law, however, this condition on $1_S$ is not required: convexity will indeed play a crucial role in the definition of the weak distributive law.
		
		\begin{defi}\label{def:convex closure SEMIMODULE}
			Let $S$ be a semiring, $X$ an $S$-left-semimodule and $A \subseteq X$. The \emph{convex closure} of $A$ is the set
			\[
			\convclos{A}{} =\left\{ \sum_{i=1}^n \lambda_i \cdot a_i \mid n \in \N,\, a_i \in A,\,  \sum_{i=1}^n \lambda_i =1 \right\} \subseteq X.
			\]
			The set $A$ is said to be \emph{convex} if and only if $A = \convclos A {}$.
		\end{defi}
		
		Recalling that the category of $S$-left-semimodules is isomorphic to $\EM{\mon S}$, we can use~\eqref{eq:S-algebra associated to semimodule} to translate Definition~\ref{def:convex closure SEMIMODULE} of convex subset of a semimodule into the following notion of convex subset of a $\mon S$-algebra $a\colon \mon S X \to X$.
		\begin{defi}\label{def:convex closure S-ALGEBRA}
			Let $S$ be a semiring, $(X,a) \in \EM{\mon S}$, $A \subseteq X$. The \emph{convex closure} of $A$ in $(X,a)$ is the set
			\[
			\convclos A a = \left\{ a(\phi) \mid \phi \in \mon S X, \supp \phi \subseteq A, \sum_{x \in X} \phi(x) =1 \right\}.
			\]
			$A$ is said to be \emph{convex} in $(X,a)$ if and only if $A = \convclos A a$. We denote by $\ConvPow X a$ the set of convex subsets of $X$ with respect to $a$.
		\end{defi}
		
		\begin{rem}
			Observe that $\emptyset$ is convex, because $\convclos \emptyset a = \emptyset$, since there is no $\phi \in \mon S X$ with empty support such that $\sum_{x \in X} \phi(x) = 1$.
		\end{rem}
		
		\begin{exa}\label{ex:convex subsets of N-semimodules}
			Suppose $S$ is such that $\etaS$ is weakly cartesian (equivalently (A) holds: $x+y=1 \implies x=0$ or $y=0$), for example $S = \N$, and let $(X,a) \in \EM{\mon S}$. A $\phi \in \mon S X$ such that $\sum_{x \in X} \phi(x)=1$ and $\supp \phi \subseteq A$ is a function that assigns $1$ to \emph{exactly one} element of $A$ and $0$ to all the other elements of $X$. These functions are precisely all the $\Delta_x$ for those elements $x\in A$. Since $a\colon \mon S X \to X$ is a structure map for an $\mon S$-algebra, it maps the function $\Delta_x$ into $x$. Therefore $\convclos A a =\{a(\Delta_x) \mid x\in A\} = \{x \mid x \in A\} = A$. Thus \emph{all}  $A\in \pow \mon S X$ are convex.
		\end{exa}
		
		\begin{exa}\label{ex:S=Bool, convex subsets are directed}
			When $S = \Bool$, we have that $\S$ is naturally isomorphic to $\Pf$, the finite powerset monad, whose algebras are idempotent commutative monoids or equivalently semilattices with a bottom element (Table~\ref{tab:axiomsinitial}, left). So, for $(X,a) \in \EM{\mon S}$, a $\phi \in \mon S X$ such that $\sum_{x \in X} \phi(x)=1$ and $\supp \phi \subseteq A$ is any finitely supported function from $X$ to $\Bool$ that assigns $1$ to at least one element of $A$. Intuitively, such a $\phi$ selects a non-empty finite subset of $A$, then  $a(\phi)$ takes the join of all the selected elements. Thus, $\convclos A a$ adds to $A$ all the possible joins of non-empty finite subsets of $A$: $A$ is convex if and only if it is closed under binary joins.
		\end{exa}
		\section{The Canonical Weak Distributive Law \texorpdfstring{$\delta \colon \S \P \to \P \S$}{of \emph{P} over \emph{S}}}\label{sec:the weak distributive law}
		
		Weak extensions of $\mon S$ to $\Kl\P = \Rel$ only consist of extensions of the functor $\mon S$ and of the multiplication $\muS$, for whose existence it is sufficient that $\S$ and $\muS$ be weakly cartesian. Hence for semirings $S$ satisfying the criteria given in Theorems~\ref{thm:S,etaS,muS weakly cartesian iff}.1 and~\ref{thm:S,etaS,muS weakly cartesian iff}.3 we have a canonical weak distributive law $\delta \colon \mon S \P \to \P \mon S$, arising from the Barr extension of the functor $\mon S$ to $\Rel$.
		
		\begin{thm}\label{thm:existence of delta} 
			Let $S$ be a positive, refinable semiring satisfying \emph{(B)} and \emph{(E)} in Table~\ref{tab:semiring properties}. Then there exists a canonical weak distributive law $\delta \colon \mon S \P \to \P \mon S$ defined for all sets $X$ and $\Phi \in \mon S \pow X$ as:
			\begin{equation}\label{eqn:def of delta}
				\delta_X (\Phi) = \Biggl\{ \phi \in \mon SX \mid \exists \psi \in \mon S (\ni_X) \ldotp \begin{cases}
					\forall A \in \pow X \ldotp \Phi(A) = \sum\limits_{x \in A} \psi(A,x) & (a)\\
					\forall x \in X \ldotp \phi(x) = \sum\limits_{A \ni x} \psi(A,x) & (b)
				\end{cases}  
				\Biggr\}
			\end{equation}
			where ${}\ni_X{}$ is the set $\{ (A,x) \in \pow X \times X \mid x \in A \}$.
		\end{thm}
		\begin{proof}
			We calculate $\delta$ by first analysing the canonical extension $\widetilde{\mon S}$ of $\mon S$ to $\Rel$ and then by following the proof of Theorem~\ref{thm:correspondence distributive laws and extensions}.
			
			Theorem~\ref{thm:extension to Rel iff weakly cartesian} gives us the formula to extend the functor $\mon S$ to $\Rel$: on objects it behaves like $\mon S$; if $R \subseteq A \times B$ is a relation and $\pi_A \colon R \to A$ and $\pi_B \colon R \to B$ are the two projections, we have that for all $\psi \in \S(R)$
			\[
			\mon S (\pi_A) (\psi) = \Bigl( a \mapsto \sum_{r \in \pi_A^{-1}\{a\}} \psi(r) \Bigr) = 
			\Bigl( a \mapsto \sum_{\substack{b \in B \\ a R b}} \psi(a,b) \Bigr)
			\]
			and similarly for $\mon S (\pi_B)(\psi)$. So, since $\widetilde{\mon S} (R)= \Graph{\mon S(\pi_B)} \circ \op{\Graph{\mon S(\pi_A)}}$, we have that
			\[
			\widetilde{\mon S}(R) = \Bigl\{ \Bigl( \bigl( a \mapsto \sum_{\substack{b \in B \\ a R b}} \psi(a,b)  \bigr) ,   \bigl( b \mapsto \sum_{\substack{a \in A \\ a R b}} \psi(a,b)  \bigr)   \Bigr) \mid \psi \in \S(R) \Bigr\}.
			\]
			
			We now use the proof of Theorem~\ref{thm:correspondence distributive laws and extensions} to calculate the canonical weak distributive law $\mon S \P \to \P \mon S$. Recall the identity-on-objects isomorphism of categories $F \colon \Kl\P \to \Rel$ where for $f \colon X \to \pow(Y)$ in $\Set$, $F(f)=\{(x,y) \mid y \in f(x)\} \subseteq X \times Y$ and for $R \subseteq X \times Y$, $\pre F (R) = ( x \mapsto \{y \in Y \mid x \mathrel R y\}) \colon X \to \pow(Y)$. Consider $\id{\pow X} \colon \pow X \to \pow X$ in $\Set$. This is a Kleisli map $\id{\pow X} \colon \pow X \to X$ in $\Kl \P$, which corresponds to the relation ${}\ni_X{} \colon \pow X \to X = \{ (A,x) \mid x \in A  \}$. Then we have
			\[
			\widetilde{\mon S}(\ni_X)=\Bigl\{ \Bigl( \bigl( A \mapsto \sum_{x \in A } \psi(A,x)  \bigr) ,   \bigl( x \mapsto \sum_{A \ni x} \psi(A,x)  \bigr)   \Bigr) \mid \psi \in \mon S (\ni_X) \Bigr\} \subseteq \S\P X \times \S X.
			\]
			Remember that $\widetilde{\mon S}$ and $\mon S$ coincide on objects. This relation, seen back as a Kleisli map $\widetilde{\mon S}\pow X \to \widetilde{\mon S} X = \S\P X \to \S X$, thus a function $\S\P X \to \P\S X$, gives us the definition of $\delta_X$:
			\begin{align*}
				\delta_X(\Phi) &= \pre F( \widetilde{\mon S}(\ni_X)) (\Phi) \\
				&= \{ \phi \in \S (X) \mid (\Phi,\phi) \in \widetilde \S (\ni_X) \}
			\end{align*}
			which coincides with~\eqref{eqn:def of delta}. By Theorem~\ref{thm:correspondence weak distributive law with weak extensions-liftings}, $\delta$ so defined is indeed a weak distributive law.	
		\end{proof}
		
		\begin{exa}\label{ex:working example delta=delta' part 1}
			Take $S=\Rp$ with the usual operations of sum and multiplication. Consider $X=\{x,y,z,a,b\}$, $A_1 =\{x,y\}$, $A_2=\{y,z\}$ and $A_3=\{a,b\}$. Let $\Phi \in \mon S (\pow X)$ be defined as
			\[\Phi = ( A_1 \mapsto 5, \quad A_2 \mapsto 9, \quad A_3 \mapsto 13)\]
			and $\Phi(A)=0$ for all other sets $A\subseteq X$, so $\supp \Phi = \{A_1,A_2,A_3\}$. In order to find an element $\phi \in \delta_X(\Phi)$, we can first take a $\psi\in \mon S (\ni_X)$ satisfying  condition (a) in \eqref{eqn:def of delta} and then compute the $\phi\in \mon S X$ using condition (b).
			
			Among the $\psi \in \mon S (\ni_X)$, consider for instance the following:
			\[
			\psi = \left(  
			\begin{array}{rclrclrcl}
				(A_1,x)&\mapsto& 2 \quad & (A_2,y)&\mapsto& 4 \quad (A_3,a)&\mapsto& 6  \\
				(A_1,y)&\mapsto& 3 \quad & (A_2,z)&\mapsto& 5 \quad (A_3,b)&\mapsto& 7  & 
			\end{array}
			\right).
			\]
			Since $\Phi (A_1) = \psi(A_1,x) + \psi(A_1,y)$, $\Phi(A_2)=\psi(A_2,y) + \psi(A_2,z)$ and $\Phi(A_3)=\psi(A_3,a) + \psi(A_3, b)$, we have that $\psi$ satisfies condition (a) in \eqref{eqn:def of delta}. Condition (b) forces $\phi$ to be the following:
			\[\phi = ( x \mapsto 2, \quad y \mapsto 3+4, \quad z \mapsto 5, \quad a \mapsto 6, \quad b\mapsto 7).\]
		\end{exa}
		
		\begin{rem}
			If $S$ enjoys (A) in Table~\ref{tab:semiring properties}, then the transformation $\delta$ given in~(\ref{eqn:def of delta}) is actually a distributive law, and for $S=\N$ we recover the well-known $\delta$ of Example~\ref{example:distributive law for N}. Example~\ref{ex:working example delta=delta' part 1} can be repeated with $S=\N$: then $\Phi$ is the multiset where the set $A_1$ occurs five times, $A_2$ nine times and $A_3$ thirteen times. The elements of $\delta_X(\Phi)$ are all those multisets containing one element per copy of $A_1$, $A_2$ and $A_3$ in $\supp\Phi$. The $\phi$ provided indeed contains five elements of $A_1$ (two copies of $x$ and three of $y$), nine elements of $A_2$ (four copies of $y$ and five of $z$), thirteen elements of $A_3$ (six copies of $a$ and seven of $b$).
		\end{rem}
		
		As Example~\ref{ex:working example delta=delta' part 1} shows, each element $\phi$ of $\delta_X(\Phi)$ is determined by a function $\psi$ choosing for each set $A \in \supp \Phi$ a finite number of elements $x^A_1,\dots,x^A_m$ in $A$ and $s^A_1,\dots,s^A_m$ in $S$ in such a way that $\sum_{j=1}^m s^A_j = \Phi(A)$. The function $\phi$ maps each $x^A_j$ to $s^A_j$ if the sets in $\supp \Phi$ are \emph{disjoint}; if however there are $x^A_{j}$ and $x^B_k$ such that $x^A_j=x^B_k$ (like $y$ in Example~\ref{ex:working example delta=delta' part 1}), then $x^A_j$ is mapped to $s^A_j + s^B_k$. 
		
		Among those $\psi$'s, there are some special, \emph{minimal} ones as it were, that choose for each $A$ in $\supp \Phi$ exactly \emph{one} element of $A$, and assign to it $\Phi(A)$. The induced $\phi$ in $\delta_X(\Phi)$ can be described as $\phi(x)=\sum_{A\in u^{-1}\{x\}} \Phi(A)$ (equivalently $\phi = \mon{S}(u)(\Phi)$\footnote{More precisely, we should write $\S(u)(\Phi')$ where $\Phi'$ is the restriction of $\Phi$ to $\supp \Phi$.}) where $u\colon \supp \Phi \to X$ is a function selecting an element of $A$ for each $A\in \supp \Phi$ (that is $u(A)\in A$). We denote the set of such $\phi$'s by $\choice \Phi$.
		\begin{equation}\label{eqn:c(Phi)}
			\choice \Phi =\{ \mon{S}(u)(\Phi) \mid u \colon \supp \Phi \to X \text{ such that } \forall A \in \supp \Phi \ldotp u(A)\in A \}
		\end{equation}
		\begin{exa}
			Take $X$, $A_1$ and $A_2$ as in Example~\ref{ex:working example delta=delta' part 1}, but a different, smaller, $\Phi\in  \mon S (\pow X)$ defined as
			$\Phi = ( A_1 \mapsto 1, \quad A_2 \mapsto 2)$.
			There are only four functions $u\colon \supp \Phi \to X$ such that $u(A)\in A$ and thus only four  functions $\phi$ in $\choice \Phi$:
			\[\begin{array}{c|l}
				u_1=(A_1 \mapsto x, \quad A_2 \mapsto y) \quad & \quad  \phi_1 = (x \mapsto 1, \; y \mapsto 2)   \\
				u_2=(A_1 \mapsto x, \quad A_2 \mapsto z) \quad & \quad  \phi_2 = (x \mapsto 1, \; z \mapsto 2)   \\
				u_3=(A_1 \mapsto y, \quad A_2 \mapsto y) \quad & \quad  \phi_3 = ( y \mapsto 3)   \\
				u_4=(A_1 \mapsto y, \quad A_2 \mapsto z) \quad & \quad  \phi_4 = (y \mapsto 1, \; z \mapsto 2)  
			\end{array}\]
			Observe that the function $\phi= (x \mapsto 1, y\mapsto 1, z\mapsto 1)$ belongs to $\delta_X(\Phi)$ but not to $\choice \Phi$. Nevertheless $\phi$ can be retrieved as the convex combination $\frac{1}{2}\cdot \phi_1 + \frac{1}{2} \cdot \phi_2$.
		\end{exa}
		
		Our key result of this section, Theorem~\ref{thm:delta for positive refinable semifields} below, states that every $\phi \in \delta_X(\Phi)$ can be written as a convex combination (performed in the $\S$-algebra $(\S X,\muS_X)$) of functions in $\choice\Phi$, at least when $S$ is a positive semifield, which by Remark~\ref{rem:S positive refinable semifield then ok} and Proposition~\ref{prop:positive semifield implies refinable} satisfies all the conditions that make (\ref{eqn:def of delta}) a weak distributive law. The proof is rather laborious and will be the subject of the rest of this section.
		\begin{thm}\label{thm:delta for positive refinable semifields}
			Let $S$ be a positive semifield. Then for all sets $X$ and $\Phi \in \S \P X$
			\begin{equation}\label{eqn:delta for positive refinable semifields}
				\delta_X(\Phi)=\left\{ \muS_X(\Psi) \mid \Psi \in \mon S^2 X\ldotp \sum\limits_{\phi \in \mon S X} \Psi(\phi)=1, \, \supp \Psi \subseteq \choice\Phi \right\} = \convclos{\choice \Phi}{\muS_X}.
			\end{equation}
		\end{thm}
		
		\begin{rem}\label{rem:nice delta does not hold for minimal semirings}
			If we drop the hypothesis of semifield and only have the minimal assumptions of Theorem~\ref{thm:existence of delta}, then (\ref{eqn:delta for positive refinable semifields}) does not hold any more: $S=\N$ is a counterexample. Indeed, in this case every subset of $\S X$ is convex with respect to $\muS_X$ (see Example~\ref{ex:convex subsets of N-semimodules}), therefore we would have $\delta_X(\Phi)=\choice\Phi$, which is false: the function $\phi$ of Example~\ref{ex:working example delta=delta' part 1} is an example of an element in $\delta_X(\Phi) \setminus \choice\Phi$.
		\end{rem}
		
		\newcommand{\macroA}{\mathcal{A}}
		\begin{rem}\label{remarkKlinRot}
			When $S = \Bool$ (which is a positive semifield), the monad $\S$ coincides with the monad $\P_f$. The function $\choice{\cdot}$ in~\eqref{eqn:c(Phi)} can then be described as
			\begin{equation*}
				\choice \macroA =\{ \P_f(u)(\macroA) \mid u \colon \macroA \to X \text{ such that } \forall A \in \macroA \ldotp u(A)\in A \}
			\end{equation*}
			for all $\macroA \in \P_f \P X $. It is worth remarking that this is the transformation $\chi$ appearing in Example 9 of~\cite{DBLP:conf/fossacs/KlinR15} (which is in turn equivalent to the one in Example 2.4.7 of~\cite{manes2007monad}). This transformation was erroneously supposed to be a distributive law, as it fails to be natural (see~\cite{DBLP:journals/entcs/KlinS18}). However, by taking its convex closure, as displayed in~\eqref{eqn:delta for positive refinable semifields}, one can turn it into a \emph{weak} distributive law.
		\end{rem}
		
		The difficult inclusion of~\eqref{eqn:delta for positive refinable semifields} will be the left-to-right: before giving the proof, we shall present a ``working example'' which shows what sort of techniques one has to do to find a convex linear combination $\Psi$ in $\mon S^2 X$ of elements in $\choice\Phi$ for a given $\phi \in \delta_X(\Phi)$ in such a way that $\muS_X(\Psi)=\phi$. In particular we will use the possibility of performing divisions in the semiring, whence the required hypothesis for $S$ to be a semifield in Theorem~\ref{thm:delta for positive refinable semifields}.
		
		\begin{exa}\label{ex:working example delta=delta' part 2}
			Let us continue Example~\ref{ex:working example delta=delta' part 1}. We have shown that
			\[\phi = ( x \mapsto 2, \quad y \mapsto 3+4, \quad z \mapsto 5, \quad a \mapsto 6, \quad b\mapsto 7)\]
			is an element of $\delta_X(\Phi)$ determined by
			\[
			\psi = \left(  
			\begin{array}{rclrclrcl}
				(A_1,x)&\mapsto& 2 \quad & (A_2,y)&\mapsto& 4 \quad (A_3,a)&\mapsto& 6  \\
				(A_1,y)&\mapsto& 3 \quad & (A_2,z)&\mapsto& 5 \quad (A_3,b)&\mapsto& 7  & 
			\end{array}
			\right) \in \S(\ni_X).
			\]
			We want to write $\phi$ as a convex linear combination in the $\S$-algebra $(\S X, \muS_X)$ of elements in $\choice\Phi$, which means we have to find some functions $\chi \in \choice\Phi$ and some numbers $\Psi(\chi)$ for all these $\chi$'s such that their sum is $1$ and $\muS_X(\Psi)=\phi$. Now, to give a $\chi \in \choice\Phi$ we have to choose one element in each $A_i$ (which will have to be one of those that $\psi$ above has already ``selected''): we have a choice of $x$ or $y$ for $A_1$, $y$ or $z$ for $A_2$ and $a$ or $b$ for $A_3$. With the aim of finding a general strategy that does not rely on the peculiarities of this example where, for instance, it happens that the number of elements for each $A_i$ we can choose from is the same, we shall consider \emph{all} the possible combinations of choices, and define one $\chi$ for each. The number of these combinations is $2 \cdot 2 \cdot 2$, that is the cardinality of the set $\natset 2 \times \natset 2 \times \natset 2$, whose elements are indeed triples of indexes $1$, $2$. So, write
			\begin{align*}
				x^1_1 &= x & x^2_1 &= y & x^3_1 &= a \\
				x^1_2 &= y & x^2_2 &= z & x^3_2 &= b
			\end{align*}
			Now, for each $w=(w_1,w_2,w_3) \in \natset 2 ^3$, we define $\chi_w$ as the function choosing the $w_i$-th element in $A_i$ among those given by $\psi$, namely $x^i_{w_i}$; all these $\chi_w$ will always associate to the $w_i$-th element of $A_i$ the number $\Phi(A_i)$. For example, for $w^1=(1,1,1)$, we have 
			\[
			\chi_{w^1} = \left(  
			\begin{tikzcd}[row sep=0em,ampersand replacement=\&]
				x^1_1 \ar[r,|->] \& \Phi(A_1) \\
				x^2_1 \ar[r,|->] \& \Phi(A_2) \\
				x^3_1 \ar[r,|->] \& \Phi(A_3)
			\end{tikzcd}
			\right)
			\]
			while for $w^4=(1,2,2)$ we have
			\[
			\chi_{w^4} = \left(  
			\begin{tikzcd}[row sep=0em,ampersand replacement=\&]
				x^1_1 \ar[r,|->] \& \Phi(A_1) \\
				x^2_2 \ar[r,|->] \& \Phi(A_2) \\
				x^3_2 \ar[r,|->] \& \Phi(A_3)
			\end{tikzcd}
			\right)
			\]
			and so on. All in all, the eight functions $\chi_w$'s are:
			\begin{align*}
				\chi_{w^1} = \left(  
				\begin{tikzcd}[row sep=0em,ampersand replacement=\&]
					x \ar[r,|->] \& 5 \\
					y \ar[r,|->] \& 9 \\
					a \ar[r,|->] \& 13
				\end{tikzcd}
				\right)
				& &
				\chi_{w^2} = \left(  
				\begin{tikzcd}[row sep=0em,ampersand replacement=\&]
					x \ar[r,|->] \& 5 \\
					y \ar[r,|->] \& 9 \\
					b \ar[r,|->] \& 13
				\end{tikzcd}
				\right) 
				\displaybreak[0] \\
				\chi_{w^3} = \left(  
				\begin{tikzcd}[row sep=0em,ampersand replacement=\&]
					x \ar[r,|->] \& 5 \\
					z \ar[r,|->] \& 9 \\
					a \ar[r,|->] \& 13
				\end{tikzcd}
				\right)
				& &
				\chi_{w^4} = \left(  
				\begin{tikzcd}[row sep=0em,ampersand replacement=\&]
					x \ar[r,|->] \& 5 \\
					z \ar[r,|->] \& 9 \\
					b \ar[r,|->] \& 13
				\end{tikzcd}
				\right)
				\displaybreak[0] \\
				\chi_{w^5} = \left(  
				\begin{tikzcd}[row sep=0em,ampersand replacement=\&]
					y \ar[r,|->] \& 5 \\
					y \ar[r,|->] \& 9 \\
					a \ar[r,|->] \& 13
				\end{tikzcd}
				\right)
				& &
				\chi_{w^6} = \left(  
				\begin{tikzcd}[row sep=0em,ampersand replacement=\&]
					y \ar[r,|->] \& 5 \\
					y \ar[r,|->] \& 9 \\
					b \ar[r,|->] \& 13
				\end{tikzcd}
				\right)
				\displaybreak[0] \\
				\chi_{w^7} = \left(  
				\begin{tikzcd}[row sep=0em,ampersand replacement=\&]
					y \ar[r,|->] \& 5 \\
					z \ar[r,|->] \& 9 \\
					a \ar[r,|->] \& 13
				\end{tikzcd}
				\right)
				& &
				\chi_{w^8} = \left(  
				\begin{tikzcd}[row sep=0em,ampersand replacement=\&]
					y \ar[r,|->] \& 5 \\
					z \ar[r,|->] \& 9 \\
					b \ar[r,|->] \& 13
				\end{tikzcd}
				\right)
			\end{align*}
			with the understanding that $\chi_{w^5}(y)=\chi_{w^6}(y)=5+9$.
			
			Now, how should we define $\Psi(\chi_w)$ in such a way that $\sum_{w} \Psi(\chi_w)=1$? We shall use the image along $\psi$ of the elements we chose in each $\chi_w$, multiplying them all, and dividing by the product of $\Phi(A_i)$, for example,
			\[
			\Psi(\chi_{w^1}) = \frac{\psi(A_1,x) \cdot \psi(A_2,y) \cdot \psi(A_3,a)} {\Phi(A_1) \cdot \Phi(A_2) \cdot \Phi(A_3)}.
			\]
			(Here we crucially rely on the fact that $\Rp$ is a semi\emph{field}.) So, we have
			\begin{align*}
				\Psi(\chi_{w^1}) = \frac{2 \cdot 4 \cdot 6}{5 \cdot 9 \cdot 13}
				& &
				\Psi(\chi_{w^2}) = \frac{2 \cdot 4 \cdot 7}{5 \cdot 9 \cdot 13}
				\\[1em]
				\Psi(\chi_{w^3}) = \frac{2 \cdot 5 \cdot 6}{5 \cdot 9 \cdot 13}
				& &
				\Psi(\chi_{w^4}) = \frac{2 \cdot 5 \cdot 7}{5 \cdot 9 \cdot 13}
				\\[1em]
				\Psi(\chi_{w^5}) = \frac{3 \cdot 4 \cdot 6}{5 \cdot 9 \cdot 13}
				& &
				\Psi(\chi_{w^6}) = \frac{3 \cdot 4 \cdot 7}{5 \cdot 9 \cdot 13}
				\\[1em]
				\Psi(\chi_{w^7}) = \frac{3 \cdot 5 \cdot 6}{5 \cdot 9 \cdot 13}
				& &
				\Psi(\chi_{w^8}) = \frac{3 \cdot 5 \cdot 7}{5 \cdot 9 \cdot 13}
			\end{align*}
			Now, in order to calculate $\muS_X(\Psi)(y)$, for example, we have two, equivalent algorithms we can follow. The first one scans each $\chi_{w^k}$ (with $1 \le k \le 8$) one at a time, checks each ``line'' (indexed by $i \in \{1,2,3\}$) to see if $y$ appears there, and if it does, writes down the addend $\Psi(\chi_{w^k}) \cdot \chi_{w^k}(y)$ (where the second factor is $\Phi(A_1)$ if $y$ appears in the first line, and $\Phi(A_2)$ if it appears in the second), moves to the next line of $\chi_{w^k}$, repeats until it has to change $k$. The other algorithm, instead, scans first each $i \in \natset 3$ and checks for each $k$ if $y$ appears in the $i$-th line of $\chi_{w^k}$. If it does, it writes down the addend $\Psi(\chi_{w^k}) \cdot \chi_{w^k}(y)$, then moves on to the next function, until it has to change $i$. Following this second approach, we obtain that $\muS_X(\Psi)(y)$ is equal to:
			\begin{align*}
				5 \cdot \frac{3 \cdot 4 \cdot 6}{5 \cdot 9 \cdot 13} + 5 \cdot \frac{3 \cdot 4 \cdot 7}{5 \cdot 9 \cdot 13} + 5 \cdot \frac{3 \cdot 5 \cdot 6}{5 \cdot 9 \cdot 13} + 5 \cdot \frac{3 \cdot 5 \cdot 7}{5 \cdot 9 \cdot 13}+ & & (i=1,\,k=5,6,7,8) \\[1em]
				+9 \cdot \frac{2 \cdot 4 \cdot 6}{5 \cdot 9 \cdot 13} + 9 \cdot \frac{2 \cdot 4 \cdot 7}{5 \cdot 9 \cdot 13} + 9 \cdot \frac{3 \cdot 4 \cdot 6}{5 \cdot 9 \cdot 13} + 9 \cdot \frac{3 \cdot 4 \cdot 7}{5 \cdot 9 \cdot 13} & & (i=2,\,k=1,2,5,6)
			\end{align*} 
			Let us look at the first row, corresponding to index $i=1$. There are a few repeated numbers, which we can isolate on the left: there is $5=\Phi(A_1)$, $\frac{1} {5 \cdot 9 \cdot 13} = \frac 1 {\Phi(A_1) \cdot \Phi(A_2) \cdot \Phi(A3)}$, and there is also $3 = \psi(A_1,y)$. Similarly, in the second row, when $i=2$, we can isolate $\Phi(A_2)$, $\frac 1 {\Phi(A_1) \cdot \Phi(A_2) \cdot \Phi(A3)}$ and $\psi(A_2,y)$. So, we have that $\mu(\Psi)(y)$ is equal to
			\begin{align*}
				5 \cdot \frac{1} {5 \cdot 9 \cdot 13} \cdot 3 \cdot (4 \cdot 6 + 4 \cdot 7 + 5 \cdot 6 + 5 \cdot 7) +\\
				+9 \cdot \frac 1 {5 \cdot 9 \cdot 13} \cdot 4 \cdot (2 \cdot 6 + 2 \cdot 7 + 3 \cdot 6 + 3 \cdot 7)
			\end{align*}
			Now, by collecting together the same factors, we have that
			\begin{equation}\label{eqn:example generalised distributivity}
				4 \cdot 6 + 4 \cdot 7 + 5 \cdot 6 + 5 \cdot 7 = (4+5) \cdot (6+7)
			\end{equation}
			which is precisely, by construction, $\Phi(A_1) \cdot \Phi(A_3)$. Hence, by simplifying everything, we obtain
			\[
			\muS_X(\Psi)(y)=3+4=\psi(A_1,y)+\psi(A_2,y) = \phi(y).
			\]
			In a similar way, one can also prove that $\muS_X(\Psi)(t)=\phi(t)$ for all $t \in X$.
		\end{exa}
		
		The proof of Theorem~\ref{thm:delta for positive refinable semifields} is a direct generalisation of the calculations taken in the previous example. We need a lemma that generalises the distributivity property in an arbitrary semiring, showing how the observation of equation~(\ref{eqn:example generalised distributivity}) can always be done. The statement of the following Lemma makes sense only for commutative semirings, but it can be adapted for arbitrary semirings by restricting it to sets $L$ only of the form $\natset n = \{1,\dots,n\}$.
		
		\begin{lem}\label{lemma:generalised distributivity}
			Let $n \in \N$, $L \subset \N$ such that $\card L = n$,  $(s_k)_{k \in L} \in \N^n$, and $(\lambda^k_j)$ a family of elements of $S$, with $k \in L$ and $j \in \natset{s_k}$.  Then
			\begin{equation}\label{eqn:generalised distributivity}
				\sum_{w \in \prod\limits_{k \in L} \natset {s_k}} \, \prod_{k \in L} \lambda^k_{w_k} = 
				\prod_{k \in L} \, \sum_{j=1}^{s_k} \lambda^k_j.
			\end{equation}
			(Mind the overloading use of $\prod$ as a cartesian product of the sets $\underline{s_k}$ and also as a product of the elements $\lambda^k_{w_k}$ and $\sum_{j=1}^{s_k} \lambda^k_j$ in $S$.)
		\end{lem}
		\begin{proof}
			By induction on $n$. If $n=0$, both sides of~(\ref{eqn:generalised distributivity}) are $0$. 
			
			Let now $n \ge 0 $, suppose that the statement of the Lemma holds for $n$, let us consider a $L \subset \N$ such that $\card L = n+1$, $(s_k)_{k \in L} \in \N^{n+1}$, $(\lambda^k_j)_{\substack{k \in L \\ j \in \natset{s_k}}}$. Without loss of generality we can assume that $L=\natset{n+1}$. Then we have:
			\begin{align*}
				\sum_{w \in \prod\limits_{k =1}^{n+1} \natset{s_k}} \, \prod_{k =1}^{n+1} \lambda^k_{w_k} &=
				\sum_{j=1}^{s_{n+1}} \biggl( \sum_{w' \in \prod\limits_{k =1}^{n} \natset{s_k}} \, \Bigl((\prod_{k =1}^{n} \lambda^k_{w_k'}) \cdot \lambda^{n+1}_j \Bigr) \biggr) \displaybreak[0] \\
				&= \sum_{j=1}^{s_{n+1}} \biggl( \Bigl( \sum_{w' \in \prod\limits_{k =1}^{n} \natset{s_k}} \, \prod_{k =1}^{n} \lambda^k_{w_k'}\Bigr) \cdot \lambda^{n+1}_j  \biggr) \displaybreak[0] \\
				&= \sum_{j=1}^{s_{n+1}} \biggl( \Bigl( \prod_{k=1}^n \, \sum_{u=1}^{s_k} \lambda^k_u \Bigr) \cdot \lambda^{n+1}_j \biggr) \displaybreak[0] \\
				&= \Bigl( \prod_{k=1}^n \, \sum_{u=1}^{s_k} \lambda^k_u \Bigr) \cdot \Bigl( \sum_{j=1}^{s_{n+1}} \lambda^{n+1}_j  \Bigr)\displaybreak[0]  \\
				&= \prod_{k=1}^{n+1} \, \sum_{j=1}^{s_k} \lambda^k_j. \qedhere
			\end{align*} 
		\end{proof}
		
		\subsubsection*{Proof of Theorem~\ref{thm:delta for positive refinable semifields}.} 	We shall discuss first some preliminary cases involving the empty set, excluding each time all the cases previously covered. Recall that 
		\[
		\delta_X(\Phi)=\Biggl\{ \phi \in \mon S(X) \mid \exists \psi \in \mon S (\ni_X) \ldotp \begin{cases}
			\forall A \in \pow X \ldotp \Phi(A) = \sum_{x \in A} \psi(A,x) \\
			\forall x \in X \ldotp \phi(x) = \sum_{A \ni x} \psi(A,x)
		\end{cases}  
		\Biggr\}
		\]
		and write, within the scope of this proof, $\delta_X'(\Phi) = \convclos{\choice \Phi}{\muS_X}$, that is:
		\[
		\delta_X'(\Phi)= \left\{ \muS_X(\Psi) \mid \Psi \in \mon S^2 X\ldotp \sum\limits_{\chi \in \mon S X} \Psi(\chi)=1, \, \supp \Psi \subseteq \choice\Phi \right\}
		\]
		where
		\[
		\choice \Phi =\{ \chi \in \mon S X \mid \exists u \in \prod\limits_{A \in \supp \Phi} A \ldotp \forall x \in X \ldotp \chi(x) = \sum\limits_{\substack{A \in \supp \Phi \\ x=u_A}} \Phi(A)  \}
		\]
		(this is of course an equivalent formulation of~\eqref{eqn:c(Phi)}). We need to prove that $\delta_X(\Phi) = \delta_X'(\Phi)$. Since there is no possibility of confusion, we shall write $\ni$ instead of $\ni_X$ for the rest of the proof. Recall that we denote with $\zero$ the null function, \ie the function constantly equal to $0$.
		\paragraph{Case 1: $\Phi(\emptyset)\ne 0$.} We have that $\delta_X(\Phi)=\emptyset$ because there is no $\psi \in \mon S (\ni)$ such that $\Phi(\emptyset)=\sum_{x \in \emptyset} \psi(A,x)=0$. At the same time, also $\delta_X'(\Phi)=\emptyset$, because $\choice{\Phi}=\emptyset$--since $\prod_{A \in \supp \Phi} A = \emptyset$--hence there is no $\Psi \in \mon S ^2 X$ with empty support that can satisfy $\sum_{\chi \in \mon S X} \Psi(\chi)=1$.
		
		\paragraph{Case 2: $X = \emptyset$.} (We also assume that $\Phi(\emptyset)=0$ from now on.) We have that $\mon S \pow \emptyset = \{ \Omega \colon \{\emptyset\} \to S \}$, therefore $\Phi = \zero \colon \{\emptyset\} \to S$, the null function. Moreover, ${}\ni{} \subseteq \pow(\emptyset) \times \emptyset = \emptyset$ and $\mon S (\emptyset) = \{\emptyset \colon \emptyset \to S\}$ is the singleton of the empty map, so
		\begin{align*}
			\delta_\emptyset (\Phi) &= \{ \phi \in \mon S (\emptyset) \mid \exists \psi \in \mon S (\emptyset) \ldotp \forall A \subseteq \emptyset \ldotp \Phi(A)=\sum_{x \in \emptyset} \psi(A,x) \} \\
			&= \{\phi \in \mon S (\emptyset) \mid \Phi(\emptyset)=0\} \\
			&= \{ \emptyset \colon \emptyset \to S  \}
		\end{align*}
		because, by assumption, $\Phi(\emptyset) = 0$. On the other hand, we have that, since $\supp \Phi = \emptyset$, 
		\[
		\choice\Phi = \{ \chi \in \mon S (\emptyset) \mid \exists u \in \prod\limits_{A \in \emptyset} A  \} = \mon S (\emptyset)
		\]
		because the zero-ary product is a choice of a terminal object of $\Set$, a singleton. So, 
		\begin{align*}
			\delta_\emptyset'(\Phi) &= \{ \mu(\Psi) \mid \Psi \in \mon S ^2 \emptyset, \sum_{\chi \in \mon S \emptyset} \Psi(\chi)=1\} \\
			&= \{ \mu(\mon S (\emptyset) \ni \emptyset \mapsto 1)  \} = \{\emptyset \colon \emptyset \to S  \}.
		\end{align*}
		
		\paragraph{Case 3: $\Phi = \zero \colon \pow X \to S$.} (We also assume that $X \ne \emptyset$ from now on.) We have that the only $\psi \in \mon S (\ni)$ such that for all $\sum_{x \in A} \psi(A,x) = 0$ for all $A \subseteq X$ is the null function, therefore $\delta_X(\zero \colon \pow X \to S)=\{\zero \colon X \to S\}$. On the other hand, we have that $\supp \Phi = \emptyset$, so
		\[
		\choice\Phi = \{ \chi \in \mon S X \mid \exists u \in \prod_{A \in \emptyset} A \ldotp \forall x \in X \ldotp \chi(x)=\sum_{A \in \emptyset} \Phi(A) \} = \{\zero \colon X \to S\}.
		\]
		It follows then that
		\begin{align*}
			\delta_X'(\zero \colon \pow S \to S) &= \{\mu(\Psi) \mid \Psi \in \mon S ^2 X, \sum_{\chi \in \mon S X} \Psi(\chi)=1, \supp \Psi \subseteq \{0 \colon X \to S\} \} \\
			&= \{ \mu(\mon S X \ni 0 \mapsto 1)  \} \\
			&= \{ (S \ni x \mapsto 1 \cdot 0)  \} = \{\zero \colon X \to S\}.
		\end{align*}
		
		We have now discussed all the preliminary cases. For the rest of the proof, we shall assume that $X \ne \emptyset$, $\Phi \ne \zero \colon \pow X \to S$ and $\Phi(\emptyset) = 0$.
		
		We first prove $\delta_X(\Phi) \subseteq \delta_X'(\Phi)$. To this end, let $\phi \in \mon S X$ and $\psi \in \mon S (\ni)$ such that $\Phi(A) = \sum_{x \in A} \psi(A,x)$ for all $A \in \pow X$ and $\phi(x)=\sum_{A \ni x} \psi(A,x)$ for all $x \in X$. Observe that:
		\begin{itemize}
			\item for all $(A,x) \in \supp \psi$ we have that $A \in \supp \Phi$ and $x \in \supp \phi \cap A$,
			\item for all $A \in \supp \Phi$ there is $x \in \supp \phi \cap A$ such that $(A,x) \in \supp \psi$,
			\item for all $x \in \supp \phi$ there exists $A \in \supp \Phi$ such that $A \ni x$ and $(A,x) \in \supp \psi$.
		\end{itemize}
		Hence the first elements of pairs in $\supp \psi$ range over all and only elements of $\supp \Phi$, and the second entries of pairs in $\supp \psi$ range over all and only elements of $\supp \phi$. In other words, say $\supp \Phi = \{A_1,\dots,A_n\}$: then we have
		\[
		\supp \psi = \Bigl\{ (A_1,x^1_1), \dots, (A_1,x^1_{s_1}), \dots, (A_n, x^n_1), \dots, (A_n,x^n_{s_n}) \Bigr\}
		\]
		where $\bigcup_{i=1}^n \{x^i_1,\dots,x^i_{s_i}\} = \supp \phi$. Notice that for all $i \in \natset{n}$, $u,v \in \natset{s_i}$, if $u \ne v$ then $x^i_u \ne x^i_v$, but we may have $x^i_a = x^j_b$ if $i \ne j$, because the $A_i$'s are distinct but not \emph{disjoint}. We can then write:
		\begin{enumerate}[(I)]
			\item\label{enumerate 1 lemma delta = delta'}  $\Phi(A_i) = \sum\limits_{j=1}^{s_i} \psi(A_i,x^i_j)$ for all $i \in \natset{n}$
			\item\label{enumerate 2 lemma delta = delta'}  $\phi(x) = \sum\limits_{\substack{i \in \natset n \\ \exists j \in \natset{s_i}\ldotp x = x^i_j}} \psi(A_i,x)$ for all $x \in X$.
		\end{enumerate}
		We want to find a convex linear combination $\Psi$ of elements of $\mon S X$ of the form $\sum_{i=1}^n \Phi(A_i) \cdot a_i$ for some $a_i \in A_i$ such that $\muS_X(\Psi)=\phi$. Now, for every $i \in \natset{n}$, we have many candidates for $a_i$, namely $x^i_1, \dots, x^i_{s_i}$. Given that every $\chi \in \supp \Psi$ can only involve one $x^i_j$ for every $i$, we shall have as many $\chi$'s as the number of ways to pick one element of $\{ x^i_1,\dots,x^i_{s_i} \}$ for each $i$: let then $w \in \prod_{i=1}^n \natset{s_i}$ (it is a tuple of indexes $w_i$ which we are going to use as $j$'s). Define $\chi_w$ as
		\[
		\chi_w(x^1_{w_1}) = \Phi(A_1), \,\dots,\, 
		\chi_w(x^n_{w_n}) = \Phi(A_n) 
		\]
		with the understanding that if $x^i_{w_i} =  x^j_{w_j}$ then $\chi_w(x^i_{w_i}) = \Phi(A_i) + \Phi(A_j)$. Written more precisely, we define for all $x \in X$
		\[
		\chi_w (x) = \sum_{\substack{i \in \natset n \\ x = x^i_{w_i}}} \Phi(A_i).
		\]
		We now define $\Psi \in \mon S ^2 X$ with $\supp \Psi=\{ \chi_w \mid w \in \prod_{i=1}^n \natset{s_i}  \}$ that assigns to each $\chi_w$ a number such that $\sum_w \Psi(\chi_w) = 1$. We shall use the previous Lemma to show that by defining
		\[
		\Psi(\chi_w) = \frac{ \prod_{i=1}^n \psi(A_i,x^i_{w_i}) } { \prod_{i=1}^n \Phi(A_i) }
		\]
		for all $w$, or more precisely by defining
		\[
		\Psi(\chi) = \sum_{\substack{w \in \prod_{i=1}^n \natset{s_i} \\ \chi=\chi_w}} \frac{ \prod_{i=1}^n \psi(A_i,x^i_{w_i}) } { \prod_{i=1}^n \Phi(A_i) }
		\]
		for all $\chi \in \mon S X$ we indeed have that $\Psi$ satisfies the conditions of $\delta_X'(\Phi)$ and $\phi=\muS_X(\Psi)$. First we prove that $\Psi$ is in fact a convex linear combination:
		\begin{align*}
			\sum_{\chi \in \supp \Psi} \Psi(\chi) &= \frac 1 {\prod_{i=1}^n \Phi(A_i)} \sum_{w \in \prod \natset{s_i}} \prod_{i=1}^n \psi(A_i,x^i_{w_i})  \displaybreak[0]\\
			&= \frac 1 {\prod_{i=1}^n \Phi(A_i)} \prod_{i=1}^n \sum_{j=1}^{s_i} \psi(A_i,x^i_j) & \text{Lemma~\ref{lemma:generalised distributivity}} \displaybreak[0]\\
			&= \frac 1 {\prod_{i=1}^n \Phi(A_i)} {\prod_{i=1}^n \Phi(A_i)}  & \text{because of \ref{enumerate 1 lemma delta = delta'}} \\
			&=1.
		\end{align*}
		Next, notice that for every $w$ the vector $u \in \prod_{i=1}^n A_i$ required by the definition of $\delta_X'(\Phi)$ is exactly $(x^1_{w_1},\dots,x^n_{w_n})$. Finally, we compute $\muS_X(\Psi)(x)$ for an arbitrary $x \in X$. The equations marked with $(\ast)$ will be explained later.
		\begin{align*}
			\muS_X(\Psi)(x) &= \sum_{\chi \in \supp \Psi} \Psi(\chi) \cdot \chi(x) \displaybreak[0] \\
			&= \sum_{w \in \prod_{k=1}^n \natset{s_k}} \Biggl[ \frac {\prod_{i=1}^n \psi(A_i,x^i_{w_i})} {\prod_{i=1}^n \Phi(A_i)} \sum_{\substack{i \in \natset{n} \\ x=x^i_{w_i}}} \Phi(A_i) \Biggr] & (\ast_1) \displaybreak[0] \\
			&= \sum_{w \in \prod_{k=1}^n \natset{s_k}} \, \sum_{\substack{i \in \natset{n} \\ x=x^i_{w_i}}} \Biggl[   \frac {\prod_{k=1}^n \psi(A_k,x^k_{w_k})} {\prod_{k=1}^n \Phi(A_k)} \Phi(A_i)  \Biggr] \displaybreak[0] \\
			&= \sum_{i \in \natset{n}} \,\, \sum_{\substack{w \in \prod_{k=1}^n \natset{s_k} \\ x=x^i_{w_i} }} \Biggl[   \frac {\prod_{k=1}^n \psi(A_k,x^k_{w_k})} {\prod_{k=1}^n \Phi(A_k)} \Phi(A_i)  \Biggr]  \displaybreak[0] \\
			&= \sum_{\substack{i \in \natset n \\ \exists j \in \natset{s_i} \ldotp x = x^i_j}} \, \sum_{\substack{w \in \prod_{k=1}^n \natset{s_k} \\ x=x^i_{w_i} }} \Biggl[ \Phi(A_i) \frac {\psi(A_i,x^i_{w_i})} {\prod_{k=1}^n \Phi(A_k)} \prod_{\substack{k \in \natset n \\ k \ne i}} \psi(A_k,x^k_{w_k}) \Biggr] & (\ast_2) \displaybreak[0] \\
			&= \sum_{\substack{i \in \natset n \\ \exists j \in \natset{s_i} \ldotp x = x^i_j}} \Phi(A_i)  \cdot \frac {\psi(A_i,x)} {\prod_{k=1}^n \Phi(A_k)} \cdot \sum_{\substack{w \in \prod_{k=1}^n \natset{s_k} \\ x=x^i_{w_i} }} \, \prod_{\substack{k \in \natset n \\ k \ne i}} \psi(A_k,x^k_{w_k})
		\end{align*}
		Now, use Lemma~\ref{lemma:generalised distributivity} with $L=\{1,\dots,i-1,i+1,\dots,n\}$ and $\lambda^k_j = \psi(A_k,x^k_j)$ in the following chain of equations:
		\begin{align*}
			\sum_{\substack{w \in \prod_{k=1}^n \natset{s_k} \\ x=x^i_{w_i} }} \, \prod_{\substack{k \in \natset n \\ k \ne i}} \psi(A_k,x^k_{w_k}) &= \sum_{w' \in \prod_{k \in L} \natset{s_k} } \, \prod_{k \in L} \psi(A_k,x^k_{w_k'}) \displaybreak[0] \\
			&= \prod_{k \in L} \, \sum_{j=1}^{s_k} \psi(A_k, x^k_j) \displaybreak[0] \\
			&= \prod_{\substack{k \in \natset n \\ k \ne i}} \, \sum_{j=1}^{s_k} \psi(A_k,x^k_j) \displaybreak[0] \\
			&= \prod_{\substack{k \in \natset n \\ k \ne i}} \Phi(A_k)
		\end{align*}
		Therefore, we obtain
		\begin{align*}
			\muS_X(\Phi) &= \sum_{\substack{i \in \natset n \\ \exists j \in \natset{s_i} \ldotp x = x^i_j}} \Phi(A_i)  \cdot \frac {\psi(A_i,x)} {\prod_{k=1}^n \Phi(A_k)} \cdot \prod_{\substack{k \in \natset n \\ k \ne i}} \Phi(A_k) \\
			&= \sum_{\substack{i \in \natset n \\ \exists j \in \natset{s_i} \ldotp x = x^i_j}} \psi(A_i,x) \\
			&= \phi(x) & \text{because of \ref{enumerate 2 lemma delta = delta'}}.
		\end{align*}
		It remains to explain equations $(\ast_1)$ and $(\ast_2)$ above. The latter is simply due to the fact that if $i$ is such that there is no $j$ for which $x =x^i_j$, then there is not a $w$ such that $x=x^i_{w_i}$ either, and vice versa. The former is more delicate. It may be the case that $\chi_w = \chi_{w'}$ for $w \ne w'$. So, let us write
		\[
		\supp \Psi = \{\chi_{w^1}, \dots, \chi_{w^m}\}
		\]
		where the $\chi_{w^l}$ are now all distinct. Then we have
		\begin{align*}
			\sum_{\chi \in \supp \Psi} \Psi(\chi) \cdot \chi(x) &= \sum_{l=1}^m \Psi(\chi_{w^l}) \cdot \chi_{w^l}(x) \\
			&= \sum_{l=1}^m \Biggl( \sum_{\substack{w \\ \chi_{w^l} = \chi_w}} \frac{ \prod_{i=1}^n \psi(A_i,x^i_{w_i}) } { \prod_{i=1}^n \Phi(A_i) } \Biggr) \chi_{w^l}(x) \displaybreak[0] \\
			&= \sum_{l=1}^m\, \sum_{\substack{w \\ \chi_{w^l} = \chi_w}} \frac{ \prod_{i=1}^n \psi(A_i,x^i_{w_i}) } { \prod_{i=1}^n \Phi(A_i) } \chi_{w^l}(x)  \displaybreak[0] \\
			&= \sum_{l=1}^m\, \sum_{\substack{w \\ \chi_{w^l} = \chi_w}} \frac{ \prod_{i=1}^n \psi(A_i,x^i_{w_i}) } { \prod_{i=1}^n \Phi(A_i) } \chi_w(x) \displaybreak[0] \\
			&= \sum_{w} \frac{ \prod_{i=1}^n \psi(A_i,x^i_{w_i}) } { \prod_{i=1}^n \Phi(A_i) } \chi_w(x)
		\end{align*}
		which is the right-hand side of $(\ast_1)$.
		
		The other inclusion, $\delta_X'(\Phi) \subseteq \delta_X(\Phi)$, is easier. Let $\Psi \in \mon S ^2 X$ be such that $\sum_{\chi \in \mon S X} \Psi(\chi) = 1$ and 
		\[
		\supp \Psi \subseteq \{ \chi \in \mon S X \mid \exists u \in \prod\limits_{A \in \supp \Phi} A \ldotp \forall x \in X \ldotp \chi(x) = \sum\limits_{\substack{A \in \supp \Phi \\ x=u_A}} \Phi(A)  \}.
		\]
		Write again $\supp \Phi=\{A_1,\dots,A_n\}$ and let $\supp \Psi = \{\chi_1,\dots,\chi_m\}$. For all $j \in \natset m$, let $u^j \in \prod_{i=1}^n A_i$ be such that $\chi_j(x) = \sum_{\substack{i \in \natset n \\ x = u^j_i}} \Phi(A_i)$. Then we have
		\[
		\muS_X(\Psi)(x) = \sum_{j=1}^m \Psi(\chi_j) \chi_j(x) = \sum_{j=1}^m \Psi(\chi_j) \sum_{\substack{i \in \natset n \\ x = u^j_i}} \Phi(A_i).
		\]
		Define then
		\[
		\psi(B,x) = \begin{cases}
			\sum_{\substack{j \in \natset m \\ x=u^j_i }} \Psi(\chi_j) \Phi(A_i) & \exists (!) i \in \natset n \ldotp B=A_i \\
			0 & \text{otherwise}
		\end{cases}
		\]
		We have $\supp \psi = \{ (A_i,u^j_i) \mid i \in \natset n, j \in \natset m  \}$ is finite and so $\psi \in \mon S (\ni)$. We now verify the two conditions required in the definition of $\delta_X(\Phi)$. For all $x \in X$:
		\begin{align*}
			\sum_{A \ni x} \psi(A,x) &= \sum_{i=1}^n \psi(A_i,x) = \sum_{i=1}^n \sum_{\substack{j \in \natset m \\ x=u^j_i }} \Psi(\chi_j) \Phi(A_i) = \sum_{j=1}^m \sum_{ \substack{i \in \natset n \\ x = u^j_i} } \Psi(\chi_j) \Phi(A_i) \\
			&= \muS_X(\Psi)(x)
		\end{align*}
		while for all $A \subseteq X$, if $A \notin \supp \Phi$, then $\sum_{x \in A} (\psi(A,x)) = 0 = \Phi(A)$ by definition and, for all $i \in \natset n$:
		\begin{align*}
			\sum_{x \in A_i} \Psi(A_i,x) &= \sum_{x \in A_i} \sum_{\substack{j \in \natset m \\ x = u^j_i}} \psi(\chi_j) \Phi(A_i) \overset{(\ast_3)}{=} \sum_{j=1}^m \Psi(\chi_j) \Phi(A_i) = \Phi(A_i) \sum_{j=1}^m \Psi(\chi_j) \\
			&= \Phi(A_i)
		\end{align*}
		where equation $(\ast_3)$ holds because for all $j \in \natset m$ the addend $\Psi(\chi_j)\Phi(A_i)$ appears on the left-hand side exactly once, when in the first sum we are using, as $x$, precisely $u^j_i \in A_i$. In other words, given $j$, there is a unique $x \in A_i$ such that $x=u^j_i$, so $\Psi(\chi_j)\Phi(A_i)$ appears in the left-hand side of $(\ast_3)$ and it does so only once. Therefore $\muS_X(\Psi) \in \delta_X(\Phi)$, and the proof is complete. \qed

		\section{The Canonical Weak Lifting of \texorpdfstring{$\P$}{\emph{P}} to \texorpdfstring{$\EM\S$}{\emph{EM(S)}}}\label{sec: the weak lifting}
		
		By exploiting the characterisation of the weak distributive law $\delta$ for positive semifields (Theorem~\ref{thm:delta for positive refinable semifields}), we can now describe the weak lifting of $\P$ to $\EM\S$ generated by $\delta$, by listing all the data required by Definition~\ref{def:weak lifting}.
		
		Let $(X,a \colon \S X \to X)$ be a $\S$-algebra. Define $\tilde \P (X,a) = (\ConvPow X a, \alpha_a)$, where recall from Definition~\ref{def:convex closure S-ALGEBRA} that $\ConvPow X a$ is the set of convex subsets of $X$ with respect to $a$, while the function $\alpha_a \colon \mon{S} \ConvPow X a \to \ConvPow X a$ is defined for all $\Phi \in \mon{S} \ConvPow X a$ as
		\begin{equation}\label{eq:alphaa}
			\alpha_a(\Phi)= \{a(\phi) \; \mid \; \phi \in \choice{\Phi} \}\text{.}
		\end{equation}
		To be completely formal, above we should have written $\choice{\mon S (\iota) (\Phi)}$ in place of $\choice{\Phi}$, but it is immediate to see that the two sets coincide. Proving that $\alpha_a \colon \mon{S} \ConvPow X a \to \ConvPow X a$ is well defined (namely, $\alpha_a(\Phi)$ is a convex set) and forms an $\mon S$-algebra requires some ingenuity and will be shown later in Section~\ref{sec:proof}. Next, define $\tilde \P$ on morphisms $f \colon (X,a) \to (X',a')$  as
		\begin{equation}\label{eq:f}
			\tilde \P(f)(A) =  \pow f (A)
		\end{equation}
		for all $A\in \ConvPow X a$. This makes $\tilde \P$ an endofunctor on $\EM\S$. For all $(X,a)$ in $\EM{\mon S}$, $\eta^{\tilde \P}_{(X,a)} \colon (X,a) \to \tilde \P (X,a) $ and  $\mu^{\tilde \P}_{(X,a)} \colon \tilde \P \tilde \P (X,a) \to \tilde \P (X,a)$ are defined for $x\in X$ and $\mathcal A \in \ConvPow {(\ConvPow X a)} {\alpha_a}$ as
		\begin{equation}\label{eq:etamu}
			\eta^{\tilde \P}_{(X,a)}(x) =  \{x\} \qquad \text{and} \qquad \mu^{\tilde \P}_{(X,a)} (\mathcal A) = \bigcup_{A \in \mathcal A} A\text{.} 
		\end{equation}
		Notice that $\{ x \}$ is of course a convex subset of $X$ with respect to the structure map $a$, but it is not at all obvious that $\bigcup_{A \in \mathcal A} A$ is. In general the union of convex subsets of $X$ is not again convex, but here $\mathcal A$ is not an arbitrary family of convex sets: it is a \emph{convex} subset of $\ConvPow X a$ with respect to $\alpha_a$, and this ensures that its union is convex in $X$ with respect to $a$. We will prove this in Proposition~\ref{prop:union of convex family of convex subsets is convex}.
		Finally, define the functions $\iota_{(X,a)} \colon	\ConvPow X a \to \pow X$
		and $\pi_{(X,a)} \colon \pow X \to \ConvPow X a$ as, for all $A\in \ConvPow X a$ and $B \in \pow X$,
		\begin{equation}\label{eq:iotapi}
			\iota_{(X,a)}(A)=A    \qquad \text{and} \qquad \pi_{(X,a)}(B)=\convclos B a \text{,}
		\end{equation}
		that is $\iota_{(X,a)}$ is just the obvious set inclusion and 
		$\pi_{(X,a)}$ performs the convex closure in $a$.
		
		\begin{thm}\label{thm:weaklift}
			Let $S$ be a positive semifield. Then the canonical weak lifting of the powerset monad $\P$ to $\EM{\mon S}$, determined by the weak distributive law $\delta$ given in~\eqref{eqn:delta for positive refinable semifields}, consists of the monad $(\tilde \P,\eta^{\tilde \P},\mu^{\tilde \P})$ on $\EM{\mon S}$ defined as in \eqref{eq:alphaa}, \eqref{eq:f}, \eqref{eq:etamu} and the natural transformations $\iota \colon \U{\mon S}\tilde \P \to \P \U{\mon S}$ and $\pi \colon \P \U{\mon S} \to \U{\mon S} \tilde \P$ defined as in \eqref{eq:iotapi}.
		\end{thm}
		
		It is worth spelling out the left-semimodule structure on $\ConvPow X a$ corresponding to the $\mon S$-algebra $\alpha_a \colon \mon{S} \ConvPow X a \to \ConvPow X a$. 
		
		Let us start with $\lambda \cdot^{\alpha_a} A $ for some $A\in \ConvPow X a$. By~\eqref{eqn:semimodule associated to S-algebra},  $\lambda \cdot^{\alpha_a} A  = \alpha_a(\Phi)$ where $\Phi=(A \mapsto \lambda)$. By~\eqref{eq:alphaa}, $\alpha_a(\Phi)=\{a(\phi) \; \mid \; \phi \in \choice{\Phi} \}$. Following the definition of $\choice{\Phi}$ given in \eqref{eqn:c(Phi)}, one has to consider functions $u\colon \supp \Phi \to X$ such that $u(B)\in B$ for all $B\in \supp \Phi$: if $\lambda \neq 0$, then $\supp \Phi= \{A\}$ and thus, for each $x\in A$, there is exactly one function $u_x \colon \supp \Phi \to X$ mapping $A$ into $x$. It is immediate to see that $\mon S (u_{x}) (\Phi)$ is exactly the function $(x \mapsto \lambda)$ and thus $a(\mon S (u_{x}) (\Phi))$ is, by \eqref{eqn:semimodule associated to S-algebra}, $\lambda \cdot^a x$. Now if $\lambda=0$, then $\supp \Phi = \emptyset$, so there is \emph{exactly one} function $u\colon \supp \Phi \to X$ and $\mon S (u) (\Phi)$ is the function mapping all $x \in X$ into $0$ and thus, by \eqref{eqn:semimodule associated to S-algebra}, $a(\mon S (u) (\Phi)) = 0^a$. Summarising,
		\begin{equation}\label{eq:lambdaP}
			\lambda \cdot^{\alpha_a} A = \begin{cases}
				\{\lambda \cdot^a x \, \mid \; x \in A\} & \text{if }\lambda \neq 0\\
				\{0^a\}  & \text{if }\lambda = 0 
		\end{cases} \end{equation}
		Following similar lines of thoughts, one can check that 
		\begin{equation}\label{eq:sumP}
			A+^{\alpha_a} B = \{x+^ay \; \mid \; x\in A, \; y\in B\} \qquad \text{and} \qquad 0^{\alpha_a}=\{0^a\}\text{.}\end{equation}
		
		\begin{rem}\label{remarkJacobs} 
			By comparing \eqref{eq:sumP} and \eqref{eq:lambdaP} with (4) and (5) in~\cite{jacobs_coalgebraic_2008}, it is immediate to see
			that our monad $\tilde \P$ coincides with a slight variation of Jacobs's convex powerset monad $\mathcal C$, the only difference being that we do allow for $\emptyset$ to be in $\ConvPow X a$. Jacobs insisted on the necessity of $\mathcal C(X)$ to be the set of \emph{non-empty} convex subsets of $X$, because otherwise he was not able to define a semimodule structure on $\mathcal C (X)$ such that $0 \cdot \emptyset = \{0^a\}$. However, we do manage to do so, since by~\eqref{eq:lambdaP}, $0 \cdot A = \{0^a\}$ for all $A$ and in particular for $A=\emptyset$. At first sight, this may look like an ad-hoc solution, but this is not the case: it is intrinsic in the definition of the canonical weak lifting of $\P$ to $\EM{\mon S}$, as stated by Theorem \ref{thm:weaklift} and shown next.
		\end{rem}
		
		\subsection{Proof of Theorem~\ref{thm:weaklift}}\label{sec:proof}
		By Theorem~\ref{thm:correspondence weak distributive law with weak extensions-liftings}, the weak distributive law~\eqref{eqn:def of delta} corresponds to a weak lifting $\tilde \P$ of $\P$ to $\EM\S$, which we are going to show coincides with the data of~\eqref{eq:alphaa}-\eqref{eq:iotapi}. 
		The image along $\tilde \P$ of a $\mon S$-algebra $(X,a)$ will be a set $Y$ together with a structure map $\alpha_a$ that makes it a $\mon S$-algebra in turn. The proof of Theorem~\ref{thm:correspondence weak distributive law with weak extensions-liftings} gives us the recipe to build $Y$ and $\alpha_a$ appropriately. $Y$ is obtained by splitting the following idempotent in $\Set$:
		\begin{equation}\label{eqn:idempotent e}
			e_{(X,a)}=\begin{tikzcd}
				\pow X \ar[r,"\etaS_{\pow X}"] & \mon S (\pow X) \ar[r,"\delta_X"] & \pow(\mon S X) \ar[r,"\pow a"] & \pow X
			\end{tikzcd}
		\end{equation}
		as a composite $e_{(X,a)}=\iota_{(X,a)} \circ \pi_{(X,a)}$, where $\pi_{(X,a)}$ is the corestriction of $e_{(X,a)}$ to its image and $\iota_{(X,a)}$ is the set-inclusion of the image of $e_{(X,a)}$ into $\pow X$. In other words, $Y$ is the set of fixed points of $e_{(X,a)}$. $\alpha_a$ is obtained as the composite
		\[
		\alpha_a=\begin{tikzcd}
			\mon S Y \ar[r,"\mon S \iota_{(X,a)}"] & \mon S \pow X \ar[r,"\delta_X"] & \pow \mon S X \ar[r,"\pow a"] & \pow X \ar[r,"\pi_{(X,a)}"] & Y.
		\end{tikzcd}
		\]
		Let us, then, fix an $\S$-algebra $(X,a)$. Given $A \in \pow X$, we have $\etaS_{\pow X}(A)=\Delta_A \colon \pow X \to S$, the Dirac-function centred in $A$. The set $\delta_X(\etaS_{\pow X}(A))$ has a simple description, shown in the next Lemma: it consists of formal convex linear combinations of elements of $A$.

		\begin{lem}\label{lemma:delta of Dirac}
			For all $A \in \pow X$
			\begin{equation}\label{eqn:delta of Dirac}
				\delta_X(\etaS_{\pow X}(A)) = \left\{ \phi \in \mon S X \mid \supp \phi \subseteq A, \sum_{x \in X} \phi(x) = 1 \right\}.
			\end{equation}
		\end{lem}
		\begin{proof}
			Recall that, by definition,
			\[
			\delta_X(\etaS_{\pow X}(A)) = \left\{ \phi \in  \mon S X \mid \exists \psi \in \mon S (\ni) \ldotp \begin{cases}
				\forall B \in \pow X \ldotp \Delta_A(B) = \sum\limits_{x \in B} \psi(B,x) \\
				\forall x \in X \ldotp \phi(x) = \sum\limits_{B \ni x} \psi(B,x)
			\end{cases}   \!\!\! \right\}.
			\]
			(In this proof we write $\ni$ instead of $\ni_X$.) First of all, if $A = \emptyset$, then the left-hand side of \eqref{eqn:delta of Dirac} is empty, because there is no $\psi \in \mon S (\ni)$ such that $1=\Delta_{\emptyset}(\emptyset) = \sum_{x \in \emptyset} \psi(B,x) = 0$, and so is the right-hand side, as the only function whose support is empty is the null function, which cannot satisfy $\sum_{x \in X} \phi(x) = 1$. 
			
			Suppose then that $A \neq \emptyset$. 
			For the left-to-right inclusion of~\eqref{eqn:delta of Dirac}: observe first of all that $\supp \phi \ne \emptyset$, because otherwise we would have that $\psi(B,x)=0$ for all $x \in X$ and for all $B \ni x$ due to the fact that $S$ is a positive semiring. This would then lead to the contradiction $1=\Delta_A(A) = \sum_{x \in A} \psi(A,x) = 0$. Let then $x \in \supp \phi$: then $\phi(x) \ne 0$, hence there exists $B \ni x$ such that $\psi(B,x) \ne 0$. It is necessarily the case, however, that $B=A$, because if $B \ne A$, then $0 = \Delta_A(B) = \sum_{y \in B} \psi(B,y) \ge \psi(B,x) \ne 0$, which is a contradiction. Thus $\supp \phi \subseteq A$. Moreover,
			\[
			\sum_{x \in X} \phi(x) = \sum_{x \in A} \phi(x) = \sum_{x \in A} \sum_{B \ni x} \psi(B,x) \overset{(\ast)}{=} \sum_{x \in A} \psi(A,x) = \Delta_A(A)=1
			\]
			where equation $(\ast)$ holds because for all $B \ne A$ and for all $y \in B$ we have $\psi(B,y)=0$, hence the only addend of $\sum_{B \ni x} \psi(B,x)$ which is possibly not-null is the ``$A$-th one'': $\psi(A,x)$.
			
			Vice versa, let $\phi \colon X \to S$ such that $\supp \phi$ is finite and contained in $A$ and satisfying $\sum_{x \in X} \phi(x)=1$. Then the map
			\[
			\begin{tikzcd}[ampersand replacement=\&,row sep=0em]
				\ni \ar[r,"\psi"] \& \Rp \\
				(B,x) \ar[r,|->] \& \begin{cases}
					\phi(x) & B=A \\
					0 & \text{otherwise}
				\end{cases}
			\end{tikzcd}
			\] 
			has finite support, as this is in bijective correspondence with $\supp \phi$; for all $B \in \pow X$, if $B \ne A$ then $\sum_{x \in B} \psi(B,x)=0$, otherwise 
			\[
			\sum_{x \in B} \psi(B,x) = \sum_{x \in A} \phi(x) = \sum_{x \in X} \phi(x)=1;
			\]
			and finally, for all $x \in X$, 
			$
			\sum\limits_{B \ni x} \psi(B,x) = \psi(A,x)=\phi(x). 
			$
		\end{proof}
		
		The image along $A$ of the idempotent $e$ is therefore
		\[
		e(A)=\pow a (\delta_X(\etaS_{\pow X}(A))) = \left\{a(\phi) \mid \phi \in \mon S X, \supp \phi \subseteq A, \sum_{x \in X} \phi(x)=1 \right\} = \convclos A a.
		\]
		Hence the idempotent $e$ computes the convex closure of elements of $\pow X$ and its fixed points are precisely the convex subsets of $X$ with respect to the structure map $a$. Therefore, the carrier set of $\tilde \P (X,a)$ is precisely $\ConvPow X a$, the natural transformations $\pi$ and $\iota$ are, respectively, the convex closure operator and the set-inclusion of $\ConvPow X a$ into $\pow X$ as in \eqref{eq:iotapi}.
		
		\begin{rem}\label{rem:convexity still present even if not positive semifield}
			Lemma~\ref{lemma:delta of Dirac} holds even for semirings satisfying the minimal assumptions of Theorem~\ref{thm:existence of delta}: indeed in its proof we do not rely on the characterisation of $\delta$ given in Theorem~\ref{thm:delta for positive refinable semifields} for positive semifields. This means that although for these less well-behaved semirings we cannot, in general, express the canonical weak distributive law in terms of a convex closure, as noted in Remark~\ref{rem:nice delta does not hold for minimal semirings}, convexity is still deeply woven into the structure and it emerges explicitly in the canonical weak lifting of $\pow$ to $\EM\S$ associated to $\delta$. This is ultimately due to the fact that the element $1 \in S$ plays a protagonist role in the unit $\etaS$, and forces us to consider only those formal linear combinations whose coefficients add up to 1, as we saw in the proof of Lemma~\ref{lemma:delta of Dirac}.
		\end{rem}
		
		$\ConvPow X a$ is then equipped with a structure map $\alpha_a \colon \mon S \ConvPow X a \to \ConvPow X a$ given by
		\[
		\alpha_a=\begin{tikzcd}
			\mon S \ConvPow X a \ar[r,"\mon S \iota_{(X,a)}"] & \mon S \pow X \ar[r,"\delta_X"] & \pow \mon S X \ar[r,"\pow a"] & \pow X \ar[r,"\pi_{(X,a)}"] & \ConvPow X a.
		\end{tikzcd}
		\]
		Let us try to calculate $\alpha_a$: given $\Phi \colon \ConvPow X a \to S$ with finite support, we have that $\mon S ({\iota_{(X,a)}}) (\Phi)$ is just the extension of $\Phi$ to $\pow X$ which assigns $0$ to each non-convex subset of $X$. If we write $\iota$ instead of $\iota_{(X,a)}$ for short, we have
		\begin{equation}\label{eqn:alpha_a}
			\alpha_a(\Phi) = \convclos { \pow a (\delta_X (\mon S (\iota) (\Phi))) } {a}.
		\end{equation}
		Next, we can use the following technical result, which will also be useful later on and whose proof is in Appendix~\ref{app4}.
		\begin{prop}\label{prop:image along a of convex closures}
			Let $(X,a)$ be a $\mon S$-algebra and let $\mathcal A \subseteq \mon S X$. Then
			\[
			\pow a \Bigl( \convclos{\mathcal A}{\muS_X} \Bigr) = \convclos{\pow a (\mathcal A)}{a}.
			\]
		\end{prop}

		Since $\delta_X (\mon S (\iota) (\Phi))$ is the convex closure of $\choice{\mon S (\iota) (\Phi)}$ in $(\S X, \muS_X)$, by Proposition~\ref{prop:image along a of convex closures} and equation~\eqref{eqn:alpha_a} we get: 
		\[
		\alpha_a(\Phi)=\pow a (\convclos{\delta_X (\mon S (\iota) (\Phi))}{\muS_X})=\pow a (\delta_X (\mon S (\iota) (\Phi)))=\pow a \bigl(\convclos{\choice{\mon S (\iota) (\Phi)}}{\muS_X}\bigr).
		\] 
		It turns out, however, that also the $\muS_X$-convex closure is superfluous, due to the fact that $\Phi \in \S \ConvPow X a$ (and not simply $\S\P X$), thus obtaining~\eqref{eq:alphaa}. We prove this in Proposition~\ref{thm:convpow(X,a) is a S-algebra}, in Appendix~\ref{app4}.

		We continue the proof of Theorem~\ref{thm:weaklift}. We have computed the action of the weak lifting $\tilde\P$ of $\pow$ to $\EM\S$ on objects: given an $\S$-algebra $(X,a)$, $\tilde\P (X,a)=(\ConvPow X {\muS_X},\alpha_a)$ where $\alpha_a$ is defined in~\eqref{eq:alphaa}. As for morphisms, technically for $f \colon (X,a) \to (X',a')$ in $\EM \S$, $\tilde\P(f)$ is defined as 
		\[
		\tilde\P(f)(\mathcal A) = \convclos{\P(f)(\mathcal A)}{a'} \quad \text{for all $\mathcal A \in \ConvPow X a$,}
		\]
		however it is not difficult to see that $\P(f)(\mathcal A)$ is convex in $(X',a')$ using the fact that $f$ is a morphism of $\S$-algebras and that $\mathcal A$ is convex in $(X,a)$.
		
		The unit of the monad $\tilde \P$ is given, for every $(X,a)$ object of $\EM{\mon S}$, as the unique morphism $\eta^{\tilde \P}_{(X,a)} \colon (X,a) \to \tilde \P (X,a)$ that makes the two triangles of~(\ref{eqn:weak lifting diagrams iota}) and~(\ref{eqn:weak lifting diagrams pi}) commute, which are in our case:
		\[
		\begin{tikzcd}
			X \ar[r,"\eta^{\P}_X"] \ar[d,"\U{\mon S} (\eta^{\tilde \P}_{(X,a)})"'] & \pow X \\
			\ConvPow X a \ar[ur,"\iota_{(X,a)}"',hookrightarrow]
		\end{tikzcd}
		\qquad
		\begin{tikzcd}[column sep=4em]
			X \ar[r,"\U{\mon S}(\eta^{\tilde \P}_{(X,a)})"] \ar[d,"\eta^{\P}_X"'] & \ConvPow X a \\
			\pow X \ar[ur,"\convclos{(-)}{a}"']
		\end{tikzcd}
		\]
		By definition of $\iota$ and $\eta^{\P}$, the only arrow that makes the left triangle above commutative is necessarily   
		\[
		\begin{tikzcd}[row sep=0em]
			(X,a) \ar[r,"\eta^{\tilde \P}_{(X,a)}"] & (\ConvPow X a,\alpha_a) \\
			x \ar[r,|->] & \{x\}
		\end{tikzcd}
		\]
		and this makes also the right triangle commutative, since $\convclos{\{x\}} a = \{ x \}$.
		Similarly, the multiplication $\mu^{\tilde \P}$ is defined, for every $\mon S$-algebra $(X,a)$, as the unique morphism making the two rectangles of~(\ref{eqn:weak lifting diagrams iota}) and~(\ref{eqn:weak lifting diagrams pi}) commute. These are in our case:
		\[
		\begin{tikzcd}[column sep=1.2em]
			\ConvPow {\bigl(\ConvPow X a\bigr)} {\alpha_a} \ar[r,hookrightarrow,"\iota"] \ar[d,"\mu^{\tilde \P}_{(X,a)}"'] & \pow(\ConvPow X a) \ar[r,"\pow \iota",hookrightarrow] & \pow \pow X \ar[d,"\mu^\P_X"] \\
			\ConvPow X a \ar[rr,hookrightarrow,"\iota"] & & \pow X
		\end{tikzcd}
		\quad 
		\begin{tikzcd}[column sep=1.5em]
			\pow \pow X \ar[r,"\pow \convclos{(-)} a"] \ar[d,"\mu^{\P}_X"'] 
			& \pow(\ConvPow X a) \ar[r,"\convclos{(-)}{\alpha_a}"] & \ConvPow {(\ConvPow X a)} {\alpha_a}  \ar[d,"{\mu^{\tilde \P}_{(X,a)}}"] \\
			\pow X \ar[rr,"\convclos{(-)} a"] & & \ConvPow X a
		\end{tikzcd}
		\]
		
		By definition of $\iota$ and $\mu^{\P}$, the only arrow that makes the left rectangle above commutative is necessarily  
		\[
		\begin{tikzcd}[row sep=0em]
			\ConvPow {(\ConvPow X a)} {\alpha_a} \ar[r,"\mu^{\tilde \P}_{(X,a)}"] & \ConvPow X a \\
			\mathcal A \ar[r,|->] & \bigcup_{A \in \mathcal A} A
		\end{tikzcd}
		\]
		for all $(X,a) \in \EM{\mon S}$.
		Next we show why this morphism is well defined, namely why $\mu^{\tilde \P}_{(X,a)}(\mathcal A)$ is in $\ConvPow X a$, while in Proposition~\ref{prop:rectangle lifting mu commutes} (in Appendix~\ref{app4}) we prove how it makes the right rectangle commutative. This will conclude the proof of Theorem~\ref{thm:weaklift}.

		\begin{prop}\label{prop:union of convex family of convex subsets is convex}
			Let $(X,a)$ be a $\mon S$-algebra and $\mathcal A \in \ConvPow {\bigl(\ConvPow X a\bigr)} {\alpha_a}$. Then $\bigcup_{A \in \mathcal A} A$ is convex in $(X,a)$.
		\end{prop}
		\begin{proof}
			We want to prove the following equality:
			\[
			\bigcup_{A \in \mathcal A} A = \{a(\phi) \mid \phi \in \mon S X, \sum_{x \in X} \phi(x)=1, \supp \phi \subseteq \bigcup_{A \in \mathcal A} A\}.
			\]
			The left-to-right inclusion is trivial: for all $A \in \mathcal A$ and $x \in A$, $x=a(\Delta_x)$. Vice versa, let $\phi \in \mon S X$ be such that $\sum_{x \in X} \phi(x)=1$ and $\supp \phi \subseteq \bigcup_{A \in \mathcal A} A$. We want to find an element $A \in \mathcal A$ such that $a(\phi) \in A$. Suppose
			\[
			\supp \phi = \{x_1,\dots,x_n\}.
			\]
			Then for all $i \in \natset n$ there exists $A_i \in \mathcal A$ such that $x_i \in A_i$. Notice that if $i \ne j$ then it is not necessarily the case that $A_i \ne A_j$. Define $\Phi \in \mon S (\ConvPow X a)$ as
			\[
			\Phi(B) = \sum_{ \substack{i \in \natset n \\ B=A_i} } \phi(x_i)
			\]
			for all $B \in \ConvPow X a$.
			Then $\supp \Phi = \{A_i \mid i \in \natset n\}$ is finite and
			\[
			\sum_{B \in \ConvPow X a} \Phi(B) = \sum_{B \in \ConvPow X a} \sum_{ \substack{i \in \natset n \\ B=A_i} } \phi(x_i) = \sum_{i \in \natset n} \phi(x_i) = 1.
			\]
			Since $\mathcal A$ is convex in $(\ConvPow X a,\alpha_a)$, we have that $\alpha_a(\Phi) \in \mathcal A$. This is going to be our desired $A$: we shall prove that $a(\phi) \in \alpha_a(\Phi)$.
			
			We have $\alpha_a(\Phi)=\pow a (\delta_X(\mon S (i)(\Phi)))$, hence, if we prove that $\phi \in \delta_X(\mon S (i)(\Phi))$, we have finished. To this end, recall the characterisation of $\delta_X$ using elements of $\mon S (\ni_X)$, for $\ni_X {} \subseteq \pow X \times X$:
			\begin{multline*}
				\delta_X(\mon S (i) (\Phi)) = \\
				\left\{a(\chi) \mid \chi \in \mon S X \ldotp \exists \psi \in \mon S (\ni_X) \ldotp 
				\begin{cases}
					\forall B \subseteq X \ldotp \mon S (i) (\Phi)(B) = \sum_{x \in B} \psi(B,x) \\
					\forall x \in X \ldotp \chi(x) = \sum_{B \ni x} \psi(B,x)
				\end{cases}
				\right\}.
			\end{multline*}
			Define, for all $(B,x) \in \pow X \times X$ such that $B \ni x$:
			\[
			\psi(B,x) = \sum_{ \substack{i \in \natset n \\ (B,x)=(A_i,x_i)}} \phi(x_i).
			\]
			Then $\supp \psi = \{(A_i,x_i) \mid i \in \natset n\}$ is finite, for all $B \subseteq X$
			\[
			\sum_{x \in B} \psi(B,x) = \sum_{x \in B} \sum_{ \substack{i \in \natset n \\ (B,x)=(A_i,x_i)}} \phi(x_i) = \sum_{\substack{i \in \natset n \\ B=A_i}} \phi(x_i) = \mon S (i)(\Phi)(B)
			\]
			and, for all $x \in X$,
			\[
			\sum_{B \ni x} \psi(B,x) = \sum_{B \ni x} \sum_{ \substack{i \in \natset n \\ (B,x)=(A_i,x_i)}} \phi(x_i) = 
			\begin{cases}
				0 & \forall i \in \natset n \ldotp x \ne x_i \\
				\phi(x_i) & \exists (!) i \in \natset n \ldotp x=x_i
			\end{cases}
			= \phi(x)
			\]
			where if there is $i$ such that $x=x_i$, then such $i$ is unique due to the fact that the $x_j$'s are distinct. This proves that $\phi \in \delta_X(\mon S (i)(\Phi))$.
		\end{proof}

		\section{The Composite Monad: an Algebraic Presentation}\label{sec:the composite monad}

		We can now compose the two monads $\P$ and $\mon S$ by considering the monad arising from the composition of the following two adjunctions: 
		\[\begin{tikzcd}[column sep=3em]
			{\Set} \ar[r,bend left=30,"{\F{\mon S}}"{name=F1},pos=0.55] & {\EM{\mon S}} \ar[l,bend left=30,"\U{\mon S}"{name=U1},pos=0.45] \ar[r,bend left=30,"\F{\tilde \P}"{name=F2},pos=0.55] & \EM{\tilde \P} \ar[l,bend left=30,"\U{\tilde \P}"{name=U2},pos=0.5]
			\ar[from=F1,to=U1,"\bot"{description},draw=none]
			\ar[from=F2,to=U2,"\bot"{description},draw=none]
		\end{tikzcd}
		\]
		Direct calculations show that the resulting endofunctor on $\Set$, which we call $\convpowS{}$,
		maps a set $X$ and a function $f\colon X \to Y$ into, respectively, 
		\begin{equation}\label{composite endo}
			\convpowS{X} = \ConvPow{(\mon S X)} {\muS_X} \qquad \text{and} \qquad\convpowS(f)(\mathcal A) = \{\mon{S}(f)(\Phi) \mid \Phi \in \mathcal{A}\}\end{equation}
		for all $\mathcal{A}\in \convpowS{X}$. For all sets $X$, $\eta^\convpowS_X \colon X \to \convpowS{X}$ and $\mu^{\convpowS{}}_X \colon \convpowS{}\convpowS{X} \to \convpowS{X}$ are defined as
		\begin{equation}\label{composite unit mult}
			\eta^\convpowS_X(x)= \{ \Delta_x \} \qquad \text{and} \qquad \mu^{\convpowS}_X(\mathscr A )= \bigcup\limits_{\Omega \in \mathscr A} \alpha_{\muS_X} (\Omega)
		\end{equation}
		for all $x\in X$ and $\mathscr A \in \convpowS{}\convpowS{X}$.
		
		\begin{thm}\label{thm:composite monad}
			Let $S$ be a positive semifield. Then the canonical weak distributive law $\delta \colon \mon S \P \to \P \mon S$ given in Theorem~\ref{thm:delta for positive refinable semifields} induces a monad $\convpowS$ on $\Set$ with endofunctor, unit and multiplication defined as in \eqref{composite endo} and \eqref{composite unit mult}.
		\end{thm}
		
		\begin{exa}
			Let $S=\Bool$. Then the monad $\S$ is isomorphic to the finite powerset monad $\Pf$. In this case, a convex subset of $\Pf X$ with respect to the structure map $\mu^{\Pf}_X$ is a family $\mathcal A$ of finite subsets of $X$ that is closed under binary unions (see Example~\ref{ex:S=Bool, convex subsets are directed}), in other words, it is a sub-semilattice of $\Pf X$:
			\[
			\ConvPow {(\S X)} {\muS_X} = \{ \mathcal A \in \P(\Pf X) \mid \forall A,B \in \mathcal A \ldotp A \cup B \in \mathcal A \}.
			\]
			Here a $\Phi \in \mon S \bigl( \ConvPow{(\mon S X)}{\muS_X} \bigr) $ is a finite set of the form $\{ \mathcal A_1,\dots,\mathcal A_n \}$ where each $\mathcal A_i$ is a (possibly infinite) subset of $\Pf X$ that is closed under binary unions: if $A,B \in \mathcal A_i$ then $A \cup B \in \mathcal A_i$. The set $\choice \Phi$ is now
			$
			\choice \Phi = \{\; \{ A_1,\dots,A_n \} \mid A_i \in \mathcal A_i  \; \} \subseteq \Pf (\Pf X)
			$
			hence 
			\[
			\alpha_{\muS_X} (\{ \mathcal A_1,\dots,\mathcal A_n \}) = \bigl\{ A_1 \cup \dots \cup A_n \mid A_i \in \mathcal A_i \bigr\} \subseteq \Pf X.
			\] 
			Thus an $\mathscr A \in \convpow \Pf \convpow \Pf X$ is a (possibly infinite) collections of $\Phi$'s as above, closed under binary unions, in the sense that if $\{\mathcal A_1,\dots,\mathcal A_n\}$ and $\{\mathcal B_1,\dots,\mathcal B_m\}$ are in $\mathscr A$, then $\{ \mathcal A_1,\dots,\mathcal A_n,\mathcal B_1,\dots,\mathcal B_m \} \in \mathscr A$. $\mu(\mathscr A$) is then the union of all the $\alpha_{\muS_X}(\Phi)$ with $\Phi \in \mathscr A$: it is a subset of $\Pf X$. The fact that it is actually a \emph{convex} subset of $\Pf X$ (\ie it is closed under binary unions) is ensured by Theorem~\ref{thm:composite monad} and it is a consequence of the fact that $\mathscr A$ is convex itself.
		\end{exa}
		
		Recall from Remark~\ref{remarkJacobs} that the monad $\mathcal C \colon \EM{\mon S} \to \EM{\mon S}$ from~\cite{jacobs_coalgebraic_2008} coincides with our lifting $\tilde \P$ modulo the absence of the empty set. The same happens for the composite monad, which is named $\mathcal{CM}$ in~\cite{jacobs_coalgebraic_2008}. 
		The absence of $\emptyset$ in $\mathcal{CM}$ turns out to be rather problematic for Jacobs. Indeed, in order to use the standard framework of coalgebraic trace semantics~\cite{hasuo_generic_2006}, one would need the Kleisli category $\Kl{\mathcal{CM}}$ to be enriched over $\Cppo$, the category of $\omega$-complete partial orders with \emph{bottom} and continuous functions. $\Kl{\mathcal{CM}}$ is not $\Cppo$-enriched since there is no bottom element in $\mathcal{CM}(X)$. Instead, in $\convpowS{X}$ the bottom is exactly the empty set; moreover, $\Kl{\convpowS{}}$ enjoys the properties required by~\cite{hasuo_generic_2006}.
		
		\begin{thm}\label{thm:Kleisli is CPPO enriched}
			The category $\Kl{\convpowS}$ is enriched over $\Cppo$ and satisfies the left-strictness condition: for all $f \colon X \to \convpowS Y$ and $Z$ set, $\bot_{Y,Z} \circ f = \bot_{X,Z}$.
		\end{thm}
		It is immediate that every homset in $\Kl{\convpowS}$ carries a complete partial order. Showing that composition of arrows in $\Kl{\convpowS}$ preserves joins (of $\omega$-chains) requires more work: the proof, which can be found in \sect~\ref{subsec:proof of Kleisli category is CPPO enriched}, crucially relies on the algebraic theory presenting the monad $\convpowS$, illustrated next.

		\subsection{An Algebraic Presentation.}\label{subsec:algebraic presentation}
		Recall that an \emph{algebraic theory} is a pair $\mathcal T=(\Sigma, E)$ where $\Sigma$ is a \emph{signature}, whose elements are called \emph{operations}, to each of which is assigned a cardinal number called its \emph{arity}, while $E$ is a class of \emph{equations} between $\Sigma$-terms. An \emph{algebra} for the theory $\mathcal T$ is a set $A$ together with, for each operation $o$ of arity $\kappa$ in $\Sigma$, a function $o_A \colon A^\kappa \to A$ satisfying the equations of $E$. A \emph{homomorphism} of algebras is a function $f \colon A \to B$ respecting the operations of $\Sigma$ in their realisations in $A$ and $B$. Algebras and homomorphisms of an algebraic theory $\mathcal T$ form a category $\Alg(\mathcal T)$.
		
		\begin{defi}
			Let $M$ be a monad on $\Set$, and $\mathcal T$ an algebraic theory. We say that $\mathcal T$ \emph{presents} $M$ if and only if $\EM{M}$ and $\Alg(\mathcal T)$ are isomorphic.
		\end{defi}
		
		Left $S$-semimodules are algebras for the theory $\LSM= (\Sigma_{\LSM}, E_{\LSM})$ where
		$\Sigma_{\LSM} = \{ +, 0 \} \cup \{ \lambda \cdot {} \mid \lambda \in S \}$ and $E_{\LSM}$ is the set of axioms in Table~\ref{tab:axiomsinitial}, top right. As already mentioned in Section~\ref{sec:monad},  left $S$-semimodules are exactly $\mon S$-algebras and morphisms of $S$-semimodules coincide with those of $\mon S$-algebras. Thus, the theory $\LSM$ presents the monad~$\mon S$.
		
		Similarly, semilattices are algebras for the theory $\SL=(\Sigma_{\SL}, E_{\SL})$ where $\Sigma_{\SL} = \{\sqcup, \bot\}$ and $E_{\SL}$ is the set of axioms in Table~\ref{tab:axiomsinitial}, top left. It is well known that semilattices are algebras for the \emph{finite} powerset monad. Actually, this monad is presented by $\SL$. 
		In order to present the full powerset monad $\P$ we need to take joins of arbitrary arity. A \emph{complete semilattice} is a set $X$ equipped with joins $\Sup_{x\in A} x$ for all--not necessarily finite--$A\subseteq X$. Formally the (infinitary) theory of \emph{complete semilattices} is given as $\CSL = (\Sigma_{\CSL}, E_{\CSL})$ where 
		$\Sigma_{\CSL} = \{ \Sup_{I} \mid I \text{ set} \}$
		and $E_{\CSL}$ is the set of axioms displayed in the top of Table~\ref{table:otheraxioms} (for a detailed treatment of infinitary algebraic theories see, for example,~\cite{manes_algebraic_1976}).

		We can now illustrate the theory $(\Sigma,E)$ presenting the composite monad $\convpowS$: the operations in $\Sigma$ are exactly those of complete semilattices and $S$-semimodules, while the axioms are those of complete semilattices and $S$-semi\-modules together with the set $E_{\D}$ of \emph{distributivity} axioms consisting of $(D1)$ and $(D2)$ in the bottom of Table \ref{table:otheraxioms}.
		In short, $\Sigma = \Sigma_{\CSL} \cup \Sigma_{\LSM} \text{ and } E= E_{\CSL} \cup E_{\LSM} \cup E_{\D}$. \begin{thm}\label{thm:pres}
			The monad $\convpowS$ is presented by the algebraic theory $(\Sigma,E)$.
		\end{thm}
		
		The proof of Theorem~\ref{thm:pres} crucially relies on the fact that $\convpowS$ is obtained by composing $\P$ and $\mon S$ via $\delta$, as we show in the rest of this current Section~\ref{subsec:algebraic presentation}. We shall use the notion of $\delta$-algebra below, originally defined by~\cite{beck_distributive_1969} in the context of ordinary distributive laws but that can be adapted with no changes in the weak setting as well. Here we report the definition as phrased in~\cite{garner_vietoris_2020}.
		
		\begin{defi}
			Let $T$ and $S$ be monads on $\C$ and $\delta \colon TS \to ST$ be a (weak) distributive law. A $\delta$-\emph{algebra} is an object of $\C$ endowed with a $T$-algebra structure $t \colon TX \to X$ and an $S$-algebra structure $s \colon SX \to X$ such that the following pentagon
			\[
			\begin{tikzcd}
				TS(X) \ar[rr,"\delta_X"] \ar[d,"Ts"'] & & ST(X) \ar[d,"St"] \\
				T(X) \ar[dr,"t"'] &  & S(X) \ar[dl,"s"] \\
				& X
			\end{tikzcd}
			\]
			commutes. The category $\EM{\delta}$ is the full subcategory of $\EM{S} \times_\C \EM{T}$, result of the pullback
			\[
			\begin{tikzcd}
				\EM{S} \times_\C \EM{T} \ar[d] \ar[r] & \EM{T} \ar[d,"\U{T}"] \\
				\EM{S} \ar[r,"\U{S}"] & \C
			\end{tikzcd}
			\]
			whose objects are $\delta$-algebras and morphisms are thus maps in $\C$ which are simultaneously morphisms of $S$- and $T$-algebras.
		\end{defi}
		
		The following Proposition, relating $\delta$-{algebras} with the Eilenberg-Moore algebras for $\convpowS$, appears in~\cite[Lemma 14]{garner_vietoris_2020} and was proved in a more general context in~\cite[Proposition 3.7]{bohm_weak_2010}.
		
		\begin{prop}\label{prop:delta algebras isomorphic to algebras composite monad}
			Given a weak distributive law $\delta \colon TS \to ST$, the categories $\EM{\delta}$ of $\delta$-algebras and $\EM{\widetilde{ST}}$ of Eilenberg-Moore algebras of the composite monad $\widetilde{ST}$ (arising from $\delta$) are canonically isomorphic.
		\end{prop}
		
		We can therefore use to our advantage Proposition~\ref{prop:delta algebras isomorphic to algebras composite monad} to individuate the algebraic theory presenting the monad $\widetilde{\P \mon S}$, which we called $\convpowS$: rather than studying the Eilenberg-Moore algebras for the quite complicated composite monad, we analyse the isomorphic category of $\delta$-algebras instead. In order to prove Theorem~\ref{thm:pres}, we show that $\EM{\delta}$ is isomorphic to the category $\Alg(\Sigma,E)$ of $(\Sigma,E)$-algebras (Theorem~\ref{thm:delta algebras isomorphic to (Sigma,E)-algebras}) and then use Proposition~\ref{prop:delta algebras isomorphic to algebras composite monad} to conclude.
		
		First we prove that every $\delta$-algebra gives rise to an algebra for the theory $(\Sigma,E)$, \ie a  complete semilattice and $S$-semimodule satisfying the equations in $E_\D$. For $\delta \colon \mon S \P \to \P \mon S$ given in~\eqref{eqn:delta for positive refinable semifields}, we have that a $\delta$-algebra is a set $X$ together with a $\mon S$-algebra structure $a \colon \mon S X \to X$ and a $\P$-algebra-structure $b \colon \pow X \to X$ such that pentagon
		\begin{equation}\label{eqn:distributivity pentagon for delta-algebras}
			\begin{tikzcd}
				\mon S \pow X \ar[rr,"\delta_X"] \ar[d,"\mon S b"'] & & \pow \mon S X \ar[d,"\pow a"] \\
				\mon S X \ar[dr,"a"'] & & \pow X \ar[dl,"b"] \\
				& X 
			\end{tikzcd}
		\end{equation}
		commutes. A morphism of $\delta$-algebras is a morphism of $\P$- and $\S$-algebras simultaneously by definition. Hence, by defining $+^a$, $\lambda \cdot^a$ and $0^a$ as in~\eqref{eqn:semimodule associated to S-algebra}, we have that $X$ is a $S$-left-semimodule; moreover, if we define for each set $I$
		\begin{equation}\label{eq:defSup}
			\Sup_{i \in I}^b x_i = b(\{x_i \mid i \in I \})
		\end{equation}
		then $X$ is also a complete semilattice. This means that $X$ satisfies all the equations in $E_\CSL$ and $E_\LSM$; moreover, every morphism of $\delta$-algebras preserves all the operations in $\Sigma$. The commutativity of pentagon~\eqref{eqn:distributivity pentagon for delta-algebras} will instead ensure that the equations in $E_\D$, that is $(D1)$ and $(D2)$ in Table~\ref{table:otheraxioms}, are satisfied, as we prove in the following. For $A \subseteq X$, we shall sometimes write $\Sup^b A$ for $\Sup_{x \in A}^b x$. 
		
		First we need a technical result that shows how joins and convex closures interact in $\delta$-algebras. 
		
		\begin{prop}\label{prop:sup A = sup convclos A in delta algebras}
			Let $(X,a \colon \mon S X \to X, b \colon \pow X \to X)$ be a $\delta$-algebra.
			For each $A \subseteq X$, 
			\[
			\Sup_{x \in A}^b x = \Sup_{x \in \convclos A a}^b x.
			\]	
		\end{prop}
		\begin{proof}
			Let $\Phi = \etaS_{\pow X}(A) = \Delta_A \in \mon S \pow X$. We show that the two legs of the commutative diagram~\eqref{eqn:distributivity pentagon for delta-algebras}, that is $\S b (\Phi)$ and $b(\pow a (\delta_X(\Phi)))$, coincide with the two sides of the equation in the statement.
			Indeed we have
			\[
			a\bigl( \mon S b (\Phi) \bigr) = a\bigl( \mon S b (\Delta_A) \bigr) = a\bigl( \Delta_{\Sup^b A} \bigr) = \Sup^b A 
			\]
			while
			\[
			\delta_X(\Delta_A) = \{ \phi \in \mon S X \mid \supp \phi \subseteq A, \sum_{x \in X} \phi(x)=1 \}
			\]
			because of Lemma~\ref{lemma:delta of Dirac}, hence
			\[
			b \Bigl(\pow a \bigl( \delta_X(\Delta_A) \bigr)\Bigr) = b (\convclos A a) = \Sup^b \convclos A a
			\] 
			where recall that $\convclos A a = \{ a(\phi) \mid \phi \in \S X, \supp \phi \subseteq A, \sum\limits_{x \in X} \phi(x) = 1 \}$ from Definition~\ref{def:convex closure S-ALGEBRA}.
		\end{proof}

		We can then prove that every $\delta$-algebra satisfies the equations in $E_\D$.
		
		\begin{thm}\label{thm:distributivity property of delta-algebras}
			Let $(X,a \colon \mon S X \to X, b \colon \pow X \to X)$ be a $\delta$-algebra. Then for all $A,B \subseteq X$ and for all $\lambda \in S \setminus \{ 0_S\}$:
			\[
			\lambda \cdot^a \Sup_{x \in A}^b x = \Sup_{x \in A}^b \lambda \cdot^a x, \qquad \Sup_{x \in A}^b x +^a \Sup_{y \in B}^b y  = \Sup_{(x,y) \in A \times B}^b  x +^a  y.
			\]
		\end{thm}
		\begin{proof} 
			Let
			$
			\Phi = ( 
			A \mapsto \lambda) \in \S\P X
			$.
			We prove that the images along the two legs of pentagon~\eqref{eqn:distributivity pentagon for delta-algebras} of $\Phi$ coincide with the two sides of the first equation in the statement. 
			It holds that 
			\begin{align*}
				a(\mon S(b)(\Phi)) &= a(x \mapsto \sum_{U \in b^{-1}\{x\}} \Phi(U) ) \\ &= a(b(A) \mapsto \lambda) \\ &= a(\Sup_{x\in A}^b x \mapsto \lambda) & \text{Definition of $\Sup^b$ \eqref{eq:defSup}} \\
				&= \lambda \cdot^a \Sup_{x\in A}^b x & \text{Definition of $\lambda \cdot^a$ \eqref{eqn:semimodule associated to S-algebra}}
			\end{align*}

			On the other hand we have that 
			\begin{align*}
				b\Bigl(\pow a \bigl(\delta_X (\Phi)\bigr)\Bigr) &= \Sup^b \pow a \bigl( \convclos {\choice \Phi}{\muS_X} \bigr) & \text{Theorem~\ref{thm:delta for positive refinable semifields}}\\
				&= \Sup^b \convclos {\pow a \bigl( \choice \Phi \bigr)} {a} & \text{Proposition~\ref{prop:image along a of convex closures}}\\
				&= \Sup^b \pow a \bigl( \choice \Phi \bigr) & \text{Proposition~\ref{prop:sup A = sup convclos A in delta algebras}}
			\end{align*}
			
			Following the discussion after Theorem \ref{thm:weaklift}, $\choice{\Phi} = \{(x \mapsto \lambda) \mid x\in A \}$ if $\lambda \neq 0$. Therefore, if $\lambda \neq 0$, then $\pow a \bigl( \choice \Phi \bigr) = 	\{a(x \mapsto \lambda)  \, \mid \; x \in A\}$ which by \eqref{eqn:semimodule associated to S-algebra} is exactly $\{ \lambda \cdot^a x  \, \mid \; x \in A\}$.

			Therefore 
			\[	b\Bigl(\pow a \bigl(\delta_X (\Phi)\bigr)\Bigr) = 
			\Sup^b_{x\in A} \lambda \cdot^a x. \]
			We then conclude by using the commutativity of diagram~(\ref{eqn:distributivity pentagon for delta-algebras}). Using the function $\Phi' = (A \mapsto 1,\, B \mapsto 1) \in \S\P X$ and a similar argument, one shows the second equation as well.
		\end{proof}
		
		Vice versa, we want to prove that every algebra for the theory $(\Sigma,E)$ is also a $\delta$-algebra. To this end, let $(X,+,\lambda \cdot {},0_X,\Sup_I)$ be a $(\Sigma,E)$-algebra. Then $(X,+,\lambda \cdot {}, 0_X)$ is a $S$-left-semimodule, $(X,\Sup_I)$ is a complete sup-semilattice. By defining
		\[
		\begin{tikzcd}[row sep=0em]
			\mon S X \ar[r,"a"] & X \\
			f \ar[r,|->] & \sum\limits_{x \in \supp f} f(x) \cdot x
		\end{tikzcd}
		\qquad
		\begin{tikzcd}[row sep=0em]
			\pow X \ar[r,"b"] & X \\
			A \ar[r,|->] & \Sup\limits_{x \in A} x
		\end{tikzcd}
		\]
		we have that $(X,a) \in \EM{\mon S}$ and $(X,b) \in \EM\P$. We now have to check that pentagon~(\ref{eqn:distributivity pentagon for delta-algebras}) commutes. Given $A \subseteq X$, we define its convex closure $\convclos A{}$ in the usual way: 
		\[
		\convclos{A}{} =\left\{ \sum_{i=1}^n \lambda_i \cdot a_i \mid n \in \N,\, a_i \in A,\,  \sum_{i=1}^n \lambda_i =1 \right\} \subseteq X.
		\]
		Also, we shall denote by $\Sup A$ the element $b(A) \in X$ for any $A \subseteq X$. Then pentagon~(\ref{eqn:distributivity pentagon for delta-algebras}) commutes if and only if for all $\Phi \in \mon S \pow X$:
		\[
		\sum_{A \in \supp \Phi} \!\! \Phi(A) \cdot \Sup A = \Sup \convclos{ \Bigl\{ \sum_{A \in \supp \Phi}\!\! \Phi(A) \cdot u(A) \mid u \colon \supp \Phi \to X \ldotp \forall A  \ldotp u(A) \in A \Bigr\} } {}
		\]
		where the left-hand side is $a(\S(b)(\Phi))$ and the right-hand side is $b(\P(a)(\delta_X(\Phi)))$.
		
		In the following lemma we prove that if $X$ is a $(\Sigma,E)$-algebra and $A \subseteq X$, then $\Sup A = \Sup \convclos A {}$. (This is a generalisation of the derivation in Table~\ref{tab:derivation} for arbitrary joins.) Using this fact and the axioms $(D1)$, $(D2)$ of Table~\ref{table:otheraxioms} we will have shown the commutativity of pentagon~(\ref{eqn:distributivity pentagon for delta-algebras}).
		
		\begin{lem}
			Let $(X,+,\lambda \cdot {},0_X,\Sup_I)$ be a $(\Sigma,E)$-algebra. Then for all $A \subseteq X$
			\[
			\Sup A = \Sup \convclos A {}.
			\]
		\end{lem}
		\begin{proof}
			Let $\lambda_1,\dots,\lambda_n \in S$ with $\sum_{i=1}^n \lambda_i = 1$. Then
			\[
			\Sup A = 1 \cdot \Sup A = (\sum_{i=1}^n \lambda_i) \cdot \Sup A = \sum_{i=1}^n \bigl( \lambda_i \cdot \Sup A \bigr) = \Sup \Bigl\{ \sum_{i=1}^n \lambda_i \cdot a_i \mid (a_1,\dots,a_n) \in A^n \Bigr\}
			\]
			because of the axioms of $S$-semimodule and the distributivity $(D1)$, $(D2)$ in Table~\ref{table:otheraxioms}.
			Let $\le$ be the partial order determined by the complete semilattice structure of $X$. Then we have that for all $n \in \N$, for all $\lambda_1,\dots,\lambda_n \in S$ such that $\sum_{i=1}^n \lambda_i = 1$ and for all $(a_1,\dots,a_n) \in A^n$:
			\[
			\sum_{i=1}^n \lambda_i \cdot a_i \le \Sup \Bigl\{ \sum_{i=1}^n \lambda_i \cdot b_i \mid (b_1,\dots,b_n) \in A^n \Bigr\} = \Sup A
			\]
			hence
			\[
			\Sup \convclos A {} = \Sup \Bigl\{ \sum_{i=1}^n \lambda_i \cdot a_i \mid n \in \N,\, a_i \in A, \sum_{i=1}^n \lambda_i = 1 \Bigr\} \le \Sup A
			\]
			while the other inequality is trivial because $A \subseteq \convclos A {}$.	\end{proof}
		
		Finally, again a morphism of $(\Sigma,E)$-algebras is a morphism respecting all the operations of $\Sigma$, which means of $\Sigma_{\LSM}$ (thus is a morphism of $\EM{\mon S}$) and of $\Sigma_{\CSL}$ (thus is a morphism of $\EM{\P}$) at the same time. We have therefore proved the following theorem.
		
		\begin{thm}\label{thm:delta algebras isomorphic to (Sigma,E)-algebras}
			The category $\EM{\delta}$ of $\delta$-algebras is isomorphic to the category  $\Alg(\Sigma,E)$ of $(\Sigma,E)$-algebras.
		\end{thm}
		
		Since $\EM{\delta}$ is canonically isomorphic to $\EM{\convpowS}$ by Proposition~\ref{prop:delta algebras isomorphic to algebras composite monad}, we have proved Theorem~\ref{thm:pres}.
		We conclude this section with a couple of examples.
		
		\begin{exa}[\cf~{\cite[\sect 4.3]{garner_vietoris_2020}}] 
			Let $S=\Bool=\{0,1\}$, where addition is $\vee$ and multiplication is $\wedge$. Then a $\Bool$-semimodule is a commutative monoid $(X,\ast,e)$ where $\ast$ is idempotent (see Example~\ref{ex:S=Bool, convex subsets are directed}). A $\delta$-algebra, in this case, is therefore a complete lattice $X$ with a commutative, associative, unitary and idempotent operation $\ast$ such that 
			\[
			{\Sup A} \ast {\Sup B} = \Sup \{ a \ast b \mid a \in A,\, b \in B \}.
			\]
			Thus $\delta$-algebras are exactly commutative, unitary and idempotent quantales.
		\end{exa}

		\begin{exa}\label{ex:intervals}
			Let $S$ be $\Rp$ and let $[a,b]$ with $a,b\in \Rp$ denote the set $\{x\in \Rp \mid a\leq x \leq b\}$ and $[a,\infty)$ the set $\{x\in \Rp \mid a\leq x \}$.
			For $1=\{x\}$, $\convpowS(1) = \{\emptyset\} \cup \{[a,b] \,\mid \, a,b\in \Rp\} \cup \{[a,+\infty) \, \mid \, a\in \Rp\}$. The $\convpowS$-algebra $\mu^{\convpowS}_1 \colon \convpowS \convpowS 1 \to \convpowS 1$ induces a $\delta$-algebra where the structure of complete lattice is given as\footnote{For the sake of brevity, we are ignoring the case where some $A_i=\emptyset$.}
			\[\Sup_{i\in I}A_i = \begin{cases}
				[\inf_{i\in I}, a_i, \sup_{i\in I} b_i] & \text{if, for all }i\in I, \; A_i=[a_i, b_i] \wedge \sup_{i\in I} b_i\in \Rp\\
				[\inf_{i\in I}a_i, \infty)  & \text{otherwise}
			\end{cases} \\
			\]
			The $\Rp$-semimodule is as expected, \eg $[a_1,b_1]+[a_2,b_2]=[a_1+a_2,b_1+b_2]$.
		\end{exa}
		
		\subsection{Proof of Theorem~\ref{thm:Kleisli is CPPO enriched}} \label{subsec:proof of Kleisli category is CPPO enriched} 
		
		We need to show that the Kleisli category $\convpowS$ is $\Cppo$-enriched and satisfies the left-strictness condition. The algebraic presentation of the composite monad $\convpowS$ provided by Theorem~\ref{thm:pres} is going to be useful. Let us begin with an analysis of some basic properties of $\Kl{\convpowS}$.

		First of all, since $(\convpowS X,\mu^{\convpowS}_X)$ is a $\convpowS$-algebra, we have that $\convpowS X$ has a structure of semimodule and complete semilattice given, via the canonical isomorphism $\EM{\convpowS} \to \EM{\delta}$, by:
		\begin{gather*}
			\mathcal A + \mathcal B = \{ \phi + \psi \mid \phi \in \mathcal A,\, \psi \in \mathcal B \} \\
			\lambda \cdot \mathcal A = \{ \lambda \cdot \phi \mid \phi \in \mathcal A \} \\
			\Sup_{i \in I} \mathcal A_i = \convclos { \bigcup_{i \in I} \mathcal A_i } {}
		\end{gather*}
		for all $\mathcal A, \mathcal B \in \convpowS X$ and $\lambda \in S$, $\lambda \ne 0$.
		
		This means that each homset $\Kl{\convpowS}(X,Y)$ of functions $f \colon X \to \convpowS(Y)$, which is partially ordered point-wise, inherits the structure of complete semilattice from $\convpowS Y$, in particular its bottom element $\bot_{X,Y}$ is the constant function mapping $x$ to $\emptyset$ for all $x \in X$. Given a function $g \colon Y \to \convpowS Z$, its Kleisli extension $\Sharp g \colon \convpowS Y \to \convpowS Z$ is given by $\mu^\convpowS_Z \circ \convpowS g$, which for all $\mathcal A \in \convpowS Y$ computes:
		\[
		\Sharp g (\mathcal A) = \Sup_{\phi \in \mathcal A} \, \sum_{y \in \supp \phi} \phi(y) \cdot g(y).
		\]
		Composition of a function $f \colon X \to \convpowS Y$ with $g \colon Y \to \convpowS Z$ is therefore given as
		\[
		(g \circ f) (x) = \Sharp g (f(x))= \Sup_{\phi \in f(x)} \sum_{y \in \supp \phi } \phi(y) \cdot g(y) \in \convpowS Z.
		\]
		
		We can now prove that $\Kl{\convpowS}$ is actually enriched over \emph{directed}-complete partial orders and that it satisfies the left-strictness condition, which implies Theorem~\ref{thm:Kleisli is CPPO enriched}.
		
		\begin{thm}
			The category $\Kl{\convpowS}$ is enriched over the category of directed-complete partial orders and satisfies the left-strictness condition:
			\[
			\bot_{Y,Z} \circ f = \bot_{X,Z}
			\]
			for all $f \colon X \to \convpowS Y$ and $Z$ set.
		\end{thm}
		\begin{proof}
			In this proof we will use a generalisation of (D2) in Table~\ref{table:otheraxioms} for $n$-ary sums of arbitrary joins in $(\Sigma,E)$-algebras, given by:
			\begin{equation}\label{eq:n-ary sums of joins}
				\sum_{j=1}^n \, \Sup_{i_j \in I_j} x_{i_j} = \Sup_{(i_1,\dots,i_n) \in \prod_{j=1}^n I_j} \, \sum_{j=1}^n x_{i_j}.
			\end{equation}
			Let $f \colon X \to \convpowS Y$ and $\{ g_i \mid i \in I \}$ a directed subset of $\Kl{\convpowS}(Y,Z)$. Then:
			\begin{align*}
				\Bigl( (\Sup_{i \in I} g_i ) \circ f \Bigr) (x) &= \Sup_{\phi \in f(x)} \sum_{y \in \supp \phi} \phi(y) \cdot (\Sup_{i \in I} g_i) (y)\displaybreak[0] & \text{definition of composition} \\
				&= \Sup_{\phi \in f(x)} \sum_{y \in \supp \phi} \phi(y) \cdot \Sup_{i \in I} (g_i(y)) & \Sup_{i \in I} g_i \text{ is pointwise} \\
				&= \Sup_{\phi \in f(x)} \sum_{y \in \supp \phi} \Sup_{i \in I} \phi(y) \cdot g_i(y) & (D1) \displaybreak[0] \\
				&= \Sup_{\phi \in f(x)} \, \Sup_{(i_y) \in I^{\supp \phi}} \, \sum_{y \in \supp \phi} \phi(y) \cdot g_{i_y}(y) & \eqref{eq:n-ary sums of joins}  \displaybreak[0] \\
				&= \Sup_{\phi \in f(x)} \Sup_{i \in I} \sum_{y \in \supp \phi} \phi(y) \cdot  g_i(y) & \{g_i \mid i \in I\} \text{ is directed} \displaybreak[0] \\
				&= \Sup_{i \in I} \Sup_{\phi \in f(x)} \sum_{y \in \supp \phi} \phi(y) \cdot g_i(y) & \text{joins commute} \\ 
				&= \Sup_{i \in I} (g_i \circ f) (x) & \text{definition of composition.}
			\end{align*}

			Given, instead, an \emph{arbitrary} subset $\{f_i \mid i \in I\}$ of $\Kl{\convpowS} (X,Y)$ and a $g \colon Y \to \convpowS Z$, we have
			\begin{align*}
				\Bigl( g \circ (\Sup_{i \in I} f_i) \Bigr) (x) &= \Sharp g \Bigl(\Sup_{i \in I} f_i(x)\Bigr) & \text{definition of composition} \\
				&= \Sup_{i \in I} \Sharp g (f_i(x)) & \text{ $\Sharp g$ is a morphism of $\convpowS$-algebras } \\
				&= \Sup_{i \in I} (g \circ f) (x) & \text{definition of composition}
			\end{align*}
			($\Sharp g$ is a morphism of $\convpowS$-algebras as it is given by the universal property of the free algebra $\convpowS Y$, hence it preserves arbitrary suprema). Finally, 
			\[
			(\bot_{Y,Z} \circ f)(x) = \Sup_{\phi \in f(x)} \sum_{y \in \supp \phi} \phi(y) \cdot \emptyset = \Sup_{\phi \in f(x)} \sum_{y \in \supp \phi} \emptyset = \Sup_{\phi \in f(x)} \emptyset = \emptyset
			\]
			where $\phi(y) \cdot \emptyset = \emptyset$ because $\phi(y) \ne 0$ when $y \in \supp \phi$.
		\end{proof}

		\section{Finite Joins and Finitely Generated Convex Sets}\label{sec:finite joins and finitely generated convex sets}

		In the previous section we have used the algebraic theory $(\Sigma, E)$ to show that the category $\Kl{\convpowS}$ is enriched over $\Cppo$. However, in various situations it is convenient to consider semilattices which are not complete, but have just finitary joins. In this section we show that by restricting the theory  $(\Sigma, E)$ to finite joins we obtain a presentation for the monad of \emph{finitely generated} convex sets. We will introduce several monads and maps between them, which all appear in the following diagram, to be used as a road map through this section: \[
		\begin{tikzcd}
			T_{\Sigma,E}  \ar[r,"\phi"]  &       \convpowS \\
			T_{\Sigma',E'} \ar[u,"\theta"] \ar[r,"\bb{\cdot}'"'] \ar[ru,"\bb{\cdot}"description] & \convpowfS  \ar[u,"\iota"']
		\end{tikzcd}
		\]
		Here we recall the definition of monad map.
		\begin{defi}
			Let $(S,\eta^S,\mu^S)$ and $(T,\eta^T,\mu^T)$ be monads on a category $\C$. A monad map $\lambda \colon T \to S$ is a natural transformation such that
			\[
			\begin{tikzcd}
				\id\C \ar[r,"\eta^T"] \ar[dr,"\eta^S"'] & T \ar[d,"\lambda"] & T^2 \ar[l,"\mu^T"'] \ar[d,"\lambda \lambda"] \\
				& S & S^2 \ar[l,"\mu^S"]
			\end{tikzcd}
			\]
			commutes, where $\lambda \lambda$ is $\lambda S \circ T \lambda = S\lambda \circ \lambda T$.
		\end{defi}

		The algebraic theory $(\Sigma,E)$ described in Theorem~\ref{thm:pres} determines a monad $T_{\Sigma,E}\colon \Set \to \Set$ where, for any set $X$, $T_{\Sigma,E}(X)$ is the set of all $\Sigma$-terms with variables in $X$ quotiented by the equations in $E$. Recall that a $\Sigma$-term with variables in $X$ is defined inductively as:
		\begin{itemize}
			\item every variable $x$ is a $\Sigma$-term,
			\item if $o$ is an operation in $\Sigma$ with arity $\kappa$ and $t_1,\dots,t_\kappa$ are $\Sigma$-terms, then $o(t_1,\dots,t_\kappa)$ is a $\Sigma$-term.
		\end{itemize}
		If $f \colon X \to Y$ is a function, $T_{\Sigma,E}(f) \colon T_{\Sigma,E}(X) \to T_{\Sigma,E} (Y)$ sends a term $t$ in $t[f(x) / x]$, where every variable $x$ is substituted by its image $f(x)$. The unit $\eta^T$ is simply defined as $\eta^T_X (x) = x$, while the multiplication is defined by induction as:
		\[
		\begin{tikzcd}[row sep=0em]
			T_{\Sigma,E} \bigl( T_{\Sigma,E} (X) \bigr) \ar[r,"\mu^T_X"] & T_{\Sigma,E} (X) \\
			t \in T_{\Sigma,E} (X) \ar[r,|->] & t  \\
			o_\kappa(t_1,\dots,t_\kappa) \ar[r,|->] & o_\kappa (\mu^T_X(t_1),\dots,\mu^T_X(t_\kappa))
		\end{tikzcd}
		\]
		This construction is standard for \emph{finitary} algebraic theories, where every operation in $\Sigma$ has finite arity. The fact that it makes sense also for our case, where we have an operation for every cardinal, is ensured by the fact that our $(\Sigma,E)$ is \emph{tractable} in the sense of~\cite[Definition 1.5.44]{manes_algebraic_1976}, because we proved in Theorem~\ref{thm:pres} that it presents a monad on $\Set$, namely $\convpowS$. Tractability ensures that the \emph{class} of $\Sigma$-terms, once quotiented by $E$, is forced to be a set. 
		
		Now, the category of Eilenberg-Moore algebras for the monad $T_{\Sigma,E}$ is, in fact, isomorphic to $\Alg{(\Sigma,E)}$, hence also to $\EM{\convpowS}$, via a functor $F$ such that
		\[
		\begin{tikzcd}
			\EM{\convpowS} \ar[d,"\U {\convpowS}"'] \ar[r,"F"] &       \EM{T_{\Sigma,E}}  \ar[d,"\U{T}"] \\
			\Set \ar[r,"\id{\Set}"] & \Set
		\end{tikzcd}
		\]
		commutes. This generates an isomorphism of monads $\phi \colon T_{\Sigma,E} \to \convpowS$ where, for all sets $X$, $\phi_X = F(\mu^{\convpowS}_X) \circ T_{\Sigma,E} (\eta^{\convpowS}_X)$, thanks to the following general result:
		
		\begin{thm}
			Let $(S,\eta^S,\mu^S)$ and $(T,\eta^T,\mu^T)$ be monads on a category $\C$. Suppose $\EM S$ and $\EM T$ are isomorphic via a functor $F \colon \EM S \to \EM T$ such that $\U T F = \U S$. Then $T$ and $S$ are isomorphic as monads, that is, there is an isomorphism of monads $\phi \colon T \to S$. Specifically, $\phi_X$ is given as the unique morphism of $T$-algebras (hence, a morphism in $\C$) granted by the universal property of $(TX,\mu^T_X)$ as the free $T$-algebra on $X$, induced by $\eta^S_X$:
			\[
			\begin{tikzcd}
				TTX \ar[d,"\mu^T_X"'] \ar[r,"T(\phi_X)"] & TSX \ar[d,"F(\mu^S_X)"] \\
				TX \ar[r,dotted,"\exists ! \phi_X"] & SX \\
				X \ar[u,"\eta^T_X"] \ar[ur,"\eta^S_X"']
			\end{tikzcd}
			\qquad
			\phi_X = F(\mu^S_X) \circ T(\eta^S_X).
			\]
		\end{thm}
		
		Direct calculations show that our $\phi_X \colon T_{\Sigma,E}(X) \to \convpowS(X) = \ConvPow {\S X} {\muS_X}$ acts as follows:
		\begin{align*}
			\phi_X(x) &= \{ \Delta_x \} \quad \text{for $x \in X$} \\
			\phi_X(0) &= \{ 0 \colon X \to S\} \\
			\phi_X(t_1 + t_2) &= \phi_X(t_1) +^{\alpha_{\muS_X}} \phi_X(t_2) \\
			\phi_X (\lambda \cdot t) &= \lambda \cdot^{\alpha_{\muS_X}} \phi_X (t) \\
			\phi_X(\Sup_I \{t_i \mid i \in I \}) &= \convclos {\bigcup_{i \in I} \phi(t_i)} {\muS_X}
		\end{align*}
		where $\phi_X(t_1) +^{\alpha_{\muS_X}} \phi_X(t_2)$ is the result of adding up in the $\S$-algebra $(\ConvPow {\S X} {\muS_X}, \alpha_{\muS_X})$ (recall from~\eqref{eq:alphaa} the definition of the $\S$-algebra structure map $\alpha_a$ on $\ConvPow X a$ for a given $\S$-algebra $(X,a)$), seen as a $S$-semimodule, the convex subsets $\phi_X(t_1)$ and $\phi_X(t_2)$ of $\S X$, as explained in~\eqref{eq:sumP}. Similarly $\lambda \cdot^{\alpha_{\muS_X}} \phi_X (t)$ is the scalar-product in $\ConvPow {\S X} {\muS_X}$. Hence, for $\lambda \ne 0_S$:
		\begin{align*}
			\phi_X(t_1 + t_2) &= \Bigl\{ \muS_X \Bigl(  
			\begin{tikzcd}[row sep=0em,ampersand replacement=\&]
				f_1 \ar[r,|->] \& 1 \\
				f_2 \ar[r,|->] \& 1
			\end{tikzcd}
			\Bigr) \mid f_1 \in \phi_X (t_1), \, f_2 \in \phi_X (t_2) \Bigr\} \\
			\phi_X (\lambda \cdot t) &= \{ \lambda \cdot f \mid f \in \phi_X (t) \}.
		\end{align*}

		We now consider the algebraic theory $(\Sigma', E')$ obtained by restricting $(\Sigma,E)$ to finitary joins. More precisely, we fix
		\[
		\Sigma' = \Sigma_{\SL}\cup \Sigma_{\LSM} \qquad E'= E_{\SL} \cup E_{\LSM} \cup E_{\D'}
		\]
		where $(\Sigma_{\SL}, E_{\SL})$ is the algebraic theory for semilattices, $(\Sigma_{\LSM},E_{\LSM})$ is the one for $S$-semimodules, and $E_{\D'}$ is the set consisting of the four distributivity axioms illustrated in the bottom of Table \ref{tab:axiomsinitial}.
		It is clear, then, that if we restrict the action of $\phi_X$ to all those terms involving only finite suprema, we obtain a function translating $\Sigma'$-terms into convex subsets.
		\begin{prop}\label{prop:injectivemonadmap}
			Let $T_{\Sigma',E'}(X)$ be the set of $\Sigma'$-terms with variables in $X$ quotiented by $E'$. Let $\bb{\cdot}_X \colon T_{\Sigma',E'}(X) \to \convpowS(X)$ be the function defined as
			\[
			\begin{array}{rcl}
				\bb{x} &=& \{\Delta_x\} \text{ for }x\in X\\
				\bb{0} &=& \{0^{\muS_X}\} \\
				\bb{\bot} &=& \emptyset \\
			\end{array}\qquad
			\begin{array}{rcl}
				\bb{\lambda \cdot t} &=& \begin{cases}
					\{\lambda \cdot^{\muS_X} f \, \mid \; f \in \bb{t}\} & \text{if }\lambda \neq 0\\
					\{0^{\muS_X}\}  & \text{otherwise}
				\end{cases} \\
				\bb{t_1 + t_2} &=& \{f_1 +^{\muS_X} f_2 \, \mid \, f_1 \in \bb{t_1}, \; f_2 \in \bb{t_2} \} \\
				\bb{t_1 \sqcup t_2} &=& \convclos{\bb{t_1}\cup \bb{t_2}}{{\muS_X}} \\
			\end{array}
			\]
			Let $\bb{\cdot} \colon T_{\Sigma',E'} \to \convpowS$ be the family $\{\bb{\cdot}_X\}_{X \in |\Set|}$.
			Then $\bb{\cdot}\colon T_{\Sigma',E'} \to \convpowS$ is a map of monads and, moreover, each  $\bb{\cdot}_X \colon T_{\Sigma',E'}(X) \to \convpowS(X)$ is injective.
		\end{prop}
		\begin{proof}
			The obvious restriction $\theta_X \colon T_{\Sigma',E'}(X) \to T_{\Sigma,E}(X) $ induces a monad map $\theta \colon T_{\Sigma',E'} \to T_{\Sigma,E}$. Observing that $\bb{\cdot}$ is exactly $ \phi \circ \theta$ is enough to conclude that  $\bb{\cdot}$ is a monad map. Each $\bb{\cdot}_X \colon T_{\Sigma',E'}(X) \to \convpowS(X)$ is injective since $\theta_X $ is injective and $\phi_X$ is an isomorphism.
		\end{proof}

		We say that a set $\mathcal A \in \convpowS(X)$ is \emph{finitely generated} if there exists a finite set $\mathcal B \subseteq \mon{S}(X) $ such that $\convclos{\mathcal B}{}= \mathcal A$. 
		We write $\convpowfS(X)$ for the set of all $\mathcal A \in \convpowS(X)$ that are finitely generated.
		
		\begin{prop}\label{prop:convpowfS is a monad}
			The assignment $X \mapsto \convpowfS(X)$ gives rise to a monad $\convpowfS \colon \Set \to \Set$ where the action on functions, the unit and the multiplication are defined as for $\convpowS$.
		\end{prop}
		To prove Proposition~\ref{prop:convpowfS is a monad} we only have to show that: for all functions $f \colon X \to Y$, $\convpowS(f) \colon \convpowS(X) \to \convpowS(Y)$ restricts and corestricts to $\convpowfS (X) \to \convpowfS (Y)$; the unit $\eta^\convpowS_X$ can be corestricted to $\convpowfS (X)$, and finally the multiplication $\mu^\convpowS_X$ restricts and corestricts to $\convpowfS \convpowfS X \to \convpowfS X$. These facts are illustrated in Appendix \ref{app6}.

		Next we are going to use Proposition~\ref{prop:injectivemonadmap} to give a presentation to the monad $\convpowfS$.
		
		\begin{thm}\label{thm:convpowfS is presented by (Sigma', E')}
			The monads $T_{\Sigma',E'}$ and $\convpowfS$ are isomorphic. Therefore $(\Sigma',E')$ is a presentation for the monad $\convpowfS$.
		\end{thm}
		\begin{proof}
			We first observe that the function 
			\[
			\bb{\cdot}_X \colon T_{\Sigma',E'}(X) \to \convpowS(X)
			\]
			factors as 
			\[\begin{tikzcd}
				T_{\Sigma',E'}(X) \ar[r,"\bb{\cdot}'_X"] &  \convpowfS(X) \ar[r,"\iota_X"] & \convpowS(X)
			\end{tikzcd}\]
			where $\iota_X \colon \convpowfS(X) \to\convpowS(X)$ is the obvious set-inclusion. This can be easily checked by induction on $T_{\Sigma',E'}(X)$.
			
			Observe that, since $\bb{\cdot}_X$ is injective by Proposition \ref{prop:injectivemonadmap}, then also $\bb{\cdot}'_X$ is injective. We conclude by showing that it is also surjective.
			
			Let $\mathcal A \in \convpowfS(X)$. Since $\mathcal A$ is finitely generated there exists a finite set $\mathcal B \subseteq \mon{S}(X)$ such that $\convclos{\mathcal B}{}=\mathcal A$. If $\mathcal A=\emptyset$, then $\bb{\bot}'_X = \mathcal A$. If $\mathcal B=\{\phi_1, \dots, \phi_n\}$ with $\phi_i \in \mon{S}(X)$ then, for all $i$, we take the term 
			$$t_i=\phi_i(x_1)\cdot x_1 + \dots + \phi_i(x_m)\cdot x_m$$ 
			where $\{x_1, \dots, x_m\}$ is the support of $\phi_i$.
			It is easy to check that $\bb{t_i}_X'= \{\phi_i\}$. Then by the inductive definition of $\bb{\cdot}'_X$, one can easily verify that $\bb{t_1 \sqcup \dots \sqcup t_n}'_X = \convclos{\{\phi_1, \dots \phi_n\}}{}= \mathcal A$.
		\end{proof}
		
		\begin{exa}\label{ex:segments and polytopes}
			Recall  $\convpowS(1)$ for $S=\Rp$ from Example~\ref{ex:intervals}.
			By restricting to the finitely generated convex sets, one obtains  $\convpowfS(1) = \{\emptyset\} \cup \{[a,b] \,\mid \, a,b\in \Rp\} $, that is the sets of the form $[a,\infty)$ are not finitely generated. Table \ref{tab:defbb1} illustrates the morphism $\bb{\cdot}\colon T_{\Sigma',E'}(1) \to \convpowS(1)$. 
			It is worth observing that every closed interval $[a,b]$
			is denoted by a term in $T_{\Sigma',E'}(1)$ for $1=\{x\}$: indeed, $\bb{(a \cdot x) \sqcup (b\cdot x)}= [a,b]$.
			For $2=\{x,y\}$, $\convpowfS(2)$ is the set containing all convex polygons: for instance the term $(r_1\cdot x + s_1\cdot y)\sqcup (r_2\cdot x + s_2\cdot y) \sqcup (r_3\cdot x + s_3\cdot y)$ denote a triangle with vertexes $(r_i,s_i)$. For $n=\{x_0,\dots x_{n-1}\}$, it is easy to see that $\convpowfS(n)$ contains all convex $n$-polytopes.
		\end{exa}

		\begin{table}[t]
			\[
			\begin{array}{rcl}
				\bb{x} &=& [1,1]\\
				\bb{0} &=& [0,0] \\
				\bb{\bot} &=& \emptyset \\
			\end{array}\qquad
			\begin{array}{rcl}
				\bb{\lambda \cdot t} &=& \begin{cases}
					[\lambda \cdot a, \lambda \cdot b ] & \text{if }\lambda \neq 0 ,\; \bb{t}=[a,b]\\
					\emptyset & \text{if }\lambda \neq 0 ,\; \bb{t}=\emptyset\\
					[0,0]  & \text{otherwise}
				\end{cases} \\
				\bb{t_1 + t_2} &=& \begin{cases}
					[a_1+a_2, b_1+b_2] & \text{ if } \bb{t_i}=[a_i,b_i]\\
					\emptyset & \text{ otherwise}
				\end{cases}\\
				\bb{t_1 \sqcup t_2} &=& \begin{cases}
					[min \;a_i, \, max \; b_i] & \text{ if } \bb{t_i}=[a_i,b_i]\\
					[a_1,b_1] & \text{ if } \bb{t_1}=[a_1,b_1],\; \bb{t_2}=\emptyset \\
					[a_2,b_2] & \text{ if } \bb{t_2}=[a_2,b_2],\; \bb{t_1}=\emptyset \\
					\emptyset & \text{ otherwise}
				\end{cases}  \\
			\end{array}
			\]
			\caption{The inductive definition of the function $\bb{\cdot}_1 \colon T_{\Sigma',E'}(1) \to \convpowS (1)$ for $1=\{x\}$.}
			\label{tab:defbb1}
		\end{table}
		
		\paragraph{On combining $\S$ with $\Pf$ directly}
		One might wonder whether $\convpowfS$ arises from a weak distributive law of type $\S\Pf \to \Pf \S$. By Theorem~\ref{thm:correspondence weak distributive law with weak extensions-liftings}, to give a weak distributive law $\delta \colon \S\Pf \to \Pf \S$ is equivalent to give a weak lifting of $\Pf$ to $\EM \S$. This requires, at the very least, a functor $\tilde{\Pf} \colon \EM\S \to \EM\S$, which we can presume would compute,  for a  $\S$-algebra $(X,a)$, the set of its convex and finitely generated subsets, in an analogous way of $\tilde{\P}$ computing the set of convex subsets. Moreover, we would also need a natural transformation $\iota \colon \U\S \tilde{\Pf} \to \Pf \U\S$ which, given a $\S$-algebra $(X,a)$, must assign to each convex and finitely generated subset $A$ of $X$ a finite subset of $X$ itself. A natural choice would be the set of generators of $A$, 
		but how does one find them?

		In general, when we take the convex closure of $A \subset X$, with $X$ an $S$-semimodule, we are adding to $A$ first all the segments between the elements of $A$ (the generators), then all the segments between all the new points that we have just added, and so on. The generators are exactly the \emph{extreme points} of its closure $\convclos A {}$, to use the terminology of~\cite{grunbaum_convex_2003}. For an arbitrary $U \subseteq X$, the extreme points of $U$ are defined to be those points that are not properly inside any segment of $U$:
		\[
		\ext U = \{x \in U \mid \forall y,z \in U \ldotp \forall \lambda,\mu \in S \ldotp \left(
		\begin{cases}
			\lambda + \mu =1 \\
			x = \lambda y + \mu z
		\end{cases}\right) \implies y = z \}.
		\]
		If in a semimodule $X$ we have $A = \convclos B {}$ with $B$ finite, then it is not difficult to see that $\ext A \subseteq B$. Hence a natural candidate for $\iota_{(X,a)}(A)$, with $A$ convex and finitely generated, is $\ext A$. However, this $\iota$ fails to be natural, and we show here why. The naturality square of $\iota$, for $f \colon (X,a) \to (Y,b)$ in $\EM\S$, computes $f(\ext A)$ and $\ext {f(A)}$, which do not coincide if we take $S=\Rp$, $X=\S(\{x,y,z\})$, $Y = \S(\{u,v\})$, $f = \S \left( 
		x \mapsto u,\,
		y \mapsto u,\,
		z \mapsto v
		\right) $
		and finally 
		\[
		A = \convclos{ \left\{ \frac 1 2 x + \frac 1 2 y, \, \frac 1 2 x + \frac 1 2 z,\, z \right\}}{}.
		\]
		In this case, $\ext A = \left\{ \frac 1 2 x + \frac 1 2 y, \, \frac 1 2 x + \frac 1 2 z,\, z \right\}$, so we have
		\[
		f(\ext A) = \left\{ u,\, \frac 1 2 u + \frac 1 2 v,\, v \right\}.
		\]
		On the other hand, since $f$ is a morphism of $\S$-algebras, 
		\[
		f(A) = \convclos{f\left( \left\{ \frac 1 2 x + \frac 1 2 y, \, \frac 1 2 x + \frac 1 2 z,\, z \right\} \right)}{} = \convclos{\left\{ u,\, \frac 1 2 u + \frac 1 2 v,\, v \right\}}{} 
		\]

		hence
		\[
		\ext{f(A)} = \{u,v\} \neq f(\ext A).
		\]
		
		This is, of course, saying less than saying that ``there does not exist a weak lifting of $\Pf$ to $\EM\S$'', but it shows that a very natural candidate does not work. It remains an open problem whether a weak distributive law $\S\Pf \to \Pf\S$ exists.

		\section{Conclusions: Related and Future Work}
		Our work was inspired by~\cite{goy_combining_2020} where Goy and Petrisan compose the monads of powerset  and probability distributions by means of a weak distributive law in the sense of Garner~\cite{garner_vietoris_2020}. Our results also heavily rely on the work of Clementino et al.~\cite{clementino_monads_2014} that illustrates sufficient conditions on a semiring $S$ for the existence of a weak distributive law $\delta \colon \S \P \to \P \S$. However, to the best of our knowledge, the alternative characterisation of $\delta$ provided by Theorem~\ref{thm:delta for positive refinable semifields} was never shown.
		
		This characterisation is essential to give a handy description of the lifting $\tilde\P\colon \EM{\S} \to \EM{\S}$ (Theorem~\ref{thm:weaklift}) as well as to observe the strong relationships with the work of Jacobs (Remark~\ref{remarkJacobs}) and of Klin and Rot (Remark~\ref{remarkKlinRot}). The weak distributive law $\delta$ also plays a key role in providing the algebraic theories presenting the composed monad $\convpowS$ (Theorem~\ref{thm:weaklift}) and its finitary restriction $\convpowfS$ (Theorem \ref{thm:convpowfS is presented by (Sigma', E')}). These two theories resemble those appearing in, respectively,~\cite{goy_combining_2020} and~\cite{bonchi2019theory} where the monad of probability distributions plays the role of the monad $\S$ in our work.
		
		Theorem~\ref{thm:Kleisli is CPPO enriched} allows to reuse the framework of coalgebraic trace semantics~\cite{hasuo_generic_2006} for modelling over $\Kl{\convpowS}$ systems with both nondeterminism and quantitative features. The alternative framework based on coalgebras over $\EM{\convpowS}$ directly leads to \emph{nondeterministic weighted automata}. A proper comparison with those in~\cite{droste2009handbook} is left as future work. Thanks to the abstract results in~\cite{DBLP:conf/csl/BonchiPPR14}, language equivalence for such coalgebras could be checked by means of coinductive up-to techniques. It is worth remarking that, since $\delta$ is a weak distributive law, then thanks to the work in~\cite{DBLP:journals/corr/abs-2010-00811} up-to techniques are also sound for ``convex-bisimilarity'' (in coalgebraic terms, behavioural equivalence for the lifted functor $\tilde\P \colon \EM{\S} \to \EM{\S}$).

		We conclude by recalling that we have two main examples of positive semifields: $\Bool$ and $\mathbb{R}^+$. Booleans could lead to a coalgebraic modal logic and trace semantics for \emph{alternating automata} in the style of \cite{DBLP:conf/fossacs/KlinR15}.
		For $\mathbb{R}^+$, we hope that exploiting the ideas in \cite{DBLP:journals/tcs/Rutten05} our monad could shed some light on the behaviour of linear dynamical systems featuring some sort of nondeterminism.

		\section*{Acknowledgment}
		\noindent This work was supported by the Ministero dell'Universit\`a e della Ricerca of Italy under Grant No.\ 201784YSZ5, PRIN2017 -- ASPRA (\emph{Analysis of Program Analyses}) and partially supported by Universit\`a di Pisa Project PRA-2018-66 DECLWARE: Metodologie dichiarative per la progettazione e il deployment di applicazioni.
		
		
%
%

		\bibliographystyle{alphaurl}
		\bibliography{biblio}

		\appendix
		\section{Another weak distributive law of \texorpdfstring{$\pow$}{\emph{P}} over \texorpdfstring{$\pow$}{\emph{P}}}\label{app:another weak distributive law} In Remark~\ref{rem:Kleisli vs Barr extensions} we have seen a Kleisli extension of the functor $\pow$ to $\Rel$ given by $E \colon \Rel \to \Rel$, where for $R \subseteq X \times Y$ a relation, $E(R)$ is the function
		\[
		\begin{tikzcd}[row sep=0em]
			\pow(X) \ar[r,"E(R)"] & \pow(Y) \\
			A \ar[r,|->] & \{ y \in Y \mid \exists a \in A \ldotp a R y \}.
		\end{tikzcd}
		\] 
		Now we ask: do the unit $\eta$ and the multiplication $\mu$ of the powerset monad also extend to natural transformations $\tilde\eta \colon \id{\Rel} \to E$ and $\tilde\mu \colon EE \to E$?
		
		As observed in~\cite{hyland_category_2006}, if a natural transformation $\psi \colon T_1 \to T_2$ has an extension $\tilde \psi \colon \tilde T_1 \to \tilde T_2$ to $\Kl S$, then it is unique: Definition~\ref{def:liftings and extensions}.\ref{extension natural transformation to Kleisli} requires that the components of $\tilde\psi$ be exactly $\tilde\psi_X = F_S(\psi_X)$, with $F_S \colon \C \to \Kl S $ the free functor. Hence all that one has to do to see whether a $\psi$ has an extension to $\Kl S $ is to check if the family of morphisms $(F_S(\psi_X))_{X \in \C}$ is natural with respect to the functors $\tilde T_1, \tilde T_2$. (In other words, although the functors $T_1$ and $T_2$ may have many extensions to $\Kl S $, once we fix an extension we can only have one $\tilde\psi$ between them.) In our case then, we have
		\[
		\tilde\eta_X = \{ (x, \{x\}) \mid x \in X  \} \subseteq X \times \pow (X), \quad \tilde\mu_X = \{ (\mathcal A, \bigcup_{A \in \mathcal A} A  ) \mid \mathcal A \in \pow\pow (X) \} \subseteq \pow\pow (X) \times \pow(X)
		\]
		and the question is: are they natural transformations with respect to the functor $E$?
		
		The naturality condition for $\tilde\eta$ states that for any $R \subseteq X \times Y$ the following square commutes in $\Rel$:
		\[
		\begin{tikzcd}
			X \ar[r,"\tilde\eta_X"] \ar[d,"R"'] & \pow X \ar[d,"E(R)"] \\
			Y \ar[r,"\tilde\eta_Y"] & \pow Y
		\end{tikzcd}
		\]
		It is not difficult to see that the upper leg is
		\[
		\{ (x, \{ y \in Y \mid x R y \}) \mid x \in X  \}
		\]
		while the lower leg is
		\[
		\{ (x, \{ y \} ) \mid (x,y) \in R  \}
		\]
		and the two coincide if and only if $R$ is functional and total, that is if it is a function of sets. Hence $\eta$ does not extend to a natural transformation $\tilde\eta \colon \id\Rel \to E$.
		
		Regarding the multiplication, instead, we have that $\tilde\mu$ is natural as a transformation $EE \to E$ if and only if for all $R \subseteq X \times Y$ the following square commutes in $\Rel$:
		\[
		\begin{tikzcd}
			\pow\pow X \ar[r,"\tilde\mu_X"] \ar[d,"EE(R)"'] & \pow X \ar[d,"E(R)"] \\
			\pow \pow Y \ar[r,"\tilde\mu_Y"] & \pow Y
		\end{tikzcd}
		\]
		All the relations involved in the square above are actually functions, hence both legs are functions and, given $\mathcal A \in \pow\pow X$:
		\begin{gather*}
			E(R) \circ \tilde\mu_X (\mathcal A) = \{ y \in Y \mid \exists A \in \mathcal A \ldotp \exists x \in A \ldotp x R y \} \\
			\tilde\mu_Y \circ EE(R) = \bigcup_{A \in \mathcal A} \{ y \in Y \mid \exists x \in A \ldotp x R y \}
		\end{gather*}
		and the two clearly are the same. Therefore $\mu$ does indeed extend to a natural transformation $\tilde\mu \colon EE \to E$.
		
		This means the monad $\pow$ does not extend to $\Rel$ on a whole via the endofunctor $E$, but it does \emph{weakly} extend. Hence there is an alternative weak distributive law $\pow \pow \to \pow\pow$ induced by $E$ and $\tilde \mu$. To find it we use the proof of Theorem~\ref{thm:correspondence weak distributive law with weak extensions-liftings}: we need to compute $E(\ni_X)$, which by definition is
		\[
		\begin{tikzcd}[row sep=0em]
			\pow\pow X \ar[r,"E(\ni_X)"] & \pow X \\
			\mathcal A \ar[r,|->] & \{x \in X \mid \exists A \in \mathcal A \ldotp A \ni x \}.
		\end{tikzcd}
		\]
		This means that 
		$
		E(\ni_X)(\mathcal A) = \bigcup_{A \in \mathcal A} A = \mu_X(\mathcal A):
		$
		we have re-obtained $\mu$! Hence the weak distributive law $\delta \colon \pow\pow \to \pow\pow$ is 
		\[
		\delta_X (\mathcal A) = \{ \mu_X(\mathcal A) \}.
		\]
		In other words, $\delta$ is the horizontal composition of $\mu$ and $\eta$, as in the following situation:
		\[
		\begin{tikzcd}
			\Set \ar[r,bend left,"\pow\pow"{name=F}] \ar[r,bend right,swap,"\pow"{name=G}] & \Set \ar[r,bend left,"\id{}"{name=H}] \ar[r,bend right,swap,"\pow"{name=K}] & \Set.
			\arrow[Rightarrow,from=F,to=G,"\mu",shorten >=2pt, shorten <=2pt]
			\arrow[Rightarrow,from=H,to=K,"\eta",shorten >=2pt, shorten <=2pt]
		\end{tikzcd}
		\]
		Now, what is the induced weak lifting of $\pow$ to $\EM\pow$? Again the proof of Theorem~\ref{thm:correspondence weak distributive law with weak extensions-liftings} tells us how to find it. Given $(X,a \colon \pow X \to X)$ a $\pow$-algebra (which means that $X$ is a complete semilattice), the weak lifting $\tilde\pow$ computes a new $\pow$-algebra $\tilde\pow (X,a) = (Y,b)$ where $Y$ is obtained by splitting the idempotent
		\[
		\begin{tikzcd}[row sep=0em]
			\pow X \ar[r,"\eta_X"] & \pow\pow X \ar[r,"\delta_X"] & \pow \pow X \ar[r,"\pow a"] & \pow X \\
			A \ar[r,|->] & \{ A \} \ar[r,|->] & \{ \underbrace{\mu_X(\{ A \} )}_{A} \} \ar[r,|->] & \{ a(A)\}
		\end{tikzcd}
		\]
		whose fixed points are the singletons. Hence the carrier set of $\tilde\pow(X,a)$ is $\pow_1 X = \{ \{x\} \mid x \in X \}$. The two natural transformations $\pi$ and $\iota$ required by the definition of weak lifting are
		\[
		\begin{tikzcd}[row sep=0em]
			\pow X \ar[r,"{\pi_{(X,a)}}"] & \pow_1 X \\
			A \ar[r,|->] & \{ a(A) \}
		\end{tikzcd}
		\qquad \text{and} \qquad
		\begin{tikzcd}[row sep=0em]
			\pow_1 X \ar[r,"{\iota_{(X,a)}}"] & \pow X \\
			\{ x \} \ar[r,|->] & \{ x \}
		\end{tikzcd}
		\]
		while the $\pow$-algebra structure map $b \colon \pow \pow_1 X \to \pow_1 X$ is given by:
		\[
		\begin{tikzcd}[row sep=0em]
			\pow \pow_1 X \ar[r,"{\pow \iota_{(X,a)}}"] & \pow \pow X \ar[r,"\delta_X"] & \pow\pow X \ar[r,"\pow a"] & \pow X \ar[r,"{\pi_{(X,a)}}"] & \pow_1 X \\
			\mathcal A \ar[r,|->] & \mathcal A \ar[r,|->] & \{ \bigcup \mathcal A \} \ar[r,|->] & \{ a( \bigcup \mathcal A)\} \ar[r,|->] & \{ a( \bigcup \mathcal A)\}
		\end{tikzcd}
		\]
		In other words, it  takes $\mathcal A \subseteq \pow_1 X$ which, being a collection of singletons in $X$, is in bijective correspondence with $\bigcup \mathcal A \subseteq X$, and returns the singleton of $a(\bigcup \mathcal A)$. $\pow_1 X$ is therefore a $\pow$-algebra essentially in the same way $X$ is.  \section{Further Details of the Proof of Theorem~\ref{thm:weaklift}}\label{app4}
		
		\paragraph{Proof of Proposition~\ref{prop:image along a of convex closures}} Let $(X,a)$ be a $\mon S$-algebra and let $\mathcal A \subseteq \mon S X$. We need to prove that:
		\[
		\pow a \Bigl( \convclos{\mathcal A}{\muS_X} \Bigr) = \convclos{\pow a (\mathcal A)}{a}.
		\]
		We have
		\[
		\pow a \Bigl( \convclos{\mathcal A}{\muS_X} \Bigr) = \{ a(\muS_X(\Psi)) \mid \Psi \in \mon S^2 X, \sum_{\psi \in \mon S X} \Psi(\psi) = 1, \supp \Psi \subseteq \mathcal A \}
		\]
		while
		\[
		\convclos{\pow a (\mathcal A)}{a} = \{ a(\phi) \mid \phi \in \mon S X, \sum_{x \in X} \phi(x)=1, \supp \phi \subseteq \{a(\psi) \mid \psi \in \mathcal A\}   \}.
		\]
		For the left-to-right inclusion: let $\Psi \in \mon S^2 X$ be such that $\sum_{\psi \in \mon S X} \Psi(\psi) = 1$ and with $\supp \Psi \subseteq \mathcal A$. Define $\phi=\mon S (a) (\Psi)$. Then $x \in \supp \phi$ if and only if there exists $\psi \in \supp \Psi$ such that $x = a(\psi)$. Hence, if $x \in \supp \phi$, then $x=a(\psi)$ for some $\psi \in \mathcal A$. Moreover,
		\[
		\sum_{x \in X} \phi(x) = \sum_{x \in X} \sum_{\substack{\psi \in \supp \Psi \\ a(\psi)=x}} \Psi(\psi) = \sum_{\psi \in \mon S X} \Psi(\psi) = 1.
		\]
		Since $a(\muS_X(\Psi)) = a(\mon S (a)(\Psi))$ by the properties of $(X,a)$ as $\mon S$-algebra, and $\mon S (a) (\Psi) = \phi$ by definition, we conclude.
		
		Vice versa, given $\phi \in \mon S X$ such that $\sum_{x \in X} \phi(x)=1$ and $\supp \phi = \{a(\psi_1),\dots,a(\psi_n)\}$ where $\psi_1,\dots,\psi_n \in \mathcal A$, define
		\[
		\Psi = \left( 
		\begin{tikzcd}[row sep=0em]
			\psi_1 \ar[r,|->] & \phi(a(\psi_1)) \\
			\vdots & \vdots \\
			\psi_n \ar[r,|->] & \phi(a(\psi_n))
		\end{tikzcd}
		\right)
		\]
		Then $\sum_{\psi \in \mon S X} \Psi(\psi) = 1$ and
		\[
		a(\muS_X(\Psi)) = a(\mon S (a) (\Psi)) = a(\phi)
		\]
		because $\mon S (a) (\Psi) = \phi$.\qed
		
		\begin{prop}\label{thm:convpow(X,a) is a S-algebra}
			Let $S$ be a positive semifield, $(X,a)$ a $\mon S$-algebra, $\Phi \in \S\ConvPow X a$. Then $\pow a (\delta_X (\mon S (\iota) (\Phi))) = \pow a (\choice{\mon S (\iota) (\Phi)})$.
		\end{prop}
		\begin{proof}
			In this proof we shall simply write $\Phi$ instead of the more verbose ${\mon S (\iota) (\Phi)}$. We want to prove that 
			\begin{multline}\label{eqn:linear combination of convex subsets of X}
				\pow a \bigl(\delta_X(\Phi)\bigr) = \\
				\Bigl\{
				a(\psi) \mid \psi \in \mon S X \ldotp \exists u \colon \!\! \supp \Phi \to X \ldotp u(A) \!\in \! A ,\, \forall x \in X\ldotp \psi(x) = \!\!\sum_{\substack{A \in \supp \Phi \\ u(A)=x}} \!\!\!\Phi(A)
				\Bigr\}
			\end{multline}
			where we have, by Theorem~\ref{thm:delta for positive refinable semifields}, that
			\[
			\pow a \bigl(\delta_X( \Phi)\bigr) = \\
			\{
			a(\muS_X(\Psi)) \mid \Psi \in \mon S ^2 X, \sum_{\phi \in \mon S X} \Psi(\phi)=1, \supp \Psi \subseteq \choice{\Phi}
			\}.
			\]
			First of all, $\emptyset$ is \emph{not} a $\mon S$-algebra, because there is no map $\mon S (\emptyset) \to \emptyset$ given that $\mon S (\emptyset) = \{ \emptyset \colon \emptyset \to S\}$, hence $X \ne \emptyset$. Next, if $\Phi = \zero \colon \P X \to S$, namely the $0$-constant function, then $\choice{\Phi} = \{\zero \colon X \to S\}$  therefore one can easily see that the left-hand side of~\eqref{eqn:linear combination of convex subsets of X} is equal to $\{a(\zero \colon X \to S)  \}$. For the same reason, the right-hand side is also equal to $\{a(\zero \colon X \to S)\}$. Moreover, if $\Phi(\emptyset) \ne 0$, then there is no $u \colon \supp \Phi \to X$ such that $u(\emptyset) \in \emptyset$, 
			so $\choice{\Phi} = \emptyset$ and so is the left-hand side of~(\ref{eqn:linear combination of convex subsets of X}); for the same reason, also the right-hand side is empty. 
			
			Suppose then, for the rest of the proof, that $\Phi \ne 0$ and that $\Phi(\emptyset)=0$.
			
			For the right-to-left inclusion in~(\ref{eqn:linear combination of convex subsets of X}): given $\psi \in \choice\Phi$, consider $\Psi = \etaS_{\mon S X} (\psi) = \Delta_\psi \in \mon S ^2 X$. Then $\Psi$ clearly satisfies all the required properties and $\muS_X(\Psi)=\psi$.
			
			The left-to-right inclusion is more laborious. Let $\Psi \in \mon S ^2 X$ be such that $\sum_{\chi \in \S X} \Psi(\chi)=1$ and such that $\supp \Psi \subseteq \choice\Phi$, that is, for all $\phi \in \supp \Psi$ there is $u^\phi \colon \supp \Phi \to X$ such that $u^\phi(A) \in A$ for all $A \in \supp \Phi$ and 
			$\phi = \S(u^\phi)(\Phi)$. 
			We have to show that $a(\muS_X(\Psi))= a(\psi)$ for some $\psi \in \mon S X$ of the form $\sum_{A \in \supp \Phi} \Phi(A) \cdot u(A)$ for some choice function $u \colon \supp \Phi \to X$. Notice that the given $\Psi$ is a convex linear combination of functions $\phi$'s in $\mon S X$ like the one we have to produce: the trick will be to exploit the fact that each $A \in \supp \Phi$ is convex. We first give the idea of the proof, and then write the details of it.
			
			Suppose $\supp \Phi = \{A_1,\dots,A_n\}$ and $\supp \Psi = \{\phi^1,\dots,\phi^m\}$. Call $u^j$ the choice function that generates $\phi^j$. 
			Then $\Psi$ is of this form:
			\[
			\Psi= \Bigg(   
			\underbrace{\left(
				\parbox{2.9cm}{$u^1(A_1) \mapsto \Phi(A_1)$ \\ \vdots \\ $u^1(A_n) \mapsto \Phi(A_n)$}
				\right)}_{\phi^1} \mapsto \Psi(\phi^1), \,  \dots, \,  
			\underbrace{\left(
				\parbox{3cm}{$u^m(A_1) \mapsto \Phi(A_1)$ \\ \vdots \\ $u^m(A_n) \mapsto \Phi(A_n)$}
				\right)}_{\phi^m}
			\mapsto \Psi(\phi^m)
			\Bigg)
			\]
			Define the following element of $\S^2 X$:
			\[
			\Psi'= \Bigg(   
			\underbrace{\left(
				\parbox{3cm}{$u^1(A_1) \mapsto \Psi(\phi^1)$ \\ \vdots \\ $u^m(A_1) \mapsto \Psi(\phi^m)$}
				\right)}_{\chi^1}
			\mapsto \Phi(A_1), \, \dots, \,
			\underbrace{\left(
				\parbox{3.1cm}{$u^1(A_n) \mapsto \Psi(\phi^1)$ \\ \vdots \\ $u^m(A_n) \mapsto \Psi(\phi^m)$}
				\right)}_{\chi^n}
			\mapsto \Phi(A_n)
			\Bigg)
			\]
			Observe that $u^1(A_i), \dots, u^m(A_i) \in A_i$ by definition, and $A_i$ is convex by assumption: since $\sum_{j=1}^m \Psi(\phi^j)=1$, we have that $a(\chi^i) \in A_i$. Set then $u({A_i}) = a(\chi^i)$ and define $\psi= \S (a) (\Psi')$: we have $\psi \in \choice\Phi$ with $u$ as the generating choice function. We shall prove that $\muS_X(\Psi)=\muS_X(\Psi')$, thus obtaining
			\[
			a(\psi)=a\bigl(\mon S (a)(\Psi')  \bigr) = a\bigl(\muS_X (\Psi')\bigr) = a\bigl(\muS_X(\Psi)\bigr)
			\]
			as desired. We now set this down in detail.
			
			Let $A \in \supp \Phi$. Define for all $x \in X$:
			\[
			\chi_A(x) = \sum_{\substack{\phi \in \supp \Psi \\ x=u^\phi(A)}} \Psi(\phi).
			\]
			Then $\supp{\chi_A} = \{u^\phi(A) \mid \phi\in \supp \Psi\}$ is finite, hence $\chi_A \in \mon S X$. Let now 
			\[
			u(A) = a(\chi_A).
			\]
			We have that
			\[
			\sum_{x \in X} \chi_A(x) = \sum_{x \in X} \sum_{\substack{\phi \in \supp \Psi \\ x=u^\phi(A)}} \Psi(\phi) = \sum_{\phi \in \supp \Psi} \Psi(\phi)=1,
			\]
			hence $u(A) \in A$ because $A$ is convex. Define, for all $x \in X$:
			\[
			\psi(x) = \sum_{\substack{A \in \supp \phi \\ x=a(\chi_A)}} \Phi(A).
			\]
			Then $\supp \psi = \{ u(A) \mid A \in \supp \Phi\}$ so $\psi \in \mon S X$.
			To conclude, we have to prove that $a(\muS_X(\Psi)) = a(\psi)$. To that end, let for all $\chi \in \mon S X$
			\[
			\Psi'(\chi) = \sum_{\substack{ A \in \supp \phi \\ \chi_A = \chi}} \Phi(A).
			\]
			Then $\supp \Psi' = \{ \chi_A \mid A \in \supp \Phi\}$ is finite. Then we have for all $x \in X$ that
			\begin{align*}
				\muS_X(\Psi')(x) &= \sum_{\chi \in \supp \Psi'} \Psi'(\chi) \cdot \chi(x) \\
				&\overset{(\ast)}{=} \sum_{A \in \supp \Phi} \Phi(A) \cdot \sum_{\substack{\phi \in \supp \Psi \\ x=u^\phi(A)}} \Psi(\phi) \\
				&= \sum_{\phi \in \supp \Psi} \Psi(\phi) \cdot \sum_{\substack{A \in \supp \Phi \\ x=u^\phi(A)}} \Phi(A) \\
				&= \sum_{\phi \in \mon S X} \Psi(\phi) \cdot \phi(x) \\
				&= \muS_X(\Psi)(x)
			\end{align*}
			where equation $(\ast)$ is explained in a similar way than $(\ast_1)$ in the proof of Theorem~\ref{thm:delta for positive refinable semifields}. We also have that 
			\[
			\mon S (a)(\Psi')(x) = \sum_{\chi \in a^{-1}\{x\}} \Psi'(\chi) = \sum_{\substack{\chi \in \mon S X \\ a(\chi) = x}} \sum_{\substack{A \in \supp \Phi \\ \chi_A = \chi}} \Phi(A) = \sum_{\substack{A \in \supp \Phi \\ a(\chi_A)=x}} \Phi(A) = \psi(x).
			\]
			Therefore,
			\[
			a(\psi)=a\bigl(\mon S (a)(\Psi')  \bigr) = a\bigl(\muS_X (\Psi')\bigr) = a\bigl(\muS_X(\Psi)\bigr).\qedhere
			\]
		\end{proof}
		
		\begin{prop}\label{prop:rectangle lifting mu commutes}
			The following diagram commutes.
			\[
			\begin{tikzcd}
				\pow \pow X \ar[r,"\pow \convclos{(-)} a"] \ar[d,"\mu^{\P}_X"'] 
				& \pow(\ConvPow X a) \ar[r,"\convclos{(-)}{\alpha_a}"] & \ConvPow {(\ConvPow X a)} {\alpha_a}  \ar[d,"{\mu^{\tilde \P}_{(X,a)}}"] \\
				\pow X \ar[rr,"\convclos{(-)} a"] & & \ConvPow X a
			\end{tikzcd}
			\]	
		\end{prop}
		\begin{proof}
			In this proof we shall simply write $\overline A$ for $\convclos A a$ for any $A \subseteq X$, because there is no risk of confusion.
			We have to show that
			\[
			\{ a(\phi) \mid \phi \in \mon S X, \sum_{x \in X} \phi(x)=1, \supp \phi \subseteq \bigcup_{U \in \mathcal U} U \}
			=
			\bigcup_{\substack{\Phi \in \mon S \ConvPow X a \\
					\sum \Phi(C)=1 \\
					\supp \Phi \subseteq \{ \convclos U {} \mid U \in \mathcal U \}
			}}
			\alpha_a(\Phi)   .
			\]
			For the left-to-right inclusion: let $\phi \in \mon S X$ such that $\sum_{x \in X} \phi(x)=1$, and suppose $\supp \phi = \{x_1,\dots,x_n\}$. We have that for all $i \in \natset n$ there exists $A_i \in \mathcal U$ such that $x_i \in A_i$. While the $x_i$'s are distinct, the $A_i$'s need not be. Define, for all $B \in \ConvPow X a$,
			\[
			\Phi(B)=\sum_{ \substack{i \in \natset n \\ B=\overline{A_i} }} \phi(x_i).
			\]
			Then $\supp \Phi = \{ \overline {A_1},\dots,\overline{A_n} \}$ hence $\Phi \in \mon S \ConvPow X a$. One can then prove that $a(\phi) \in \alpha_a(\Phi)$ in the same way as showed in the proof of Proposition~\ref{prop:union of convex family of convex subsets is convex}.
			
			Vice versa, an element of the right-hand side is of the form $a(\psi)$ for some $\Phi \in \mon S \ConvPow X a$ such that $\sum_{B \in \ConvPow X a} \Phi(B) = 1$, $\supp \Phi = \{\overline {A_1},\dots,\overline{A_n}\}$ with $A_i \in \mathcal U$ for all $i \in \natset n$ and for some $\psi \in \mon S X$ and $u \in \prod_{i =1}^n \overline{A_i}$ such that 
			\[
			\psi(x)=\sum_{ \substack{i \in \natset n \\[.25em] x=u_i} } \Phi(\overline{A_i}).
			\]
			We want to prove that $a(\psi)= a(\phi)$ for an appropriate $\phi \in \mon S X$ such that $\sum_{x \in X} \phi(x) =1$ and $\supp \phi \subseteq \bigcup_{U \in \mathcal U} U$.
			
			Since $u_i \in \overline{A_i}$, we have that $u_i = a(\phi_i)$ for some $\phi_i \in \mon S X$ such that $\sum_{x \in X} \phi_i(x)=1$ and $\supp \phi_i \subseteq A_i$. These $\phi_i$'s are not necessarily distinct. Let then, for all $\phi \in \mon S X$,
			\[
			\Psi(\phi) = \sum_{ \substack{i \in \natset n \\ \phi=\phi_i} } \Phi(\overline{A_i}).
			\]
			Then $\supp \Psi = \{f_1,\dots,f_n\}$ so $\Psi \in \mon S ^2 X$ and it is easy to prove, by means of direct calculations, that $\psi=\mon S (a)(\Phi)$, $\sum_{x \in X} \muS_X(\Psi)=1$ and that $\supp{\muS_X(\Psi)} \subseteq \bigcup_{U \in \mathcal U} U$. Then we have that $a(\psi)=a(\mon S (a)(\Psi)) = a(\muS_X(\Psi))$, and $\muS_X(\Psi)$ is the $\phi$ we were looking for.
		\end{proof}
		
		\section{Proof of Proposition~\ref{prop:convpowfS is a monad}}\label{app6}
		
		In this appendix, we provide a proof of  Proposition~\ref{prop:convpowfS is a monad}. We need to show that:
		\begin{itemize}[align=left]
			\item[(a)] for all functions $f \colon X \to Y$, $\convpowS(f) \colon \convpowS(X) \to \convpowS(Y)$ restricts and corestricts to $\convpowfS (X) \to \convpowfS (Y)$; 
			\item[(b)] the unit $\eta^\convpowS_X$ can be corestricted to $\convpowfS (X)$ (trivial); 
			\item[(c)] the multiplication $\mu^\convpowS_X$ restricts and corestricts to $\convpowfS \convpowfS X \to \convpowfS X$. 
		\end{itemize}
		Hereafter, we shall simply write $\convclos {\mathcal B} {}$ in lieu of $\convclos {\mathcal B} {\muS_X}$, for some
		$\mathcal B \subseteq \S X$.

		\begin{prop}
			Let $f \colon X \to Y$, $\mathcal A \subseteq \S X$ such that $\mathcal A =  \convclos {\mathcal B} {} $ for some finite $\mathcal B \subseteq \mathcal A$. Then
			\[
			\{ \S f (\phi) \mid \phi \in \mathcal A  \} = \convclos { \{\S f (\psi) \mid \psi \in \mathcal B \}} {}.
			\]
		\end{prop}
		\begin{proof}
			For the left to right inclusion: let $\phi \in \mathcal A = \convclos {\mathcal B} {}$. Then there exists $\Psi \in \S^2 X$ such that $\sum_{\chi \in \S X} \Psi(\chi) =1$, $\supp \Psi \subseteq \mathcal B$, $\phi = \muS_X(\Psi)$. We have to prove that there is a $\Psi' \in \S^2 Y$ such that $\sum_{\chi \in \S Y} \Psi'(\chi) = 1$, $\supp \Psi' \subseteq \{\S f (\psi) \mid \psi \in \mathcal B \}$, $\muS_Y(\Psi')=\phi$.
			
			Now, because of the naturality of $\muS$, we have
			\[
			\S f (\phi) = \S f (\muS_X(\Psi)) = \muS_Y(\S^2 f (\Psi)).
			\]
			One can easily see that $\S^2 f (\Psi)$ works for our desired $\Psi'$.
			
			Vice versa, let $\Psi' \in \S^2 Y$ be such that $\sum_{\chi \in \S Y} \Psi'(\chi) = 1$ and $\supp \Psi' \subseteq \{\S f (\psi) \mid \psi \in \mathcal B \}$. We have to show that there is $\phi \in \mathcal A$ such that $\muS_Y (\Psi') = \S f (\phi)$.
			
			We have that $\Psi'$ is of the form
			\[
			\left(
			\begin{tikzcd}[row sep=0em]
				\S f (\psi_1) \ar[r,|->] & \Psi'(\S f (\psi_1)) \\
				\vdots & \vdots \\
				\S f (\psi_n) \ar[r,|->] & \Psi'(\S f (\psi_n)) 
			\end{tikzcd}
			\right)
			\]
			Then that means that $\Psi'=\S(\S f)(\Psi)$ where $\Psi$ is defined as
			\[
			\Psi = \left(
			\begin{tikzcd}[row sep=0em]
				\psi_1 \ar[r,|->] & \Psi(\S f (\psi_1)) \\
				\vdots & \vdots \\
				\psi_n \ar[r,|->] & \Psi(\S f (\psi_n))
			\end{tikzcd}
			\right) \in \S^2 X
			\]
			Then, again by naturality of $\muS$, we have
			\[
			\muS_Y(\Psi')=\muS_Y(\S^2 f (\Psi)) = \S f (\muS_X (\Psi))
			\]
			and $\phi=\muS_X (\Psi)$ is indeed in $\mathcal A$ because $\mathcal A = \convclos {\mathcal B} {}$.
		\end{proof}
		
		This tells us that $\convpowfS$ is an endofunctor on $\Set$. Next, $\eta^\convpowS_X (x) = \{ \Delta_x \}$ and $\{\Delta_x\}$ is obviously finitely generated, therefore $\eta^\convpowS$ corestricts to $\convpowfS$. How about $\mu^\convpowS$?
		
		Recall that $\mu^\convpowS_X \colon \convpowS \convpowS X \to \convpowS X$ is defined for every $\mathscr A$ convex subset of $\S \bigl( \convpow (\S X) \bigr)$ as
		\[
		\mu^{\convpowS}_X (\mathscr A) = \bigcup_{\Omega \in \mathscr A} \{ \muS_X (\Psi) \mid \Psi \in \choice \Omega \}
		\]
		where 
		\[
		\choice\Omega = \{ \Psi \in \S^2 X \mid \forall \mathcal A \in \supp \Omega \ldotp \exists u_{\mathcal A} \in \mathcal A \ldotp \forall \phi \in \S X \ldotp \Psi(\phi) = \sum_{\substack{\mathcal A \in \supp \Phi \\ \phi = u_{\mathcal A}}} \Omega(\mathcal A) \}
		\]
		We aim to prove that $\bigcup_{\Omega \in \mathscr A} \{ \muS_X (\Psi) \mid \Psi \in \choice \Omega \}$ is, in fact, finitely generated in the hypothesis that $\mathscr A \in \convpowfS \convpowfS X$. We will achieve this in three steps.
		\begin{itemize}
			\item \textbf{Step 1}: let $\mathscr B$ be a finite subset of $\mathscr A$ such that $\mathscr A = \convclos {\mathscr B} {}$. Then we prove that
			\[
			\bigcup_{\Omega \in \mathscr A} \{ \muS_X (\Psi) \mid \Psi \in \choice \Omega \} = 
			\convclos { \bigcup_{\Theta \in \mathscr B} \{ \muS_X (\Xi) \mid \Xi \in \choice \Theta \} } {\muS_X}
			\]
			showing therefore that we can reduce ourselves to a finite union.
			\item \textbf{Step 2}: we prove that each $\{ \muS_X (\Xi) \mid \Xi \in \choice \Theta \}$ as of Step 1 is convex and finitely generated.
			\item \textbf{Step 3}: we prove that the convex closure of a finite union of convex and finitely generated sets is in turn finitely generated.
		\end{itemize}
		The next three lemmas will perform each step.
		\begin{lem}
			Let $\mathscr A \in \convpowfS \convpowfS X$ and let $\mathscr B$ be a finite subset of $\mathscr A$ such that $\mathscr A = \convclos {\mathscr B} {}$. Then
			\[
			\bigcup_{\Omega \in \mathscr A} \{ \muS_X (\Psi) \mid \Psi \in \choice \Omega \} = 
			\convclos { \bigcup_{\Theta \in \mathscr B} \{ \muS_X (\Xi) \mid \Xi \in \choice \Theta \} } {\muS_X}.
			\]
		\end{lem}
		\begin{proof}
			Let $\Omega \in \mathscr A = \convclos{\mathscr B}{}$. We have that
			\[
			\Omega = \muS_{\convpowS X} \left(
			\begin{tikzcd}[row sep=0em]
				\Theta_1 \ar[r,|->] & \sigma_1 \\
				\vdots & \vdots \\
				\Theta_t \ar[r,|->] & \sigma_t
			\end{tikzcd}
			\right)
			\]
			where $\Theta_i \in \mathscr B$ and $\sum_{l=1}^t \sigma_l = 1$. Notice that $\supp \Omega = \bigcup_{l=1}^t \supp \Theta_l$. Now, if $\supp \Omega = \{ \mathcal A_1, \dots, \mathcal A_n \} $ say, we have that any $\Psi \in \choice \Omega$ is of the form
			\[
			\Psi =  \left(
			\begin{tikzcd}[row sep=0em]
				\phi_1 \ar[r,|->] & \Omega(\mathcal A_1)=\sum_{l=1}^t \sigma_l \cdot \Theta_l (\mathcal A_1) \\
				\vdots & \vdots  \\
				\phi_n \ar[r,|->] & \Omega(\mathcal A_n) = \sum_{l=1}^t \sigma_l \cdot \Theta_l (\mathcal A_n)
			\end{tikzcd}
			\right)
			\]
			where each $\phi_i \in \mathcal A_i \subseteq \S X$. This leads us to define, for each $l \in \natset t$, a function $\Xi_l \in \S^2 X$ as:
			\[
			\Xi_l = \left(
			\begin{tikzcd}[row sep=0em]
				\phi_1 \ar[r,|->] & \Theta_l (\mathcal A_1) \\
				\vdots & \vdots \\
				\phi_n \ar[r,|->] & \Theta_l (\mathcal A_n)
			\end{tikzcd}
			\right)
			\]
			Notice that, in fact, $\Xi_l \in \choice{\Theta_l}$. Indeed, fixed $l$, we have that $\supp \Theta_l \subseteq \supp \Omega$, so it can be the case that $\Theta_l(\mathcal A_i)=0$ for some $i \in \natset n$, in which case we have $\Xi_l(\phi_i) = 0$. Nonetheless, for each $\mathcal A_i$ in $\supp {\Theta_l}$, the function $\Xi_l$ chooses one of its elements and associates to it $\Theta_l (\mathscr A_i)$.
			
			Define then
			\[
			\Psi'=\left(
			\begin{tikzcd}[row sep=0em]
				\muS_X(\Xi_1) \ar[r,|->] & \sigma_1 \\
				\vdots & \vdots \\
				\muS_X(\Xi_t) \ar[r,|->] & \sigma_t
			\end{tikzcd}
			\right)
			\]
			Then clearly $\sum_{\chi \in \S X} \Psi'(\chi) = 1$, and 
			\[
			\supp \chi \subseteq \bigcup_{l=1}^t \{ \muS_X(\Xi) \mid \Xi \in \choice {\Theta_l} \} \subseteq \bigcup_{\Theta \in \mathscr B} \{ \muS_X (\Xi) \mid \Xi \in \choice \Theta \}.
			\]
			It is a matter of direct calculations to show that $\muS_X (\Psi) = \muS_X (\Psi')$. This proves the left-to-right inclusion. The vice versa is immediate to see, recalling that we know that the left-hand side is convex.
		\end{proof}
		
		Applying the following lemma for the $\S$-algebra $(\S X, \muS_X)$, we obtain Step 2.
		\begin{lem}
			Let $(X,a)$ in $\EM \S$. Then for all $\Phi \in \S \P_{cf}^a X$ we have that
			\[
			\{ a(\phi) \mid \phi \in \choice\Phi \} = \convclos { \{ a(\psi) \mid \psi \in \choice{\Phi'} \} } a
			\]
			where, if $\supp \Phi = \{ A_1,\dots,A_n \}$ say, and for all $i$ $A_i = \convclos {B_i} a$ with $B_i \subseteq A_i$ finite, then
			\[
			\Phi' = (B_1 \mapsto \Phi(A_1), \dots, B_n \mapsto \Phi(A_n)).
			\]
		\end{lem}
		\begin{proof}
			Let $\phi \in \choice \Phi$. Then we can write $\phi$ as
			\[
			\phi = (u_1 \mapsto \Phi(A_1), \dots, u_n \mapsto \Phi(A_n))
			\]
			where $u_i \in A_i$ for all $i$. Notice that it is possible for $u_i = u_j$ for $i \ne j$: in that case, the notation above implicitly says that $u_i \mapsto \Phi(A_i) + \Phi(A_j)$. Now, since $A_i = \convclos {B_i} a$, we have that $u_i = a (\chi_i)$ for some $\chi_i \in \S X$ such that $\sum \chi_i(x) = 1$ and $\supp \chi_i \subseteq B_i$. So:
			\begin{align*}
				\phi &= (a(\chi_1) \mapsto \Phi(A_1), \dots, a(\chi_n) \mapsto \Phi(A_n)) \\
				&= \S (a) \bigl( \chi_1 \mapsto \Phi(A_1), \dots, \chi_n \mapsto \Phi(A_n) \bigr)
			\end{align*}
			Call $\Psi = \bigl( \chi_1 \mapsto \Phi(A_1), \dots, \chi_n \mapsto \Phi(A_n) \bigr)$. Then we just said that $\phi = \S (a) (\Psi)$. We can write $\Psi$ more explicitly, by listing down the action of each $\chi_i$:
			\[
			\Psi = \left(
			\begin{tikzcd}[row sep=0em,column sep=1em]
				\chi_1=\left(
				\parbox{1.7cm}{$x^1_1 \mapsto \lambda^1_1$ \\ \vdots \\ $x^1_{s_1} \mapsto \lambda^1_{s_1}$}
				\right)
				\ar[r,|->] & \Phi(A_1), & \dots &
				\chi_n=\left(
				\parbox{1.8cm}{$x^n_1 \mapsto \lambda^n_1$ \\ \vdots \\ $x^n_{s_n} \mapsto \lambda^n_{s_n}$}
				\right)
				\ar[r,|->] & \Phi(A_n)
			\end{tikzcd}
			\right)
			\]
			where, for each $k=1,\dots,n$, $\sum_{j=1}^{s_k} \lambda^k_j = 1$ and for all $j=1,\dots,s_k$ we have $x^k_j \in B_k$.
			
			Now, define for all $w \in \prod_{k=1}^n \natset{s_k}$, a $n$-tuple whose $j$-th entry is a number between $1$ and $s_j$, the function
			\[
			\psi_w = \bigl( x^1_{w_1} \mapsto \Phi(A_1), \dots, x^n_{w_n} \mapsto \Phi(A_n) \bigr), \text{ \ie } \psi_w(x) = \sum_{\substack{i \in \natset n \\ x=x^i_{w_i}}} \Phi(A_i).
			\]
			Define also
			\[
			\Psi' = ( \psi_w \mapsto \prod_{k=1}^n \lambda^k_{w_k} )_{w \in \prod \natset {s_k}} \text{ \ie } \Psi'(\psi) = \sum_{\substack{w \in \prod \natset{s_k} \\ \psi=\psi_w }} \prod_{k=1}^n \lambda^k_{w_k}.
			\]
			Notice that, by Lemma~\ref{lemma:generalised distributivity}, we have that
			\[
			\sum_{w \in \prod \natset{s_k}} \prod_{k=1}^n \lambda^k_{w_k} = \prod_{k=1}^n \sum_{j=1}^{s_k} \lambda^k_j = \prod_{k=1}^n 1 = 1.
			\]
			This immediately implies that $\sum_{\psi \in \S X} \Psi'(\psi)=1$. Next, we show that $\muS_X (\Psi') = \muS_X(\Psi)$.
			\begin{align*}
				\muS_X(\Psi')(x) &= \sum_{\psi \in \S X} \Psi'(\psi) \cdot \psi(x) \\
				&= \sum_w \Bigl[ \prod_{k=1}^n \lambda^k_{w_k} \cdot \sum_{\substack{i \in \natset n \\ x=x^i_{w_i} }} \Phi(A_i)  \Bigr] \displaybreak[0] \\
				&= \sum_w \sum_{\substack{i \in \natset n \\ x=x^i_{w_i}}} \Bigl[ (\prod_{k=1}^n \lambda^k_{w_k} ) \cdot \Phi(A_i)  \Bigr] \displaybreak[0]\\
				&= \sum_{i \in \natset n} \sum_{\substack{ w \\ x=x^i_{w_i}}} \Bigl[ \Phi(A_i) \cdot \underbrace{ \lambda^i_{w_i}}_{=\chi_i(x^i_{w_i}) = \chi_i(x)} \cdot \prod_{\substack{k \in \natset n \\ k \ne i}} \lambda^k_{w_k} \Bigr]\displaybreak[0] \\
				&= \sum_{i=1}^n \Phi(A_i) \cdot \chi_i(x) \cdot \underbrace{\sum_{\substack{w \\ x=x^i_{w_i}}} \prod_{k \ne i} \lambda^k_{w_k}}_{=\prod_{k \ne i} \sum_{j=1}^{s_k} \lambda^k_j = 1} \displaybreak[0] \\
				&= \muS_X(\Psi)(x).
			\end{align*}
			Hence:
			\[
			a(\phi) = a \bigl( \S(a)(\Psi)  \bigr) = a(\muS_X(\Psi)) = a(\muS_X (\Psi')) = a (\S(a)(\Psi'))
			\]
			where $a(\S(a)(\Psi'))$ indeed belongs to $\convclos { \{ a(\psi) \mid \psi \in \choice{\Phi'} \} } a$ because 
			\[
			\supp{(\S(a)(\Psi'))} = \{ a(\psi_w) \mid w \in \prod_{k=1^n} \natset{s_k} \} \subseteq \{ a(\psi) \mid \psi \in \choice{\Phi'} \}
			\]
			and the sum of all its images is $\sum_w \prod_{k=1}^n \lambda^k_{w_k} = 1$. 
		\end{proof}
		
		Notice that if $\Phi' \in \S\P X$ is such that $\supp{\Phi'} \subseteq \P_f (X)$, then $\choice{\Phi'}$ is finite. Finally, the next lemma, when $A$ and $B$ are finite, proves Step 3.
		
		\begin{lem}
			Let $(X,a)$ be in $\EM \S$, $A,B \subseteq X$. Then $\convclos{\convclos A {} \cup \convclos B {} } {} = \convclos{A \cup B} {}$.
		\end{lem}
		\begin{proof}
			We have:
			\begin{align*}
				\convclos { \convclos A {} \cup \convclos B {} } {} &= \{ a(\phi) \mid \phi \in \S X, \, \sum_{x} \phi(x)=1, \supp \phi \subseteq \convclos A {} \cup \convclos B {} \} \\
				\convclos A {} \cup \convclos B {} &= \{ a(\psi) \mid \psi \in \S X,\, \sum_x \psi(x)=1,\, \supp \psi \subseteq A \text{ or } \supp \psi \subseteq B \} \\
				\convclos{A \cup B} {} &= \{ a(\chi) \mid \chi \in \S X,\, \sum_x \chi(x) =1,\, \supp \chi \subseteq A \cup B \}.
			\end{align*}
			Let then $x \in \convclos{\convclos A {} \cup \convclos B {} } {}$. Then
			\[
			x=a(\phi) = a \left(
			\begin{tikzcd}[row sep=0em]
				a(\psi_1) \ar[r,|->] & \phi(a(\psi_1)) \\
				\vdots \\
				a(\psi_n) \ar[r,|->] & \phi(a(\psi_n))
			\end{tikzcd}
			\right)
			= a(\S (a) (\Phi)) = a (\muS_X (\Phi))
			\]
			where $\Phi = (\psi_1 \mapsto \phi(a(\psi_1)),\dots, \psi_n \mapsto \phi(a(\psi_n))$. Calling $\chi=\muS_X(\Phi)$, one can easily check that $\sum_x \chi(x)=1$ and that $\supp \chi \subseteq A \cup B$, hence $\convclos{\convclos A {} \cup \convclos B {} } {} \subseteq \convclos{A \cup B} {}$. The other inclusion is obvious, given that $A \cup B \subseteq \convclos A {} \cup \convclos B {}$. 
		\end{proof}
		
		This concludes the proof of Proposition~\ref{prop:convpowfS is a monad}.  
	\end{document}